\documentclass[12pt]{amsart}
\usepackage{tikz-cd}
\usepackage{amsmath}
\usepackage{amssymb}
\usepackage{graphicx} 
\usepackage{caption} 
\usepackage{float} 
\usepackage{amsmath, amstext, amsbsy, amssymb, amscd}
\usepackage{amsxtra}
\usepackage{amscd}
\usepackage{amsthm}
\usepackage{amsfonts}
\usepackage{mathrsfs}
\usepackage{eucal}
\usepackage{color}
\usepackage{tikz-cd}
\usepackage[CJKbookmarks=true]{hyperref}
\usepackage[all,cmtip]{xy}
\usepackage{supertabular}
\usepackage{tabularx}

\newtheorem{theorem}{Theorem}[section]
\newtheorem{lemma}[theorem]{Lemma}

\newtheorem{proposition}[theorem]{Proposition}

\theoremstyle{definition}

\newtheorem{definition}[theorem]{Definition}
\newtheorem{example}[theorem]{Example}
\newtheorem{remark}[theorem]{Remark}
\newtheorem{assumption}[theorem]{Assumption}

\newtheorem{notation}[theorem]{Notation}

\numberwithin{equation}{section}

\def\diag{\mathrm{diag}}

\def\Im{\mathrm{Im}\,}

\newcommand{\abs}[1]{\lvert#1\rvert}

\makeatletter

\newcommand{\Rmnum}[1]{\expandafter\@slowromancap\romannumeral #1@}

\newcommand{\Sone}{\mathrm{S}^1}

\newcommand{\Uthree}{\mathbb{U}_3}
\newcommand{\imagunit}{\mathbf{i}}
\newcommand{\SU}{\mathrm{SU}}
\newcommand{\su}{\mathfrak{su}}

\newcommand{\SUtwo}{\SU(2)}
\newcommand{\SUthree}{\SU(3)}

\newcommand{\rme}{\mathrm{e}}

\newcommand{\Hamiltonian}[1]{\mathcal{H}_{#1}}

\makeatother

{\vskip-\lastskip\medskip
	\noindent
	{\em #1.}\enspace
}%
{\qed\par\medskip
}

\newcommand{\RR}{\mathbb{R}}
\newcommand{\CC}{\mathbb{C}}

\newcommand{\gstar}{\mathfrak{g}^*}
\newcommand{\Oxi}{\mathcal{O}_{\xi}}
\newcommand{\galgebra}{\mathfrak{g}}
\newcommand{\CinfM}{C^\infty(M)}
\newcommand{\CinfcalM}{C^\infty(\mathcal{M})}
\newcommand{\Cinf}{C^\infty}
\newcommand{\Zxi}{\mathcal{Z}_{\xi}}
\newcommand{\Reducedxi}{\mathcal{R}_\xi}

\newcommand{\SET}[1]{\{#1\}}
\newcommand{\PoissonBracket}[1]{\{#1\}}

\newcommand{\Pxi}{\mathcal{P}_{\xi}}

\newcommand{\Bxi}{\mathcal{B}_{\xi}}
\newcommand{\Adjointaction}{\mathrm{Ad}}
\newcommand{\Lieg}{\mathfrak{g}}
\newcommand{\Hreg}{\Cartan_{\mathrm{reg}}}
\newcommand{\hreg}{\mathfrak{h}_{\mathrm{reg}}}

\newcommand{\Cartan}{\mathbf{H}}
\newcommand{\cartan}{\mathfrak{h}}

\newcommand{\Oprimexi}{\mathcal{O}^{\prime}_{\xi}}
\newcommand{\maxpart}{\mathrm{max}}

\newcommand{\lowpart}{\mathrm{low}}
\newcommand{\middleK}{\mathcal{K}}

\newcommand{\Projection}{\mathrm{Pr}}

\newcommand{\partialless}{\preccurlyeq}
\newcommand{\Rxizero}{\mathcal{R}_{0}}
\newcommand{\Pxizero}{\mathcal{P}_{0}}
\newcommand{\Bxizero}{\mathcal{B}_{0}}

\begin{document}
	\begin{abstract}We define the notion of superintegrability of a Hamiltonian system on a  stratified symplectic space. We focus on spin Calogero-Moser-Sutherland (sCMS) systems, where the phase space is a stratified symplectic space obtained by the Hamiltonian reduction of the cotangent bundle over a compact Lie group, and demonstrate that the sCMS systems for $\SU(3)$ are superintegrable.



\end{abstract}
	
	\title[]{Superintegrability of stratified symplectic spaces}
	
	
	\author{Zhuo Chen}
	\address{Department of Mathematics, Tsinghua University, Beijing 100084, China}
	\email{\href{mailto:chenzhuo@tsinghua.edu.cn}{chenzhuo@tsinghua.edu.cn}}
	
	\author{Kai Jiang$^*$} 
	\address{Paris Curie Engineer School, Beijing University of Chemical Technology, Beijing 100029, China}
	\email{\href{kai.jiang@buct.edu.cn}{kai.jiang@buct.edu.cn}(corresponding author)}
	
	\author{Nicolai Reshetikhin} 
	\address{YMSC, Tsinghua University, Beijing 100084, China; BIMSA, Beijing, China, \& Department of Mathematics and Computer Sciences, St. Petersburg University, Russian Federation}
	\email{\href{mailto:reshetik@math.berkeley.edu}{reshetik@math.berkeley.edu}}
	
	\author{Husileng Xiao}
	\address{College of Mathematical Sciences, Harbin Engineering University, Harbin, 150001, China}
	\email{\href{hslxiao@hrbeu.edu.cn}{hslxiao@hrbeu.edu.cn}}
	
	\date{First version, \today}
	\subjclass{  54E20, 70H06, 37J39 } 
	\keywords{Spin Calogero-Moser-Sutherland system, Superintegrability, Stratified space, Hamiltonian reduction}%
	
	\maketitle
	\tableofcontents
	
	\section{Introduction}\label{Sec:introduction}

\subsection{An overview of superintegrability}   Liouville integrability in Hamiltonian mechanics has many equivalent formulations. In physics, a dynamical Hamiltonian system is said to be completely integrable if it has the maximal number of Poisson-commuting integrals of motion. We will distinguish here the geometric notion of an integrable system on a symplectic manifold from the dynamical notion of integrability of a given Hamiltonian system. 

Geometrically, a completely integrable system (or Liouville integrable system)   is a Lagrangian fibration on a symplectic manifold that is the phase space of the system. For a symplectic manifold $M^{2n}$, a Lagrangian fibration is the projection 
\begin{equation}\label{Lint}
M^{2n}\stackrel{\psi}\rightarrow B^{n}
\end{equation}
where $B^{n}$ is a manifold, or a smooth stratified space, and $\psi$ is a Lagrangian Poisson surjective submersion, i.e. it is a Poisson surjective submersion where the generic fiber is a Lagrangian submanifold. Note that $B^{n}$ may also have subsets on which fibers of $\psi$ are singular or lower-dimensional; see, for example,  \cite{Zung1996,Zung2003,BF,MGRarxiv2506.16610v2}. In many examples such a Lagrangian fibration is given by level surfaces of $n$ Poisson-commuting independent functions. In this case $B^{n}$ is a subset in $\RR^n$ of all possible values of these functions.

Each Lagrangian fiber of $\psi$ is diffeomorphic to $\RR^k\times \mathbb{T}^{n-k}$ for some $k$. If a Lagrangian fiber is compact, it is diffeomorphic to an $n$-dimensional torus. 

A Hamiltonian vector field generated by a Hamiltonian function $H$ is completely integrable 
(Liouville integrable) if there is a completely integrable system (\ref{Lint}) such that $H$ is the pull-back of a function on $B^{n}$. The structure of Liouville integrable Hamiltonian dynamics is described by the Arnold-Liouville theorem; see, for example, \cite{Arnold}.

A superintegrable system, as a geometric structure on a symplectic manifold, is defined by a sequence of Poisson projections:
\begin{equation}\label{sint}
M^{2n}\stackrel{\psi}\rightarrow P^{2n-k}\stackrel{\pi}\rightarrow B^{k}
\end{equation}
where $k<n$, $P^{2n-k}$ is a Poisson manifold, $B^{k}$ is a manifold with trivial Poisson structure,
mappings $\pi$ and $\psi$ are Poisson surjective submersions, and connected components of $\pi^{-1}(b) $ for generic $b\in B^{k}$ are symplectic leaves. Projections $\psi$ and $\pi$, as in \eqref{sint}  should be Poisson surjective submersions. When $k=n$, then $P^{n}=B^{n}$ and $\pi$ becomes the identity map, and this definition agrees with the Liouville integrable system \eqref{Lint}. A generic fiber of $\psi$ in a superintegrable system is isotropic of dimension $k$.

In physics, a superintegrable system is a dynamical system that possesses the maximal possible number of integrals of motion, which are not necessarily all mutually Poisson commuting.  In this setting, these Poisson-commuting integrals 
play the same role as   in Liouville integrability. The number $k$ of such independent integrals can be less than $n$ on a  symplectic manifold of dimension $2n$. When $k=n$ we are back to the Liouville case. All integrals of motion form a Poisson subalgebra in the algebra of functions on the phase space which is the Poisson centralizer of the subalgebra of Poisson-commuting integrals.  One can say that the Poisson subalgebra of integrals of motions describes the internal symmetry of the dynamics.

For superintegrable systems there is a version of the Arnold-Liouville theorem. It was formulated and proved by Nekhoroshev \cite{Nek1972}. Examples of superintegrable systems were known even before this formal definition was introduced \cite{Nek1972}. A well-known example is the Kepler system in classical mechanics; see  Example \ref{Ex:Kepler}. 
For other examples of superintegrable systems, see,  for example, \cite{Old-review,BF}.

The notion of superintegrability is broader than  Liouville integrability in the sense that for Liouville integrable systems generic Liouville tori are $n$-dimensional and for a superintegrable system (\ref{sint}) they are $k$-dimensional. This means  superintegrability brings more restrictions to dynamics than   Liouville integrability.

In this paper we extend the notion of superintegrability to stratified symplectic spaces. In such a system the phase space $M^{2n}$ is not necessarily a manifold, but a  stratified symplectic space, $P^{2n-k}$ can be a stratified Poisson space and $B^{k}$  is a stratified space. Mappings $\pi$ and $\psi$ are morphisms of stratified Poisson spaces. For details, see Section \ref{Sec:Superintegrability}.

\subsection{Spin Calogero-Moser-Sutherland (sCMS)   systems}\label{Sec:SpinCM}
Many stratified Poisson spaces arise as quotients of Poisson manifolds under Poisson actions of Lie groups, while many stratified symplectic spaces are obtained through Hamiltonian reduction of symplectic manifolds or symplectic stratified spaces. A representative example is the Hamiltonian reduction of \(T^*G\), where \(G\) is a compact connected Lie group, with respect to the coadjoint action of \(G\); see below and Section \ref{sCMS}. The resulting reduced space carries a natural stratified symplectic structure, serving as the phase space of the sCMS model. For \(G=\SU(2)\) and \(G=\SU(3)\), these models provide explicit, nontrivial examples of superintegrable systems on stratified symplectic spaces.  We now  sketch the main ingredients of sCMS systems, which constitute a primary source of inspiration for this paper.

\subsubsection{The  space $\Reducedxi$ and the sCMS geometric structure}\label{subSec:phasespaceRxifromG}
Let $G$ be a compact and connected Lie group. 
For an element $g\in G$ the right translation by $g^{-1}$ identifies $T_gG$ with $T_eG$, i.e. with the space $\gstar$ dual to the Lie algebra $\Lieg=T_eG$. This gives the trivialization of the cotangent bundle by right translations $T^*G\cong \gstar\times G$. 

The action of $G$ on itself by conjugation lifts to a proper $G$-action on $T^*G \cong \gstar \times G$:
\[
g\colon\quad (x,\gamma)  \mapsto (\Adjointaction^*_g (x),  g \gamma g^{-1}).
\]
where $x\in \gstar$ and $\gamma\in G$. We will call this action the {\it coadjoint action}. This action is Hamiltonian with the moment map $J: T^*G \to \mathfrak{g}^*$ which is defined by 
\begin{equation}\label{Eqt:JTstarGgeneral}
		J(x, \gamma)=x-{\Adjointaction^*_{\gamma^{-1}}}(x).
\end{equation}
Note that $J$ is Poisson and $G$-equivariant.

For simple Lie groups this mapping $J$ is indeed surjective. As a demonstration, we prove this fact for $G=\SU(n)$ in 
Proposition \ref{Prop:SUnmomentmapsurjective} of Appendix \ref{Sec:JTstarGsurjectiveSUn}.

Choose $\xi$ in the image of the moment map $J$ and let $\mathcal{O}_\xi\subset \gstar$ be the coadjoint orbit through $\xi$. Define the space 
\begin{eqnarray*}
		&& \Reducedxi :=	J^{-1}(\mathcal{O}_\xi)/G\\
		&=&\{(x,\gamma)\in \mathfrak{g}^*\times G \mid x-{\Adjointaction^*_{\gamma^{-1}}}(x)\in\mathcal{O}_\xi \}/G. \quad  (\subset(\mathfrak{g} ^*\times G) /G ~ \cong T^*G/G) 
\end{eqnarray*}
The space $\Reducedxi$ is the phase space of the sCMS system corresponding to the coadjoint orbit $\mathcal{O}_\xi$ 
 (known as the Hamiltonian reduction of $T^*G$ with respect to  $\xi$). In general, $\Reducedxi$ is not a manifold but a stratified symplectic space.  For special $\xi$  it  may become a genuine manifold. This space can also be viewed as a stratified symplectic leaf of the stratified Poisson space $T^*G/G$.

\subsubsection{The sCMS system on $\Reducedxi$} 

First, let us set notation. For a \emph{compact simple} Lie group $G$, let $G_\CC$ be its complexification. For example, we can take $G=\SU(n)$ and $G_\CC=\mathrm{SL}(n,\CC)$. A choice of Borel subgroup $B\subset G_\CC$  fixes the Cartan subgroup $H_\CC\subset G_\CC$ corresponding to $B$. Let $\Cartan\subset G$ be the corresponding Cartan subgroup in $G$ (the intersection of $H_\CC$ and $G$ in $G_\CC$). For $G=\SU(n)$ the Cartan subgroup consists of unitary unimodular diagonal matrices. 

Denote by $\Hreg $ the open dense subset of \textbf{regular elements} in $\Cartan$ , i.e., elements with the centralizer being the Cartan subgroup. For $\SU(n)$, elements of $\Hreg $ are diagonal unitary unimodular matrices with $n$ distinct eigenvalues.
	
Let $\cartan$   be the Lie algebra of $\Cartan$, and let $j$ be the inclusion map $\cartan \hookrightarrow \Lieg$. The coadjoint orbit $\Oxi$ is a symplectic leaf of $\gstar$. The coadjoint action of the Cartan subgroup $\Cartan$ on $\Oxi$ is Hamiltonian with the moment map $P: \Oxi \hookrightarrow \Lieg^* \stackrel{j^*}{\rightarrow} \cartan^*$. Denote by $\Oprimexi$ the result of the Hamiltonian reduction of $\Oxi$ at the point $0 \in \cartan^*$ for the coadjoint action by $\Cartan$, i.e. $\Oprimexi = P^{-1}(0) /\Cartan$.

The coadjoint action of the normalizer $N(\Hreg )\subset G$  on $\Oxi$ descends to the action of the Weyl group $\mathbf{W}\simeq N(\Hreg )/\Cartan$ on $\Oprimexi$. 

Denote by  $\Reducedxi'$ the open dense subset of $\Reducedxi$, which consists of $G$-orbits through $(x,\gamma)$ where $\gamma\in \Hreg $. The diagonalization of $\gamma$ gives an isomorphism
$$
\Reducedxi' \cong (T^* \Hreg  \times \Oprimexi)/\mathbf{W} \cong (\Hreg \times \cartan^* \times \Oprimexi)/\mathbf{W}.
$$
Here the Weyl group $\mathbf{W}$ acts diagonally on the product $\Hreg \times \cartan^* \times \Oprimexi$.

We refer to \(\Reducedxi\) as the full phase space of the sCMS model, and to its open dense subset \(\Reducedxi'\) as the small phase space. To reveal finer structures within both \(\Reducedxi\) and \(\Reducedxi'\), the present paper decomposes \(\Reducedxi\) into smaller strata that organize its elements. In particular, we obtain a complete classification of the points of \(\Reducedxi\) for \(G=\SU(2)\) and \(G=\SU(3)\), for various types of \(\xi\) (see Sections \ref{Sec:SU(2)} and \ref{Sec:sCMSRPBSU3}). These classification results also identify the corresponding subset \(\Reducedxi' \subset \Reducedxi\) in these examples.

Briefly,   an sCMS   system as a geometric structure is the following combination of Poisson maps:
\begin{equation}\label{sCMS-int}
\Reducedxi\stackrel{\psi}\rightarrow \mathcal{P}_{\xi}\stackrel{\pi}\rightarrow \mathcal{B}_\xi.
\end{equation}
Here $\Reducedxi$ is the full phase space as described above,
\[
\mathcal{P}_{\xi}=\{(x,y)\in \mathfrak{g^*}\times \mathfrak{g^*}|\mathcal{O}_x=-\mathcal{O}_y; x+y\in \mathcal{O}_\xi\}/G\subset (\mathfrak{g^*}\times \mathfrak{g^*})/G,
\]
where $\mathcal{O}_x$ is the coadjoint orbit through $x$, $G$ acts diagonally on $\mathfrak{g^*}\times \mathfrak{g^*}$, and 
\[
\mathcal{B}_\xi=\{x\in \mathfrak{g}^*| \mbox{ there exists }\gamma\in G, x-\Adjointaction^*_{\gamma^{-1}}(x)\in \Oxi\}/G\subset  \mathfrak{g}^*/G.
\]
In Section \ref{sCMS}, we establish refined stratifications of the spaces $\Reducedxi$ and $\Pxi$ in \eqref{sCMS-int},  which are essential for defining a stratified superintegrable system when $G=\SU(2)$ or $\SU(3)$. These results naturally suggest that the sCMS system associated with Sequence \eqref{sCMS-int} may be superintegrable for a broad class of choices of $G$. We leave the investigation of this subject for future research.
 
\subsubsection{The sCMS system as an integrable Hamiltonian dynamical system}

A particular Casimir function on $\mathfrak{g^*}$ is $\frac{1}{2}(x,x)$,  where $(\cdot,\cdot)$ is the dual Killing form. Being pulled back from $\mathfrak{g^*}$ to $T^*G\simeq \mathfrak{g^*}\times G$, it defines a $G$-invariant function on $T^*G$, i.e., a function on $T^*G/G$. 
The Hamiltonian  $H_{\mathcal{O}_\xi}$ of the associated sCMS  system is defined to be the restriction of this function to the full phase space $\Reducedxi=J^{-1}(\mathcal{O}_\xi)/G$. 

Next, we demonstrate that the restriction of   $H_{\mathcal{O}_\xi}$ to the small phase space $\Reducedxi'$ can be computed explicitly as a function on $(\Hreg \times \cartan^* \times \Oprimexi)/\mathbf{W}$.


Let   \(\Delta \subset \cartan^*\)  be the root system of the Lie algebra \(\Lieg_\CC\). Let $\Delta_+$ be the subset of positive roots. The complexification of the Lie algebra $\Lieg$ decomposes as $\mathfrak{g}_\CC=\cartan\oplus \bigoplus_{\alpha\in \Delta_+} (\CC e_\alpha\oplus  \CC e_{-\alpha})$ where $\CC e_\alpha$ are root subspaces. The Killing form defines the dual decomposition. In other words, every $x\in \mathfrak{g}_\CC^*$ decomposes as $p+\sum_{\alpha\in \Delta_+} (x_\alpha e^*_\alpha+ x_{-\alpha}e^*_{-\alpha})$ where $p\in \cartan^*$,  $e^*_{\alpha}$ form a dual root basis, and $x_\alpha$, $x_{-\alpha}\in \CC$. In the compact real form $\mathfrak{g}$ and $\cartan$ are  real vector spaces and $x_{-\alpha}$ is complex conjugate to $x_\alpha$. 


Let \(\gamma_{\alpha}\) be functions on \(\Cartan\) defined by its action on root subspaces in $\mathfrak{g}$:
	\[
	\Adjointaction_{h}(e_{\alpha}) = \gamma_{\alpha}(h) e_{\alpha}.
	\]
For the compact real form we have $\overline{\gamma_{\alpha}}=\gamma_{\alpha}^{-1}=\gamma_{-\alpha}$.
The restriction of the Hamiltonian $H_{\mathcal{O}_\xi}$ to $\Reducedxi'$ takes the form (see, for example, \cite{NR2003}):
\[
H_{\mathcal{O}_\xi}|_{\Reducedxi'}= \frac{1}{2}(p, p) + \sum_{\alpha \in \Delta_{+}} \frac{(\alpha, \alpha) |x_\alpha|^2}{\left(\gamma_{\alpha / 2}-\gamma_{\alpha / 2}^{-1}\right)^2},
\]
where \(\left(\gamma_{\alpha / 2} - \gamma_{\alpha / 2}^{-1}\right)^2 := \gamma_{\alpha} + \gamma_{-\alpha} - 2\). The function $|x_\alpha|^2$ on $P^{-1}(0)$ is $\Cartan$-invariant and therefore is a function on $\mathcal{O}_\xi'=P^{-1}(0)/\Cartan$.

 In general, the Poisson-commuting Hamiltonians of the sCMS system arise from \(G\)-invariant polynomial functions on \(\mathfrak{g}^*\), which are pulled back to \(\mathfrak{g}^*\times G\) and subsequently restricted to \(J^{-1}(\mathcal{O}_\xi)\). In the present paper, we show only that these Hamiltonians preserve the refined strata (see Proposition \ref{Prop:tangenttorefinedstrata}); studying their behavior and the corresponding flows is left for future work.

\subsection{Structure of the paper} 
In Section \ref{Sec:stratifiedpoissonsymplectic} we give an overview of stratified spaces with symplectic or Poisson structures. Stratified superintegrable systems are introduced in Section \ref{Sec:Superintegrability}. Stratified superintegrable structure for an sCMS system is described in Section \ref{sCMS}. Its detailed analysis for $G=\SU(3)$ is given in Section \ref{Sec:sCMSRPBSU3}. In Appendix \ref{Apd:stratifiedspace} we recall basic facts about stratified spaces. Appendix \ref{Sec:JTstarGsurjectiveSUn} contains the proof of surjectivity of the moment map $T^*G\to \mathfrak{g}^*$ for $G=\SU(n)$ and the adjoint action of $G$ on $T^*G$. In Appendix \ref{appendixSec:Horn} we recall Weyl inequalities and Horn's theorem. 



\subsection{Acknowledgments} We are grateful to A. Liashyk, Z. Liu, G. Ma,  I. Sechin, and P. Xu for discussions. ~ H. Xiao is   supported by the Natural Science Foundation of Heilongjiang Province, China (Grant
No.YQ2023A007). The work
of  ~N.~Reshetikhin  was supported by the Changjiang fund, by the Collaboration Grant ``Categorical Symmetries'' from
the Simons Foundation, and by the project 075-15-2024-631 funded
by the Ministry of Science and Higher Education of the Russian Federation. Research of Z. Chen was partially supported by the Natural Science Foundation of China (Grant
No.12071241).
  

 \section{Stratified spaces with Poisson and symplectic  structures}\label{Sec:stratifiedpoissonsymplectic}  
For foundational aspects of stratified spaces, see the references \cite{Pflaumbook} and \cite{Pflaumpaper}. In this paper, we   follow the conventions and notions of \cite{SLannals}. Recall that a stratified space is defined as a pair $(\mathcal{M}, A)$, where $\mathcal{M}$ is the disjoint union of its smooth strata $\mathcal{M}_{\alpha}$ indexed by elements $\alpha$ of the index set $A$ (which is finite and partially ordered). A review of essential concepts pertaining to stratified spaces is provided in Appendix \ref{Apd:stratifiedspace}.

	\subsection{Orbit space of a $G$-manifold}

	Let $G$ be a Lie group, and $H_1$, $H_2\subset G$ be subgroups. We write $H_1 \sim H_2$ if and only if   $H_1$ and $H_2$ are conjugate via elements in $G$. For a subgroup $H$ of $G$, we denote by $(H)$ the conjugacy class of $H$. For simplicity,  we write $H$ instead of $(H)$ in many situations. Conjugacy classes of subgroups are partially ordered. In this partial order,   $(H) \partialless (H')$ if and only if $H'$ is conjugate to a proper subgroup of $H$.

	Let $M$ be a smooth manifold with the \emph{left} action of a Lie group $G$. Throughout this paper, all Lie group actions on manifolds are assumed to be \textit{proper}\footnote{Recall the definition of proper Lie group actions. Let $\sigma$: $G\times M\to M $
		be the action of $G$ on $M$. It   
		is said to be proper if the $\sigma$-preimage of any compact set in $M$ is compact in $G\times M$.}. We denote the action of an element $g \in G$ by $\sigma_g : M \to M$. Let $\galgebra := T_eG$  be the Lie algebra of $G$ whose Lie bracket is determined by associating each element $X \in \galgebra$ with the corresponding \textit{right-invariant} vector field  on $G$ (denoted by $\overrightarrow{X}$).

	The orbit space $M/G$ admits a stratified space structure, as elaborated in \cite{Pflaumpaper}. Recall some basic facts about this space. For any $m \in M$, the  \textbf{stabilizer}  (or the isotropy group) of $m$    is defined as $G_m = \{g \in G \mid \sigma_g m = m\}$. Given a closed subgroup $H \subset G$, we define the subspaces $M_H$ and $M_{(H)}$ of $M$ as follows:
	$$
	M_{H} := \{m \in M \mid G_m = H\}
	$$
	and
	$$
	M_{(H)} := \{m \in M \mid G_m \sim H\},
	$$
	where $G_m \sim H$, as above, indicates that $G_m$ is conjugate to $H$. The closed subgroups $H \subset G$ for which $M_{(H)}$ is non-empty are called  \textit{orbit types}. Note that conjugate subgroups represent the same orbit type. Connected components of both $M_H$ and $M_{(H)}$ are closed submanifolds of $M$.
	
The coset space $M/G$ is a stratified space with each stratum of $M/G$ consisting of a \textit{connected component} of $M_{(H)}/G$ for some orbit type $H$. For  simplicity, we still say that  $M_{(H)}/G$ is a stratum of orbit type $H$ (cf. \cite[Remark 1.3]{SLannals}) and denote it by $(M/G)_{(H)} $. In other words, for any point $[m]$ in the stratum $(M/G)_{(H)}$, its stabilizer $G_m$ is conjugate to $H$. For each stratum $(M/G)_{(H)}$, its dimension is given by $\dim M_{(H)}-\dim G+\dim H$, though this value may vary across different connected components of $M_{(H)}$.

	To establish the smooth structure of the stratified space $M/G$, it is essential to define the sheaf $\Cinf$ of smooth functions on $M/G$. This can be accomplished through a two-step process. First, consider an arbitrary open set $U$ in $M/G$. Then, for this open set $U$, define the space of sections $\Cinf(U)$ as the collection of $G$-invariant smooth functions on the preimage $\pi^{-1}(U)$ where $\pi$ is the canonical quotient map $\pi: M \to M/G$, i.e., 	
	$$\Cinf(U)=\Cinf(\pi^{-1}U)^G:=\SET{f\in \Cinf (\pi^{-1}U)~|~ \sigma_g^* f=f,\forall g\in G}.$$
	For more details about this structure sheaf on the orbit space $M/G$, see \cite{Pflaumbook}.
	Moreover, this definition extends to the orbit space of {a stratified $G$-space}. Here    a \textbf{stratified $G$-space} is defined as a stratified space $\mathcal{M}=\bigsqcup_{\alpha\in{A}} \mathcal{M}_{\alpha}$ equipped with a continuous and proper $G$-action. In addition, the $G$-action should preserve strata, i.e.,  for any element $g\in G$, the corresponding $g$-action $\sigma_g:~\mathcal{M}\to \mathcal{M}$ maps each stratum $\mathcal{M}_{\alpha}$ to itself.
	
	Under these conditions, the resulting orbit space $\mathcal{M}/G$ inherits a stratified structure.
	 Specifically,  the stratification of $\mathcal{M}/G$ is characterized by two parameters: the strata of the original space $\mathcal{M}$ and their associated orbit types. In other words, each stratum of $\mathcal{M}/G$ is uniquely determined by a pair $(\alpha,(H))$:$$(\mathcal{M}/G)_{\alpha,(H)}:=(\mathcal{M}_{\alpha})_{(H)}/G   ~~(\subset \mathcal{M}_{\alpha}/G), $$ where $\alpha\in A$ indexes a stratum of $\mathcal{M}$, and $H$ is an orbit type within $\mathcal{M}_{\alpha}$ when it is considered as a smooth manifold with $G$-action.  
	
	The smooth functions on $\mathcal{M}/G$ can be naturally identified with $G$-invariant smooth functions on $\mathcal{M}$, denoted by $\CinfcalM^G$.

	\subsection{Stratified   Poisson and symplectic spaces}  
	
	Stratified symplectic spaces have been  studied in symplectic and Poisson geometry since the 1990s \cite{BatesLerman1997, ChenRuan2001, Huebschmann2004, SLannals, Pflaumbook}. These structures naturally arise in geometry and mathematical physics applications;  see, for example,  \cite{Huebschmann2004, HuebschmannRudolphSchmidt2009, Karshon1992}.

	\begin{definition}\label{def:poissonsymlecticstratifiedspace}
\cite{Pflaumpaper,anotherpaper}
		A \textbf{stratified Poisson space} $(\mathcal{M},\Pi)$ consists of a stratified space $ \mathcal{M} $ equipped with a Poisson bivector $\Pi$, which is defined as a smooth section $\Pi:~\mathcal{M}\to T\mathcal{M}\wedge T\mathcal{M}$, such that for each stratum $\mathcal{M}_{\alpha}$   of $\mathcal{M}$, the restriction $\Pi|_{\mathcal{M}_{\alpha}}$ yields a smooth Poisson bivector on $\mathcal{M}_{\alpha}$. This Poisson bivector $\Pi$ naturally induces a Poisson bracket $\PoissonBracket{-,-}^{\Pi}$ on the space of smooth functions $\Cinf(\mathcal{M})$, which we denote simply as $\PoissonBracket{-,-}$ when the Poisson structure is understood. 
		
	\end{definition}
	To illustrate these concepts, let us examine some   examples. 
	
	\begin{example}\label{Example:demo}
		Consider a plane $\mathcal{M} =\RR^2$ equipped with the bivector $$\Pi= x \frac{\partial }{\partial x}\wedge \frac{\partial }{\partial y}.$$ When we partition $\mathcal{M} $ into three distinct regions (the left half-plane $(\mathcal{M})^2_l $, the right half-plane $(\mathcal{M})^2_r $, and the $y$-axis $(\mathcal{M})^1_c $), the resulting object is a stratified Poisson space $(\mathcal{M} ,\Pi)$, though   it is not a symplectic space.

	\end{example}
	\begin{example}
		\label{Example:demo2} 
		Consider the three-dimensional Euclidean space $\mathcal{N}=\RR^3$ equipped with the bivector field $$\Sigma= (x \frac{\partial }{\partial x}-z \frac{\partial }{\partial z})\wedge \frac{\partial }{\partial y}.$$ This space admits a natural stratification $\mathcal{N}=\mathcal{M}\times \RR$, where $\mathcal{M}$ is from the previous example and $\RR$ is the $z$-coordinate axis. The space $(\mathcal{N},\Sigma)$ is another example of a stratified Poisson space.
	
	\end{example}

The following notions are generalized from \cite{Pflaumpaper,anotherpaper}.

\begin{definition}\label{def:presymlecticstratifiedspace}
		A \textbf{stratified presymplectic space} $(\mathcal{M},\omega)$ consists of a stratified space $ \mathcal{M} $ equipped with a presymplectic $2$-form $\omega$,  namely a smooth section $\omega:~\mathcal{M}\to T^*\mathcal{M}\wedge T^*\mathcal{M}$, such that for each stratum $\mathcal{M}_{\alpha}$   of $\mathcal{M}$, the restriction $\omega|_{\mathcal{M}_{\alpha}}$ is a smooth presymplectic structure on $\mathcal{M}_{\alpha}$.   When the restriction $\omega|_{\mathcal{M}_{\alpha}}$ is   non-degenerate on every   stratum $\mathcal{M}_{\alpha}$,   we   refer to $(\mathcal{M},\omega)$ as a   \textbf{stratified symplectic} space.
		
	\end{definition}

\begin{example}Consider the stratified space $\mathcal{M} =(\mathbb{R}^{\geqslant 0})^3$ with coordinate functions $x$, $y$, and $z$.  The form $\omega=z dx\wedge dy+ x dy\wedge dz + y dz\wedge dx$ defines a presymplectic structure on $\mathcal{M} $. 
\end{example}	

See Example \ref{Ex:R4} for an example  of a    stratified presymplectic space, and Example \ref{Ex:R6} for a   stratified symplectic space.

	Stratified Poisson geometry is a natural extension of the usual Poisson geometry. In particular, we have stratified analogs of a Poisson map and of symplectic leaves. 
	
	\begin{definition}

		Consider two stratified Poisson spaces $(\mathcal{M},\Pi)$ and $(\mathcal{N},\Sigma)$. A smooth map $\phi:~ \mathcal{M} \to  \mathcal{N}$ is  a  \textbf{Poisson map}  if it satisfies the condition $$\PoissonBracket{\phi^*f_1,\phi^*f_2}^{\Pi}=\phi^*\PoissonBracket{f_1,f_2}^{\Sigma}$$ for all $f_1,f_2\in \Cinf(\mathcal{N})$. Furthermore, when $\phi$ is a  morphism   of stratified spaces, it is called a  morphism  of stratified Poisson spaces.

	\end{definition}

	\begin{definition} 
		A \textbf{single symplectic leaf} of a stratified Poisson space $(\mathcal{M},\Pi)$ is defined as a path-connected, immersed submanifold $F $ that lies within some stratum $\mathcal{M}_{\alpha}$ of $\mathcal{M}$, where $F $ constitutes a classical symplectic leaf of the Poisson manifold $(\mathcal{M}_{\alpha},\Pi|_{\mathcal{M}_{\alpha}} )$.
		
If a non-intersecting union of single symplectic leaves ${F_i}$ for ${i\in I}$ constitutes a path-connected stratified subspace in $\mathcal{M}$, and the index set $I$  cannot be enlarged, then we say that $\sqcup_{i\in I} F_i$ is a \textbf{full symplectic leaf} of the stratified Poisson space $(\mathcal{M},\Pi)$.

	\end{definition}
	
	\begin{example}\label{Example:K}    Let $\galgebra$ be a finite-dimensional Lie algebra, which is  the Lie algebra of a compact and connected  Lie group $G$. 
Denote the coadjoint orbit through $x\in \galgebra^*$ by $\mathcal{O}_x\subset \galgebra^*$. Define the space $\middleK \subseteq \gstar \times \gstar$ as the set of pairs $(x, y)$ such that $y$ belongs to the coadjoint orbit $\mathcal{O}_{-x}$. One verifies directly that  $\middleK$ is a stratified Poisson space and its single  symplectic leaves  and full symplectic leaves are both of the form $\mathcal{O}_{x} \times \mathcal{O}_{-x}$ for $x \in \gstar$. 		
		 \end{example}

\begin{example}Let us examine the stratified Poisson space $(\mathcal{M} , \Pi)$ in Example \ref{Example:demo}. The left stratum $(\mathcal{M})^2_l $ is itself a single symplectic leaf, and so is $(\mathcal{M})^2_r $. Any point in the $y$-axis is a single symplectic leaf in $ (\mathcal{M})^1_c  $. The only full symplectic leaf is the whole space $\mathcal{M} $. 
\end{example}
 	 
	\subsection{Stratified symplectic spaces derived from symplectic reductions}\label{Sec:SLmethod}
	Here we recall the description of the natural stratification on the reduced space of a Hamiltonian $G$-space. We follow  \cite{SLannals}; more details can be found in \cite{BatesLerman1997,Pflaumpaper}.

Let $(M,\omega)$ be a symplectic manifold equipped with a proper action of a Lie group $G$. Assume that the symplectic form $\omega$ is $G$-invariant.

	Let $\galgebra$ be the Lie algebra of $G$. A moment map $J:~M\to \gstar$ is characterized by two key properties: it is $G$-equivariant, and for any $X\in \galgebra$, the vector field $X_{M} \in \mathfrak{X}(M)$ generated by the $G$-action on $M$ is identical to the Hamiltonian vector field associated with $J^X\in \Cinf(M)$ which is defined by $J^X(m)=\langle J(m),X \rangle$ for all $m \in M$. The triple $(M,\omega,J)$ comprising the manifold, symplectic form, and moment map is called a \textbf{Hamiltonian $G$-space}.

	Let $\Oxi$ be the coadjoint orbit through $\xi\in \gstar$. When $\Oxi$ is locally closed in $\gstar$ (which is guaranteed for cases where $G$ is either reductive or a semi-direct product of a reductive group with a vector space; see Bates and Lerman \cite{BatesLerman1997}), and $\xi$ lies in the image of the moment map $J$, we can construct two important spaces from the quintuple $(G,M,\omega,J,\xi)$. The first is the subset $\Zxi$, defined as the preimage of the coadjoint orbit under $J$:
	\begin{equation}
		\label{Eqt:Zxi}
		\Zxi:=J^{-1}(\Oxi)\subset M
	\end{equation}
	The second is the  \textbf{reduced space}  $\Reducedxi$, obtained by taking the quotient of $\Zxi$ by the group action of $G$:
	\begin{equation}
		\label{Eqt:Sxi}
		\Reducedxi:=\Zxi/G
	\end{equation}
	
	The following   theorem was established in \cite{SLannals};  see also   \cite{BatesLerman1997,Pflaumbook}.
	\begin{theorem}[Sjamaar-Lerman]\label{Thm:SLReduced}
		Let $(M,G,\omega,J,\xi)$ be   as above. The   set $\Zxi$ and its corresponding reduced space $\Reducedxi$ have canonical stratifications determined by orbit types:
		
		\begin{itemize}
			\item[(1)] For any given orbit type $H$, the intersection $\Zxi\cap M_{(H)}$ is a submanifold of $M$. Furthermore, the corresponding orbit space, denoted by $(\Reducedxi)_{(H)}:=(\Zxi\cap M_{(H)})/G$, forms a manifold.
			
			\item[(2)] The decomposition of $\Zxi$ into the manifolds $\Zxi\cap M_{(H)}$, along with the partition of $\Reducedxi$ into the manifolds $(\Reducedxi)_{(H)}$, respectively establish stratifications of $\Zxi$ and $\Reducedxi$.
		\end{itemize}

		For the smooth structures on $\Zxi$ and $\Reducedxi$,  the following holds:
		
		\begin{itemize}
			\item[(3)] The smooth structure on $\Zxi$ is inherited from the  ambient manifold $M$ via its sheaf of smooth functions. Specifically, for any small open subset $U\subset \Zxi$, we have
			$$\Cinf_{\Zxi}(U)= \Cinf (U') / Z(U',U),$$
			where $U'\subset M$ denotes an open neighborhood of $U$, and $Z(U',U)$ represents the ideal of smooth functions on $U'$ that vanish on $U$.
			
			\item[(4)] For the quotient space $\Reducedxi$, with canonical quotient map $\pi:\Zxi\to \Reducedxi$, the sheaf of smooth functions is characterized by $G$-invariant functions on $\Zxi$. Precisely, for any open subset $V\subset \Reducedxi$,
			$$\Cinf_{\Reducedxi}(V)=\Cinf_{\Zxi}(\pi^{-1}V)^G:=\SET{f\in \Cinf_{\Zxi}(\pi^{-1}V)~|~ \sigma_g^* f=f,\forall g\in G}.$$
		\end{itemize}

		Furthermore, the reduced space $\Reducedxi$ is a stratified symplectic space:
		
		\begin{itemize}
			\item[(5)] Each stratum $(\Reducedxi)_{(H)}$ is endowed with a canonical symplectic structure $\omega_{(H)}$, the pullback of which to $\Zxi\cap M_{(H)}$ is identical to the restriction of $\omega$ to $\Zxi\cap M_{(H)}$.
			
			\item[(6)] For any pair of functions $f,f'\in \Cinf(\Reducedxi)$ and their corresponding $G$-invariant functions $\bar{f},\bar{f'}\in \CinfM$ satisfying $\pi^*f=\bar{f}|_{\Zxi}$ and $\pi^*f'=\bar{f'}|_{\Zxi}$, the following relation holds:
			$$
			\pi^*\PoissonBracket{f,f'}=\PoissonBracket{\bar{f},\bar{f'}}|_{\Zxi}
			.$$
			
			\item[(7)] Under the canonical embedding $\Reducedxi\hookrightarrow M/G$, each connected component of the stratum $(\Reducedxi)_{(H)}$ maps to a symplectic leaf within the Poisson manifold $M_{(H)}/G$. The Poisson structure of $M_{(H)}/G$ derives naturally from the symplectic structure of $M$.
		\end{itemize}
		
	\end{theorem}

\subsection{Orbit space of a stratified Poisson $G$-space}
In the above Theorem \ref{Thm:SLReduced},  $M$ is a symplectic manifold.   Fernandes, Ortega, and   Ratiu generalized this theorem to the realm of  Poisson manifolds \cite{FernandesOrtegaRatu}.  
In this part, we extend it to the case when $M$ is \emph{a stratified Poisson space}. In other words, we will start with a \textit{stratified Poisson $G$-space} and do the reduction following the same procedure as in Theorem \ref{Thm:SLReduced}.

\begin{definition}\label{Defn:stratifiedPoissonGspace}A \textbf{stratified Poisson $G$-space} is   a   stratified $G$-space $\mathcal{M}$ together with a
		stratified $G$-invariant Poisson structure $\Pi$ on $ \mathcal{M}  $.  \end{definition}

 Hence for every stratum $\mathcal{M}_{\alpha}$ of $\mathcal{M}$, the Poisson manifold $(\mathcal{M}_{\alpha},\Pi|_{\mathcal{M}_{\alpha}})$ together with the $G$-action on it forms a Poisson $G$-space.

\begin{definition}\label{Defn:momentmapofstratifiedPoissonGspace}Let $\galgebra$ be the Lie algebra of $G$. A smooth map $J: \mathcal{M} \to \gstar$ is called a  \textbf{moment map}  if it satisfies two conditions: First, $J$ is $G$-equivariant. Second, for any $X \in \galgebra$, consider the function $J^X \in \Cinf(\mathcal{M})$ defined by $J^X(m) = \langle J(m), X \rangle$ for all $m \in \mathcal{M}$; then the vector field $X_\mathcal{M} \in \mathfrak{X}(\mathcal{M})$ induced by the $G$-action on $\mathcal{M}$ must coincide with the Hamiltonian vector field of $J^X$. Equivalently, when restricted to any stratum $\mathcal{M}_{\alpha}$ of $\mathcal{M}$, the map $J|_{\mathcal{M}_{\alpha}}: \mathcal{M}_{\alpha} \to \gstar$ serves as a moment map for the Poisson $G$-space $(\mathcal{M}_{\alpha}, \Pi|_{\mathcal{M}_{\alpha}})$.
\end{definition}
	
Let $(\mathcal{M}, \Pi)$ be a stratified Poisson $G$-space as   in Definition \ref{Defn:stratifiedPoissonGspace}. The $G$-invariance of the Poisson bivector $\Pi$ implies that the  Poisson bracket $\PoissonBracket{-,-}$ is $G$-equivariant:
	$$\PoissonBracket{\sigma_g^*f_1,\sigma_g^*f_2}=\sigma_g^*\PoissonBracket{f_1,f_2}$$
	for all $f_1,f_2\in \Cinf(\mathcal{M})$ and $g\in G$.
	As a consequence, the space of $G$-invariant functions $\Cinf(\mathcal{M})^G$ forms a Poisson subalgebra of $\CinfcalM$. Given that $\Cinf(\mathcal{M}/G)=\Cinf(\mathcal{M})^G$,   the space of smooth functions on $\mathcal{M}/G$ is endowed with a Poisson algebra structure. However, this   does not guarantee that $\mathcal{M}/G$ inherits the structure of a stratified Poisson space in the sense of Definition \ref{def:poissonsymlecticstratifiedspace}. Specifically, additional conditions are required to ensure the existence of a bivector field $\Pi$ on $\mathcal{M}/G$. The following lemma  addresses this issue. For completeness, we provide a sketch of its proof.

	\begin{lemma}\label{lemma:tangentproperty} Let $(\mathcal{M} ,\Pi)$ be a stratified Poisson $G$-space. Suppose that there exists a  moment map $J:~\mathcal{M}\to \gstar$. Then for any $f\in \CinfcalM^G$  the following holds:
		\begin{itemize}
			\item[(1)]The Hamiltonian vector field $\Hamiltonian{f}$  is tangent to   $J^{-1}(\eta)$   for all $ \eta\in \Im J$;
			\item[(2)]If $G$ is a connected Lie group, then   $\Hamiltonian{f}$  is tangent to subspaces   $(\mathcal{M}_{\alpha})_H$, for all strata $\mathcal{M}_{\alpha}$ of $\mathcal{M}$ and orbit types $H$.
		\end{itemize}
	\end{lemma}
	
	\begin{proof} It suffices to consider the flow $\varphi^t$ of $\Hamiltonian{f}$ on  an arbitrary stratum $\mathcal{M}_{\alpha}$ of $\mathcal{M}$. For any $m\in \mathcal{M}_{\alpha}$, small enough $t$, and $X\in \galgebra$, we have
		\begin{eqnarray*}
			\frac{d}{dt}\langle J(\varphi^t(m)),X\rangle &=&  \frac{d}{dt}  J^X  (\varphi^t(m))
			= \Hamiltonian{f}(J^X)|_{\varphi^t(m)}=\PoissonBracket{f,J^X}|_{\varphi^t(m)}\\
			&=& - X_\mathcal{M}(f)|_{\varphi^t(m)}=0.
		\end{eqnarray*}
		The last step is due to $f$ being $G$-invariant.
		Hence $J(\varphi^t(m))\in \gstar$ is a constant, for all $t$ small enough. This proves statement (1). For (2),  consider an arbitrary $ X\in \galgebra$. On $\mathcal{M}_{\alpha}$, one computes the following
		\begin{eqnarray*}
			[\Hamiltonian{f}, X_\mathcal{M} ]=[\Hamiltonian{f}, \Hamiltonian{J^X} ]&=& \Hamiltonian{\PoissonBracket{f,J^X}}\\
			&=&-\Hamiltonian{\Hamiltonian{J^X}(f)}=-\Hamiltonian{X_\mathcal{M}(f)}=0.
		\end{eqnarray*}  
		The last step is again due to $f$ being $G$-invariant. As $G$ is connected, the flow $\varphi^t$ of $\Hamiltonian{f}$ commutes with the $G$-action on $\mathcal{M}_{\alpha}$, i.e., for all $m\in \mathcal{M}_{\alpha}$ and $g\in G$, we have 
		$$
		\varphi^t(\sigma_g m)=\sigma_g \varphi^t(m).
		$$
		It follows that $G_m$ coincides with $G_{\varphi^t(m)}$, and hence  $\Hamiltonian{f}|_m\in T_m (\mathcal{M}_{\alpha})_H$ for $H=G_m$.
		
	\end{proof}

	\begin{theorem}\label{Thm:Poissonreducedspace} 
		Let $(\mathcal{M},\Pi)$ be a stratified Poisson $G$-space with the moment map $J$ satisfying the conditions in Lemma \ref{lemma:tangentproperty}, where $G$ is a connected Lie group. Then the following holds:		
		\begin{enumerate}
			\item[(1)] The orbit space $\mathcal{M}/G$ inherits a stratified Poisson structure, and the natural quotient projection $\mathcal{M}\to \mathcal{M}/G$ preserves the Poisson structure.
			
			\item[(2)] For any $\xi\in \Im J$, the level set $\Zxi:=J^{-1}(\Oxi)$ admits a stratification whose strata are given by the intersections $\Zxi\cap (\mathcal{M}_{\alpha})_{(H)}$, where $\mathcal{M}_{\alpha}$ denotes a stratum of $\mathcal{M}$ and $(H)$ represents an orbit type. Consequently, the reduced space $\Reducedxi:=J^{-1}(\Oxi)/G$ inherits a stratification with strata of the form
		$$(\Reducedxi)_{\alpha,(H)}:=(\Zxi\cap (\mathcal{M}_{\alpha})_{(H)})/G.$$
			The smooth function spaces on $\Zxi$ and $\Reducedxi$ are characterized analogously to parts (3) and (4) of Theorem \ref{Thm:SLReduced}.
			
			\item[(3)] The reduced space $\Reducedxi$ carries a stratified Poisson structure compatible with the canonical embedding $\epsilon:\Reducedxi\hookrightarrow \mathcal{M}/G$, making $\epsilon$ a morphism of stratified Poisson spaces.
			
			\item[(4)] The embedding $\epsilon$ maps each single symplectic leaf within a stratum $(\Reducedxi)_{\alpha,(H)}$ of $\Reducedxi$ onto a single symplectic leaf in the corresponding stratum $(\mathcal{M}/G)_{\alpha,(H)}$ of $\mathcal{M}/G$.
		\end{enumerate}
		
	\end{theorem}  
	
	\begin{proof}
     The Poisson structure on the quotient space $\mathcal{M}/G$ requires careful verification. First, observe that the algebra of functions $\Cinf(\mathcal{M}/G)=\CinfcalM^G$ naturally inherits a Poisson bracket $\PoissonBracket{\cdot,\cdot}$ from the Poisson structure on $\Cinf(\mathcal{M})$ induced by $\Pi$. To establish that this inherited bracket corresponds to a well-defined Poisson bivector field $\Pi$ on $\mathcal{M}/G$, we must verify certain compatibility conditions.
		
	According to \cite[Proposition 4.5]{Pflaumpaper}, it suffices to demonstrate that for each stratum $\mathcal{M}_{\alpha}$ of $\mathcal{M}$ and its associated orbit type submanifold $(\mathcal{M}_{\alpha})_{(H)}$, the following relation holds:
		$$\PoissonBracket{f_1,f_2}|_{(\mathcal{M}_{\alpha})_{(H)}}=0$$
		for all $f_1,f_2\in \Cinf(\mathcal{M}/G)$ where $f_2|_{(\mathcal{M}_{\alpha})_{(H)}}=0$. This condition is indeed satisfied since, by Lemma \ref{lemma:tangentproperty} (2), the Hamiltonian vector field $\Hamiltonian{f_1}$ of $f_1$ is tangent to $(\mathcal{M}_{\alpha})_{H}$.
		
		The properties outlined in (2) are fundamentally equivalent to conditions (1)-(4) of Theorem \ref{Thm:SLReduced}, and their verification follows the arguments presented in \cite{SLannals,BatesLerman1997,Pflaumbook}. 
		
		Regarding (3), the tangency of the Poisson bivector field $\Pi$ to the pieces $\Zxi\cap (\mathcal{M}_{\alpha})_{(H)}$ is established by claims (1) and (2) of Lemma \ref{lemma:tangentproperty}. As a direct consequence, Property (4) follows from this tangency condition.

	\end{proof}
	\begin{remark}
		In the context of Theorem \ref{Thm:Poissonreducedspace}, one can demonstrate that $\Zxi\cap \mathcal{M}_{\alpha}$ is \textit{coisotropic} within the Poisson manifold $(\mathcal{M}_{\alpha},\Pi | _{\mathcal{M}_{\alpha}})$ across all strata $\mathcal{M}_{\alpha}$ of $\mathcal{M}$.
		
	\end{remark}

\section{Stratified superintegrable systems}\label{Sec:Superintegrability}

In this section, we start with a brief overview of classical superintegrability on symplectic manifolds and then   generalize  the definition of
superintegrability to stratified spaces. 
\subsection{Classical superintegrability for symplectic manifolds}\label{Sec:sub-classicalsuperint}
For details on superintegrability on manifolds, see,  for example, 
\cite{Old-review,Nek1972,NR2003}. Our definition below is an improvement based on \cite{NR2003}.
 \begin{definition}\label{Def:superint3}  A strong classical  \textbf{superintegrable system}   on a $2r$-dimensional symplectic manifold $M $ 
 is a structure comprising three key components:  ~ a Poisson manifold $P $ of dimension ${2r-k}$, a smooth manifold $B$ (of dimension   $\geqslant k$)  with a trivial Poisson structure, and a sequence of   Poisson maps 
	\begin{equation}\label{Seq:generalMPB}
		M  \overset{\psi}{\longrightarrow} P  \overset{\pi}{\longrightarrow} B  \,,
	\end{equation}
satisfying the following conditions:
\begin{itemize}
  \item[(1)] The map  $\psi$ is a surjective  submersion; 
  \item[(2)]  The rank of $\pi$ is $k$ (i.e., the image of the tangent map $\pi_{*p}$ is a $k$-dimensional subspace in $T_{\pi(p)}B$ for all $p\in P $);
  \item[(3)] For any point  $b$ in $\pi(P )\subset B$, connected components of the preimage $\pi^{-1}(b)$   are symplectic leaves of $P $. 
\end{itemize}
	
\end{definition} 

Here, the adjective ``strong'' indicates that Condition (2) must hold throughout \(P\), while Condition (3) must hold throughout \(\pi(P)\). In the literature, by contrast, a superintegrable system usually requires these conditions to hold only at generic points.
 
Although it is common in the literature to require the space $B$ to have dimension exactly $k$, we allow $\dim B \geqslant k$. This more general setting arises in our subsequent analysis of concrete examples. For instance, the sequence  $(\Reducedxi)^2_{(1')}$  $\to (\Pxi)^1_{(1')}$  $\to (\Bxi)^2$ in Diagram \eqref{Diagram:SU3arrowsabcrefined} provide such an example. 
 
Superintegrability extends the concept of Liouville integrability. When $k=r$, a  strong classical superintegrable system  is a Liouville integrable system.

 From a strong classical superintegrable system as in Definition \ref{Def:superint3}, we can find a  Poisson subalgebra  $\mathscr{A}:= \psi^* \Cinf(P ) $ in ${C}^\infty(M )$. 
 The Poisson center $Z(\mathscr{A})$ of $\mathscr{A}$ is also easily found:
 $$Z(\mathscr{A})=(\pi\circ \psi)^*\Cinf(B ).
 $$
 Note that $ \mathscr{A} $ is of rank $2r-k$ and $Z(\mathscr{A})$ is of rank $k$. Here, by saying that a subalgebra $\mathscr{C}\subset {C}^\infty(M )$ is of rank $s$, we mean that the subspace of $T^*_p M $ spanned by differentials of functions in $\mathscr{C}$ is   of dimension $s$ for almost all $p\in M$. 
 
 Conversely, suppose that we are given a Poisson subalgebra $\mathcal{A} \subset C^\infty(M )$ of rank $2r-k$ whose Poisson center $Z(\mathcal{A})$ has rank $k$. If the level sets $P $ and $B $, defined by the functions in $\mathcal{A}$ and $Z(\mathcal{A})$, respectively, are smooth manifolds, then the sequence \eqref{Seq:generalMPB} can be naturally constructed to form a strong classical superintegrable system. See the next Example \ref{Ex:Kepler}.

 In practice, however, these smoothness conditions are often prohibitively restrictive. This necessitates decomposing the relevant spaces into smaller components and identifying superintegrable systems on these local partitions, a procedure illustrated in the subsequent examples.

\begin{example}[Kepler system]\label{Ex:Kepler} 
	The phase space of the classical $3$-dimensional Kepler system is the $6$-dimensional symplectic manifold $M = T^*(\mathbb{R}^3 \setminus \{0\})$ with the standard Darboux symplectic form $\omega$ on the cotangent bundle. Denote by $(q,p)$ the standard coordinate system  on $M$ where $q  $ gives the position in $\mathbb{R}^3 \setminus \{0\}$ and $p $ stands for the momentum. The standard Hamiltonian on $M$ reads $$h(q,p) = \frac{1}{2}\|p\|^2 - \frac{\mu}{\|q\|} ,$$ where $\mu>0$ denotes the gravitational parameter. The global constants of motion are the energy $h$, the angular momentum ${L} = {q} \times {p}$, and the Runge-Lenz vector ${A} = {p} \times {L} - \mu \frac{{q}}{\|{q}\|}$. 
		
	We construct a  strong classical superintegrable system according to Definition \ref{Def:superint3} as follows:
	\begin{enumerate}
		\item \textbf{The Manifolds:} Define the moment map $\psi: M \to \mathbb{R}^7$ by $\psi(q,p) = (h, {L}, {A})$ and take its image $P=\Im \psi$ as the intermediate   manifold. This $P$ is indeed a smooth manifold of dimension $5$ because it is defined by   two algebraic constraints:
		$ {L} \cdot {A} = 0$ and $  \|{A}\|^2 - 2h\|{L}\|^2 - \mu^2 = 0$;
		and the Jacobian matrix of $({L} \cdot {A}, \|{A}\|^2 - 2h\|{L}\|^2 - \mu^2)$ with respect to $(h, {L}, {A})$ has  full rank $2$ everywhere.
		Meanwhile, we   take $B  = \mathbb{R}$.
			
		\item \textbf{The Poisson structures:} The functions $h$, $L_i$, $A_j$ ($i,j=1,2,3$) together generate a   Poisson subalgebra in $\Cinf(M)$. Any function in $\psi^*\Cinf(P)$ can be (locally) approximated by elements in this subalgebra and hence it induces  a smooth Poisson structure on $P$ such that $\psi$ is a   Poisson map.  
		Define $\pi: P (\subset \mathbb{R}^7)\to B=\mathbb{R}$ simply by projecting to its first component. It is   a Poisson map if we treat $B$ as  a trivial Poisson manifold. One can easily check that $\psi$ and $\pi$ are both submersions.

			
			
			
		\item \textbf{Symplectic leaves:} The Poisson center of $\psi^*\Cinf(P)$ is clearly generated by $h$. Therefore, every connected component of the inverse image $\pi^{-1}(b)$ ($b\in B$) is a symplectic leaf of $P$. 
	\end{enumerate}
In summary, the sequence $M \overset{\psi}{\longrightarrow} P \overset{\pi}{\longrightarrow} B$ constitutes a strong classical superintegrable system; for additional background, see \cite{Old-review,Keplerbook}.\end{example}

\subsection{Two illustrative examples}\label{Ex:partitionnotstratified}
We wish to extend  the concept of superintegrability for smooth symplectic manifolds to stratified symplectic spaces. To illustrate our  approach, we begin with two simple models.

 \begin{example}
    \label{Ex:R4}
    Consider $M =\mathbb{R}^4$ with standard coordinates $(x,y,r,s)$ and the symplectic form
    $dx\wedge dy+dr\wedge ds$.  Take the   Poisson subalgebra $\mathcal{A}\subset \Cinf(M )$ with the trivial bracket generated by   $  f_1=   x +r$ and  $f_2= {x}^2+r^2 $. Clearly, we have $Z(\mathcal{A})=\mathcal{A}$. Both $\mathcal{A}$ and $Z(\mathcal{A})$ have rank $2$. So we can consider the level set of functions in $\mathcal{A}$, namely
    $$
    \mathcal{P} =\Im \psi, \mbox{ where } \psi: ~M \to \mathbb{R}^2,\quad (x,y,r,s)\mapsto (f_1,f_2).
    $$
    It can be easily seen that 
 $$\mathcal{P}  =\SET{(u,v)\in \mathbb{R}^2|  u^2\leqslant 2v} .$$

     Now we form the following sequence of maps
 \begin{equation} \label{Seq:ExampleToySequenceR4} M   \overset{\psi}{\longrightarrow} \mathcal{P}  \overset{\pi}{\longrightarrow} \mathcal{B} ,
  \end{equation}
 where  $\mathcal{P} =\mathcal{B}  $   and
 $\pi$ is the identity map. Thus, we obtain objects similar to those in Definition \ref{Def:superint3}.  But we \emph{cannot} claim that the above sequence \eqref{Seq:ExampleToySequenceR4} is superintegrable in the classical sense because both $\mathcal{P} $ and $\mathcal{B} $ are  stratified spaces.  
 Indeed, $\mathcal{P} $ is stratified into the following two strata: $$(\mathcal{P})^2=\SET{(u,v)\in \mathbb{R}^2|  u^2<2v},$$ 
 $$(\mathcal{P})^1 =\SET{(u,v)\in \mathbb{R}^2|  u^2=2v},$$
 and so is $\mathcal{B} $. Here, the notation $(\mathcal{P})^2$ and $(\mathcal{P})^1$ represent, respectively,  strata of  dimension $2$ and $1$   of $\mathcal{P} $,  and similar notation  in the sequel is interpreted in a comparable manner.

  It is natural to expect that $M $ can be partitioned into smaller strata which are superintegrable.   The method of partition we propose is to consider  inverse images of strata:  
      $$
   (M)^4=\psi^{-1} (\mathcal{P})^2 =
   \SET{(x,y,r,s)\in \mathbb{R}^4| x\neq r},
   $$
   $$(M)^3 =\psi^{-1}(\mathcal{P})^1 =
   \SET{(x,y,x,s)\in \mathbb{R}^4 }.
      $$

On $(M)^4$ there is an evident strong classical superintegrable system:
\begin{equation} \label{Seq:ExampleToySequenceR4max} (M)^4  \overset{\psi}{\longrightarrow} (\mathcal{P})^2   \overset{\pi}{\longrightarrow} (\mathcal{B})^2  .
  \end{equation}
On $(M)^3 $    (a \emph{presymplectic} submanifold in $M $) there is a sequence of maps which does \emph{not} form     a strong classical superintegrable system:
\begin{equation} \label{Seq:ExampleToySequenceR4min} (M)^3   \overset{\psi}{\longrightarrow} (\mathcal{P})^1    \overset{\pi}{\longrightarrow} (\mathcal{B})^1    .
  \end{equation}
  However, the   \emph{reduced symplectic space}\footnote{The reduced symplectic space $\underline{M}$ of a presymplectic manifold $(M,\omega)$ is the quotient space of $M$   by the flow of the vector fields in the kernel of $\omega$. If it happens that $\underline{M}$ is also a smooth manifold, then $\underline{M}$ inherits a   symplectic structure from $(M,\omega)$ in a natural way.} $\underline{(M)^3}    $ (which is $2$-dimensional)  admits a sequence 
  \begin{equation} \label{Seq:ExampleToySequenceR4minreduced} \underline{(M)^3}      \overset{\underline{\psi}}{\longrightarrow} (\mathcal{P})^1    \overset{\pi}{\longrightarrow} (\mathcal{B})^1 ,   
  \end{equation}
  which is indeed a strong classical superintegrable system and the map $\underline{\psi}$ fits into the following commutative diagram: 
     \begin{equation} \nonumber     \begin{tikzcd}
 	 (M)^3 \arrow[dr,"\psi "] \arrow[d,"p"] & ~   \\
 	\underline{(M)^3}   \arrow[r,"\underline{\psi} "']     & (\mathcal{P})^1      
 \end{tikzcd}
\end{equation}
where $p$ is the standard quotient map.

In summary, we have a  symplectic stratified space $M =(M)^4\sqcup (M)^3$   whose maximum stratum $(M)^4$ admits a strong classical superintegrable system, and the symplectic reduction of the minimum stratum $(M)^3$ does so as well.
  \end{example}
  
  \begin{example}
    \label{Ex:R6}
Consider $M = \mathbb{C} ^3\cong (\mathbb{R}^2)^3   $, which is symplectic as the triple product  of $\mathbb{R}^2$. Write $(x=u+\imagunit v,y=p+\imagunit q,z=r+\imagunit s)$ for points in $M $ with the   only nontrivial Poisson brackets as follows:
$$
\SET{u,v}=\SET{p,q}=\SET{r,s}=1.
$$
 Take the Poisson subalgebra $\mathcal{A}\subset \Cinf(M )$  generated by the following functions $$ f_1=\abs{x}^2=u^2+v^2, \quad f_2=\abs{y}^2-\abs{z}^2=p^2+q^2-r^2-s^2,$$$$\quad  f_3=\mathrm{Re}(y\bar{z})=pr+qs, \quad  f_4=\mathrm{Im}(y\bar{z})=qr-ps.$$ The Poisson center  $Z(\mathcal{A})$ is clearly generated by $f_1$ and $$f_5=(\abs{y}^2+\abs{z}^2)^2= (p^2+q^2+r^2+s^2)^2=f_2^2+4(f_3^2+f_4^2).$$ 
 Therefore, we can consider the level set $\mathcal{P} $ of functions in $\mathcal{A}$
 and the level set $\mathcal{B} $ of functions in $Z(\mathcal{A})$. Namely, we define $$\psi:~M \to \mathbb{R}^4,\quad (u,v,p,q,r,s)\mapsto (f_1,f_2,f_3,f_4), $$
 $$\mathcal{P}  :=\Im \psi=\SET{(a,b,c,d)\in \mathbb{R}^4 ~|~ a\geqslant 0}\cong   \mathbb{R}^{\geqslant 0}\times \mathbb{R}^3,$$
 and define
 $$
 \pi:~\mathbb{R}^4\to \mathbb{R}^2,\quad (a,b,c,d)\mapsto (a,b^2+4(c^2+d^2)),
 $$
 $$\mathcal{B} =\Im (\pi\circ \psi)=\SET{(\alpha,\beta)\in \mathbb{R}^2| \alpha\geqslant 0,\beta\geqslant 0}\cong \mathbb{R}^{\geqslant 0}\times \mathbb{R}^{\geqslant 0}. $$
 Note that the map $\pi\circ \psi$ is indeed $(f_1,f_5)$, and we obtain a sequence of maps
 \begin{equation} \label{Seq:ExampleToySequence2} M   \overset{\psi}{\longrightarrow} \mathcal{P}  \overset{\pi}{\longrightarrow} \mathcal{B} .
  \end{equation}

  Next, we wish to see how Sequence \eqref{Seq:ExampleToySequence2} is integrated as a superintegrable system by decomposing    it into subsequences.   Clearly, $\mathcal{P}$ is stratified as:
 $$\mathcal{P} =(\mathcal{P})^4 \sqcup (\mathcal{P})^3 ,$$
 where
 $$(\mathcal{P})^4 =\mathbb{R}^{> 0}\times \mathbb{R}^3,\qquad (\mathcal{P})^3 =\SET{0}\times \mathbb{R}^3.$$
 Similarly, $\mathcal{B}$ has the following strata: 
 $$(\mathcal{B})^2 =\mathbb{R}^{> 0}\times \mathbb{R}^{> 0},\quad (\mathcal{B})^1_{(1)}  =\SET{0}\times \mathbb{R}^{> 0},\quad (\mathcal{B})^1_{(2)}  = \mathbb{R}^{> 0}\times \SET{0},\quad  (\mathcal{B})^0  =\SET{0}\times \SET{0}.
 $$
At this point, we cannot claim that Sequence \eqref{Seq:ExampleToySequence2} forms a superintegrable system. Also,    the maps $\psi$ and  $\pi$ do not  send stratum to stratum. 

We need to  refine $\mathcal{P} $. Note that $\pi$ maps $(\mathcal{P})^4    $ to $(\mathcal{B})^2$ $\sqcup$ $ (\mathcal{B})^1_{(2)}  $, and $(\mathcal{P})^3 $ to $(\mathcal{B})^1_{(1)}  $ $\sqcup$ $(\mathcal{B})^0    $. So we can define the following subspaces:
$$(\mathcal{P})^0  =\pi^{-1} (\mathcal{B})^0\cap (\mathcal{P})^3  =\SET{(0,0,0,0)}, 
 $$
$$  (\mathcal{P})^1   =
 \pi^{-1} (\mathcal{B})^1_{(2)} \cap (\mathcal{P})^4 =\mathbb{R}^{> 0}\times \SET{(0,0,0)},
$$
$$(\mathcal{P})^3_{*}  
=\pi^{-1} (\mathcal{B})^1_{(1)}   \cap (\mathcal{P})^3
= (\mathcal{P})^3  \setminus    (\mathcal{P})^0  ,
$$
$$(\mathcal{P})^4_{*}   =\pi^{-1}
(\mathcal{B} )^2 \cap (\mathcal{P})^4
= (\mathcal{P})^4   \setminus    (\mathcal{P})^1    .
$$
We thus have a refined stratification $$
  \mathcal{P} = (\mathcal{P})^4_*         \sqcup (\mathcal{P})^3_{*}   \sqcup (\mathcal{P})^1       \sqcup (\mathcal{P})^0     .
  $$ 
 We then further decompose $M =\mathbb{C}^3$ into the following four strata:
 $$(M)^6        :=\psi^{-1} (\mathcal{P})^4_*         
 \cong    (\mathbb{C}\setminus \SET{0})\times (\mathbb{C}^2\setminus \SET{(0,0)}),
 $$
 $$
 (M)^4          :=
 \psi^{-1}(\mathcal{P} )^3_* 
 \cong      \SET{0} \times (\mathbb{C}^2\setminus \SET{(0,0)}),
  $$
$$(M)^2         :=\psi^{-1}
 (\mathcal{P})^1      
 \cong    (\mathbb{C}\setminus \SET{0})\times \SET{(0,0)} ,
 $$
$$(M)^0          :=\psi^{-1}
 (\mathcal{P})^0      
 \cong      \SET{(0,0,0)} .
 $$

 With the new stratifications of $M$ and $\mathcal{P}$, Sequence \eqref{Seq:ExampleToySequence2} becomes a morphism of stratified spaces. In fact,   we can see that $\psi$ and $\pi$ are partitioned into   the following components:
\begin{equation}
	\label{Diagram:EXAMPLEM6}
	\begin{tikzcd}
		(M)^6         \arrow[d,dashed] \arrow[dd,bend left=55,dashed]    \arrow[rr,"\psi "] && (\mathcal{P})^4_*           \arrow[d,dashed] \arrow[dd,bend left=55,dashed] \arrow[rr,"\pi " ]   && (\mathcal{B})^2    \arrow[d,dashed] \arrow[dd,bend left=55,dashed] \\
		(M)^4           \arrow[dd,bend left=55,dashed]  \arrow[rr ,"\psi "] && (\mathcal{P})^{3}_*  \arrow[dd,bend left=55,dashed]  \arrow[rr ,"\pi "]  && (\mathcal{B})^1_{(1)}   \arrow[dd,bend left=55,dashed]  \\ (M)^2           \arrow[rr,"\psi "]\arrow[d,dashed]
		&& (\mathcal{P})^1      \arrow[rr,"\pi "] \arrow[d,dashed]&&  (\mathcal{B})^1_{(2)}   \arrow[d,dashed]
\\ (M)^0     \arrow[rr,"\psi "]
		&& (\mathcal{P})^0       \arrow[rr,"\pi "] && (\mathcal{B})^0  \,, \end{tikzcd}
\end{equation}
  such that all vertical sequences are superintegrable in the strong classical sense.

  In this diagram (and the sequel), the dashed arrows connecting strata of $M $, $\mathcal{P} $, and $\mathcal{B} $ are drawn to  show   boundary relations ---   the target stratum is a boundary of the source stratum.  
  
  \end{example}

 \subsection{Definition of stratified superintegrable systems}
 The preceding Examples \ref{Ex:R4} and \ref{Ex:R6} demonstrate that, under certain circumstances, the initial symplectic space requires decomposition into distinct strata, each characterized by either a symplectic or a presymplectic structure. (The maximum stratum certainly admits a symplectic structure.)  Furthermore,   each stratum -- or its corresponding symplectic reduction -- admits a strong classical superintegrable system. Based on these observations, we are now in a position to establish a formal definition of stratified superintegrable systems as a generalization of the strong classical framework.

 \begin{definition}\label{Def:stratifiedsuperintNEW}
 A \textbf{stratified superintegrable system} consists of the following data:
\begin{itemize}
    \item[(1)~] A stratified symplectic space $(\mathcal{M},\omega)$ (see Definition \ref{def:presymlecticstratifiedspace}) with index set of strata denoted by $A$.

    \item[(2)~] A stratified Poisson space $(\mathcal{P},\Pi)$ (see Definition \ref{def:poissonsymlecticstratifiedspace})  with index set of strata denoted by $C$.

    \item[(3)~] A stratified space $\mathcal{B}$ endowed with the \textit{trivial Poisson} structure  with index set of strata denoted by $B$.

    \item[(4)~] A sequence of smooth Poisson maps between stratified Poisson spaces 
    \begin{equation}\label{NEWSeq:generalMPBstratified}
        \mathcal{M} \overset{\psi}{\longrightarrow} \mathcal{P}
        \overset{\pi}{\longrightarrow} \mathcal{B},
    \end{equation}
    where $\psi$ is surjective.
\end{itemize}

These data are required to satisfy the following conditions:
\begin{itemize}
    \item[(5)~] The stratification of $\mathcal{P}$ is refined by intersecting each of its strata with the inverse images under $\pi$ of the strata of $\mathcal{B}$. Thus,
    \[
        \mathcal{P}
        = \bigsqcup_{(\delta,\beta)\in C\times B}
        \mathcal{P}_{\delta,\beta},
    \]
    where
    \[
        \mathcal{P}_{\delta,\beta}
        := \mathcal{P}_{\delta}\cap \pi^{-1}(\mathcal{B}_{\beta}).
    \](See Appendix \ref{Apd:refinebyinverseimage} for more details of this method.) 
   With respect to this refined stratification, $\mathcal{P}$ inherits a stratified Poisson structure. Equivalently, the Poisson structure $\Pi$ is tangent to every refined stratum $\mathcal{P}_{\delta,\beta}$.

    \item[(6)~] The stratification of $\mathcal{M}$ is then refined by taking the inverse images under $\psi$ of the refined strata of $\mathcal{P}$. Consequently,
    \[
        \mathcal{M}
        = \bigsqcup_{(\alpha,\delta,\beta)\in A\times C\times B}
        \mathcal{M}_{\alpha,\delta,\beta},
    \]
    where
    \[
        \mathcal{M}_{\alpha,\delta,\beta}
        := \mathcal{M}_{\alpha}\cap\psi^{-1}(\mathcal{P}_{\delta,\beta})
        = \mathcal{M}_{\alpha}\cap\psi^{-1}(\mathcal{P}_{\delta})
        \cap(\pi\circ\psi)^{-1}(\mathcal{B}_{\beta}).
    \]
     With respect to this refined stratification, $\mathcal{M}$ inherits a stratified presymplectic structure. More precisely, if
    \[
        j:\mathcal{M}_{\alpha,\delta,\beta}\hookrightarrow\mathcal{M}_{\alpha}
    \]
    denotes the inclusion map, then the presymplectic form on $\mathcal{M}_{\alpha,\delta,\beta}$ is
    \[
       \omega_{\alpha,\delta,\beta}
       =j^*(\omega|_{\mathcal{M}_{\alpha}}).
    \]

    \item[(7)~] For each refined stratum $\mathcal{M}_{\alpha,\delta,\beta}$, equipped with the presymplectic form $\omega_{\alpha,\delta,\beta}$, let $\underline{\mathcal{M}}_{\alpha,\delta,\beta}$ denote its symplectic reduction by the null foliation of $\omega_{\alpha,\delta,\beta}$. We require each $\underline{\mathcal{M}}_{\alpha,\delta,\beta}$ to be a genuine smooth manifold and   thus inherit  a symplectic structure. The collection of these reduced spaces is required to form a stratified symplectic space $(\underline{\mathcal{M}},\underline{\omega})$.

    \item[(8)~] With respect to the refined stratifications, the morphism of stratified spaces $\psi:\mathcal{M}\to\mathcal{P}$ factors through $\underline{\mathcal{M}}$. In other words, there exists a morphism of stratified spaces
    \[
        \underline{\psi}:\underline{\mathcal{M}}\to\mathcal{P}
    \]
    such that the following diagram commutes:
    \begin{equation*}
    \begin{tikzcd}
        \mathcal{M} \arrow[dr,"\psi"] \arrow[d,"p"] & & \\
        \underline{\mathcal{M}} \arrow[r,"\underline{\psi}"'] & \mathcal{P}
    \end{tikzcd}
    \end{equation*}
    where $p$ is the standard quotient map. Moreover, $\underline{\psi}$ is Poisson.

    \item[(9)~] For every triple $(\alpha,\delta,\beta)$, the sequence
    \begin{equation}\label{Eqt:subsuperintsequence}
        \underline{\mathcal{M}}_{\alpha,\delta,\beta}
        \overset{\underline{\psi}}{\longrightarrow}
        \mathcal{P}_{\delta,\beta}
        \overset{\pi}{\longrightarrow}
        \mathcal{B}_{\beta}
    \end{equation}
    is superintegrable in the strong classical sense (see Definition \ref{Def:superint3}).
\end{itemize}
     \end{definition}
    
    
 In this definition, the original stratifications of \(\mathcal{M}\) and \(\mathcal{P}\) are referred to as their primary stratifications. A sequence of the form \eqref{Eqt:subsuperintsequence}, constructed from refined stratifications, is called a \textbf{sub-superintegrable system} of the stratified superintegrable system \eqref{NEWSeq:generalMPBstratified}.  The two examples in Section \ref{Ex:partitionnotstratified}, along with the refined stratifications and maps, demonstrate key components of a stratified superintegrable system and its sub-superintegrable systems.

	\section{General construction of spin Calogero-Moser-Sutherland  systems}\label{sCMS}
	
	Here we   describe the  sCMS  systems on the cotangent bundle associated with a \textit{compact} and \textit{connected}   Lie group $G$.
 Our aim is to   construct a  sequence of stratified Poisson spaces:
	\[
	\Reducedxi  \stackrel{\psi}{\longrightarrow} \Pxi \stackrel{\pi}{\longrightarrow} \Bxi ,
	\]
	where $\xi \in \Im J_{T^*G}\subset \gstar$ is fixed and $J_{T^*G}$ is the moment map defined in \eqref{Eqt:JTstarGgeneral}. The stratified space $\Reducedxi$ is defined in Section \ref{subSec:phasespaceRxifromG}. The mappings $\psi$ and $\pi$ and the spaces $\Pxi$ and $\Bxi$ are defined in the next subsections. This sequence of maps will be useful to establish the stratified superintegrability of sCMS systems associated with   $G=\SU(2)$ (see Section \ref{Sec:SU(2)}) or $\SU(3)$ (see Section \ref{Sec:sCMSRPBSU3}).
	
	\subsection{The stratified    space  $\Bxi$}\label{subsubSec:Bxi}~
	
First, we introduce $\Gamma_{\xi}\subset \gstar$:
		
		\begin{eqnarray*}
			\Gamma_{\xi} &:=& \SET{x \in \gstar ~|~ \exists \gamma \in G \mbox{ such that } x - \Adjointaction^*_{\gamma^{-1}}(x) \in \Oxi}.  
		\end{eqnarray*} Here $\mathcal{O}_{\xi} \subset \gstar$ is the coadjoint orbit through $\xi$. 
		
The coadjoint action of $G$ on $\gstar$ preserves $\Gamma_{\xi}$ naturally. So we have the quotient space $\Bxi:=\Gamma_{\xi}/G$. This $\Bxi$ is usually a stratified space. Let us suppose that 
$$
\Bxi=\sqcup_{\beta\in B} (\Bxi)_{\beta},
$$
where $B$ is the index set. Accordingly, we can decompose
\begin{equation}\label{Eqt:GammaxidecomposebyB} 
\Gamma_{\xi}=\sqcup_{\beta\in B} (\Gamma_{\xi})_{\beta},\quad \mbox{ where }(\Gamma_{\xi})_{\beta}=\SET{x\in \Gamma_{\xi}~|~[x]\in (\Bxi)_{\beta}}.
\end{equation}

We shall treat $\Bxi$ as a stratified Poisson space whose Poisson structure is simply \textit{trivial}. In fact, $\gstar$ possesses the (nontrivial) Kirillov-Kostant Poisson   structure. But the natural Poisson structure on $\gstar/G$ is trivial. Consequently, the subspace $\Bxi \subset \gstar/G$ only admits a \textit{trivial} Poisson structure.

 We further assume   the following condition. 

\begin{assumption}\label{Assumption:sameorbittype}
Every stratum $(\Bxi)_{\beta}$ corresponds to a unique conjugacy class of orbit types $E\subset G$. In other words,   for all   $x\in  (\Gamma_{\xi})_{\beta}$ we have $G_x\sim E$.
\end{assumption}

	\subsection{The stratified  Poisson space $\Pxi$}\label{subsubSec:RPBxi}~
	
	First, let us construct the stratified Poisson space $\Pxi$ through the following steps:
	\begin{itemize}
		
		\item[(1)] Consider the Poisson $G$-manifold $\gstar \times \gstar$ with the product Poisson structure of  Kirillov-Kostant Poisson structures on $\gstar$. On each factor $\gstar$, the action of $G$ is coadjoint, and on the product, $G$ acts diagonally.
		
		\item[(2)] The diagonal action of $G$ on $\gstar \times \gstar$ is Hamiltonian with the 
		moment map $J_{\gstar \times \gstar}: \gstar \times \gstar \to \gstar$ acting as $(x,y)\mapsto x+y$.
		
		\item[(3)] Define a subset $\middleK \subset \gstar \times \gstar$ as follows:
		$$
		\middleK := \SET{(x, y) \in \gstar \times \gstar ~|~ y \in \mathcal{O}_{-x}} = \SET{(x, -\Adjointaction^*_{\gamma^{-1}}(x)) \in \gstar \times \gstar, \text{for some } \gamma \in G}.
		$$
		
		\item[(4)]The Lie group $G$ acts on $\middleK$,  as it is a $G$-invariant subset of $\gstar \times \gstar$. Hence, $\middleK$ has a natural stratification determined by the stabilizers with strata parameterized by conjugacy classes of closed subgroups: 
		$$
		\middleK(E) := \SET{(x, y) \in \middleK ~|~ G_x (\text{or } G_y) \sim E}.
		$$
		Here $E \subset G$ is a closed subgroup. Note that if $(x,y)\in \middleK$ we have $G_x\sim G_y$. 

		
		\item[(5)] Alternatively,  $\middleK$ can be described as follows
		$$
		\middleK = \bigsqcup_{[x] \in \gstar/G} \mathcal{O}_x \times \mathcal{O}_{-x}.
		$$
		From this description we see that Hamiltonian vector fields on $\gstar \times \gstar$ are tangent to every piece $\mathcal{O}_x \times \mathcal{O}_{-x}$, and to every stratum  $\middleK(E)$ as well. Consequently, the stratified space $\middleK$ inherits a Poisson bivector field from $\gstar \times \gstar$, making it a stratified Poisson $G$-space.
		It is clear that we can also write $\middleK$ as a fibered product
		\[
		\middleK =\gstar ~\tilde{\times}_{\gstar/G}~\gstar=\{(x,y)\in \gstar\times \gstar ~|~ [x]=-[y]\in \gstar/G\}
		\]

		\item[(6)] By restricting the moment map $J_{\gstar \times \gstar}$ to $\middleK$, we obtain a new moment map $J_\middleK: \middleK \to \gstar$. Note that $\Im J_\middleK = \Im J_{T^*G}$, where $J_{T^*G}$ is the moment map defined in Section \ref{Sec:introduction}. For $\xi \in \Im J_{T^*G} = \Im J_\middleK$, we have 
$$J_\middleK^{-1}(\mathcal{O}_\xi)  = \SET{(x, -\Adjointaction^*_{\gamma^{-1}}(x)) \in \gstar \times \gstar ~|~ x - \Adjointaction^*_{\gamma^{-1}}(x) \in \mathcal{O}_\xi} = \Projection_1^{-1}(\Gamma_{\xi}).$$
		Here $\Projection_1$ denotes the natural projection from $\middleK$ to its first component:
$$\Projection_{1}:  \quad \middleK \to \gstar.
$$ 

Following the decomposition \eqref{Eqt:GammaxidecomposebyB} of $\Gamma_{\xi}$, we have
 $$J_\middleK^{-1}(\mathcal{O}_\xi)=\sqcup_{\beta\in B} (J_\middleK^{-1}(\mathcal{O}_\xi))_{\beta},\quad \mbox{ where } (J_\middleK^{-1}(\mathcal{O}_\xi))_{\beta}=
\Projection_1^{-1}(\Gamma_{\xi})_{\beta}.$$

		\item[(7)]  Define the reduced space
		$$
		\Pxi := J_\middleK^{-1}(\mathcal{O}_\xi)/G
=  \SET{(x, -\Adjointaction^*_{\gamma^{-1}}(x)) \in \gstar \times \gstar ~|~ x - \Adjointaction^*_{\gamma^{-1}}(x) \in \mathcal{O}_\xi}/G,
		$$
		which is a stratified Poisson space, by Theorem \ref{Thm:Poissonreducedspace}.  Note that  
the primary stratification of $\Pxi$ is indexed by two orbit types:
$$
		(\Pxi)_{(E),(H)} = \SET{(x, -\Adjointaction^*_{\gamma^{-1}}(x)) \in J_\middleK^{-1}(\mathcal{O}_\xi)   ~|~G_x\sim E,~ G_{x, \Adjointaction^*_{\gamma^{-1}}(x)} \sim H }/G.
		$$
\item[(8)]We have a natural sequence of maps
	\begin{equation}\label{Seq:SxiPxiBxi}
		\Reducedxi\stackrel{\psi}{\longrightarrow} \Pxi \stackrel{\pi}{\longrightarrow} \Bxi
	\end{equation}
	defined as 
	\begin{equation}\label{Seq:SxiPxiBxi2}
		[(x,\gamma)]\stackrel{\psi}\mapsto [(x,-{\Adjointaction^*_{\gamma^{-1}}}(x))]\stackrel{\pi}\mapsto [x].
	\end{equation}
	\item[(9)] 
According to our Assumption \ref{Assumption:sameorbittype}, we know that every $(J_\middleK^{-1}(\mathcal{O}_\xi))_{\beta}$ belongs to some stratum $\middleK(E)$.
Therefore, any \textbf{refined stratum} of $\Pxi$ (obtained by intersecting each of its primary strata, say $(\Pxi)_{(E),(H)}$, with the
inverse images under $\pi$ of the strata of $\Bxi$, say $(\Bxi)_{\beta}$) is easily seen --- it is indexed by a pair  $(\beta, (H))$ where $\beta\in B$, $H\subset G$ such that
		$$
		(\Pxi)_{\beta,(H)} = \SET{(x, -\Adjointaction^*_{\gamma^{-1}}(x)) \in (J_\middleK^{-1}(\mathcal{O}_\xi))_{\beta} ~|~ G_{x, \Adjointaction^*_{\gamma^{-1}}(x)} \sim H }/G.
		$$
(See Section \ref{Apd:refinebyinverseimage} for more details about the refined stratification.) 
In the notation $(\Pxi)_{\beta,(H)}$, we omit the index $(E)$, as it is determined by $\beta$. In what follows, we use only this refined stratification of $\Pxi$.
 
 	\end{itemize}

	\subsection{The primary stratification of $\Reducedxi$}
	\label{subsubSec:psipi} In the sequence \eqref{Seq:SxiPxiBxi} which we established previously, 
the map $\pi: \Pxi\to \Bxi$ is by definition a morphism of stratified spaces (because $\Pxi$ is already refined). However, the first one, $\psi$, is \textit{not} necessarily a morphism. In fact, we have  the inclusion of stabilizers:
		\begin{equation}\label{Eqt:inclusionofGstuff}
		G_{x,\gamma}~\subset~ G_{x, {\Adjointaction^*_{\gamma^{-1}}}(x)} ~\subset~ G_{x}.
	\end{equation}

As detailed in Section \ref{Sec:SLmethod}, the reduced space $\Reducedxi$ is initially partitioned by orbit types   $\Reducedxi = \sqcup (\Reducedxi)_{(H)}$, which we shall call the \textbf{primary stratification} (namely the original Sjamaar-Lerman stratification) of $\Reducedxi$. Under this decomposition, however, a single stratum $(\Reducedxi)_{(H)}$ may be mapped by $\psi$ across multiple distinct strata in $\Pxi$. To ensure that $\psi$ constitutes a genuine morphism of stratified spaces, we will introduce a \textit{refined stratification} of $\Reducedxi$ in Section \ref{Sec:refineReducedxi}.

Nevertheless, one can see that both $\psi$ and $\pi$ are \textit{Poisson} maps of stratified Poisson spaces.  Indeed, we have mappings
	\begin{equation}
		\label{Seq:TstarGgstargstargstar}
		T^*G\stackrel{\Psi}{\longrightarrow} \gstar\times \gstar \stackrel{\Projection_1}{\longrightarrow} \gstar.
	\end{equation}
	where: 
	\begin{equation*} 
		\Psi(x,\gamma) =  (x,-{\Adjointaction^*_{\gamma^{-1}}}(x)),\quad \mbox{and } \Projection_1(x,y)=  x .
	\end{equation*}
	Both of these maps are Poisson. The map  $\Psi=\Projection_1 \times J_2$ is Poisson because both $\Projection_1$,  the projection to the first component, and 
$J_2: T^*G \to \gstar$, acting as $(x,\gamma) \mapsto -{\Adjointaction^*_{\gamma^{-1}}}(x)$,
	are Poisson.  Since $\Im(\Psi)=\middleK$, the sequence of projections \eqref{Seq:TstarGgstargstargstar} gives a sequence of surjective Poisson $G$-invariant maps :
	\begin{equation}\label{Seq:TstarGgstargstargstar2}
		T^*G\stackrel{\Psi}{\longrightarrow}~\middleK  ~ \stackrel{\Projection_1}{\longrightarrow} \gstar.
	\end{equation}
	Passing to quotient spaces we get a sequence of surjective Poisson maps within stratified Poisson spaces:
	\begin{equation}\label{Seq:TstarGgstargstargstaroverG}
		T^*G/G\stackrel{\tilde{\psi}}{\longrightarrow}~ \middleK/G~ \stackrel{\tilde{\pi}}{\longrightarrow} \gstar/G.
	\end{equation}
	The last space $\gstar/G$ in this sequence has a trivial Poisson structure.
	Restricting these maps to the subspace $\Reducedxi$ we obtain \eqref{Seq:SxiPxiBxi}.

	\subsection{Symplectic leaves of $\Pxi$}\label{Sec:symplecticleavesinPxi}~
	
	In the product space $\gstar \times \gstar$, symplectic leaves are products of coadjoint orbits $\mathcal{O}_{x} \times \mathcal{O}_{y}$, where $x, y \in \gstar$. Single symplectic leaves of $\middleK$ are $\mathcal{O}_{x} \times \mathcal{O}_{-x}$ for $x \in \gstar$. When considering the reduced space $\Pxi= J_\middleK^{-1}(\Oxi )/G$, a single symplectic leaf within a stratum 
	$(\Pxi)_{\beta,(H)}$ 
	is identified as (a connected component of) the set     
	$$\SET{[(x,-\Adjointaction^*_{\gamma^{-1}}(x))] ~|~ x - \Adjointaction^*_{\gamma^{-1}}(x) \in \Oxi, G_{x,\Adjointaction^*_{\gamma^{-1}}(x)} \sim H}$$ 
	for some fixed $x \in (\Gamma_{\xi})_{\beta}$  (according to Theorem \ref{Thm:Poissonreducedspace} (4)). 
This set precisely corresponds to $\pi^{-1}[x] \cap (\Pxi)_{\beta, (H)}$. Thus, any connected component of the fibers of the submersion 
	$$\pi:~(\Pxi)_{\beta,(H)} \to (\Bxi)_{\beta}$$ 
	is indeed a single symplectic leaf in $(\Pxi)_{\beta,(H)}$.

\subsection{The refined stratification of $\Reducedxi$}\label{Sec:refineReducedxi}
Consider the \emph{primary stratification} of $\Reducedxi$. Let $(\Reducedxi)_{(H_1)}$ be a stratum of $\Reducedxi$ and $(\Pxi)_{\beta,(H_2)}$ be a stratum of $\Pxi$. If the image of $(\Reducedxi)_{(H_1)}$ under $\psi$ has a non-empty intersection with $(\Pxi)_{\beta,(H_2)}$, the relation \eqref{Eqt:inclusionofGstuff} implies that $H_1$ is conjugate to a subgroup of $H_2$, denoted by $(H_2) \partialless (H_1)$. Under this condition, we define the subspace $(\Reducedxi)_{(H_1),\beta,(H_2)} \subseteq (\Reducedxi)_{(H_1)}$ as:
$$
(\Reducedxi)_{(H_1),\beta,(H_2)}=(\Reducedxi)_{(H_1)}\cap \psi^{-1} (\Pxi)_{\beta,(H_2)}.
$$
Suppose that after an appropriate subdivision of $(\Reducedxi)_{(H_1),\beta,(H_2)}$, the resulting pieces are submanifolds of $(\Reducedxi)_{(H_1)}$ and inherit \emph{presymplectic} structures from the symplectic structure of the ambient stratum. These subsets collectively define a new stratification of $\Reducedxi$, which we call the \textbf{refined stratification} (see Section \ref{Apd:refinebyinverseimage} for further details). For simplicity, we continue to denote the refined strata by $(\Reducedxi)_{(H_1),\beta,(H_2)}$, although a more detailed indexing system may be required in specific contexts to distinguish individual components.

From every refined stratum we have a sequence of restricted maps:
\begin{equation}\label{Seq:fromRH1betaH2}
		(\Reducedxi)_{(H_1),\beta,(H_2)} \stackrel{\psi}{\longrightarrow} (\Pxi)_{\beta,(H_2)} \stackrel{\pi}{\longrightarrow} (\Bxi)_{\beta}.
	\end{equation}
Such sequences collectively constitute the sCMS system \eqref{Seq:SxiPxiBxi} (equipped with the  refined stratification).  It is anticipated that such a system  on  $\Reducedxi$ is superintegrable in the sense of Definition \ref{Def:stratifiedsuperintNEW}. Indeed, the superintegrability of this system   requires specific conditions on the initial Lie group $G$. In the current paper, we can confirm this criterion for the cases $G=\SU(2)$ and $\SU(3)$, as detailed in Sections \ref{Sec:SU(2)} and \ref{Sec:sCMSRPBSU3}, respectively.
 

	\subsection{Hamiltonian functions} 
 
For any function $h \in \Cinf(\Bxi)$, the pullback $(\pi \circ \psi)^*h \in \Cinf(\Reducedxi)$ defines a Hamiltonian function on the symplectic stratified space $\Reducedxi$ (endowed with its primary  stratification). For notational simplicity, we shall denote $(\pi \circ \psi)^*h$ simply by $h$. For each orbit type $H_1$ within the primary stratification, there exists a corresponding Hamiltonian vector field $\Hamiltonian{h} \in \mathfrak{X}((\Reducedxi)_{(H_1)})$ and an associated flow $\varphi_h^t$ on the stratum $(\Reducedxi)_{(H_1)}$. The following fact implies that    the flow $\varphi_h^t$ preserves the refined stratification of $\Reducedxi$. 
\begin{proposition}\label{Prop:tangenttorefinedstrata}
  The Hamiltonian vector field $\Hamiltonian{h}$ generated by $h \in \Cinf(\Bxi)$ is tangent to every refined stratum $(\Reducedxi)_{(H_1),\beta,(H_2)}$.

\end{proposition} 

\begin{proof}
  Since we have identified $T^*G$ with $ \gstar\times G$ (via right translations), the algebra   $C^\infty(T^*G)$ is generated by two types of basic functions:  (1) vectors $X\in \galgebra$, seen as linear functions on $\gstar$ and on $T^*G$ as well; (2) functions $f\in C^\infty(G)$, seen as functions on $T^*G$ which do not change along the $\gstar$-direction. As for the Poisson bracket relations, $\SET{X_1,X_2}$ coincides with $[X_1,X_2]_{\galgebra}$, $\SET{f_1,f_2}$ is simply zero, and $\SET{X,f}$ is $\overrightarrow{X}(f)$. Therefore, the Hamiltonian vector field $\Hamiltonian{X}$ of $X\in \galgebra$ on $T^*G\cong \gstar \times G$ is expressed as follows: \begin{equation}\label{Eqt:HamiltonianXsingle}
\Hamiltonian{X}|_{(x,\gamma)}=(-\mathrm{ad}_X^*(x),\overrightarrow{X}|_{\gamma}),\quad \forall (x,\gamma)\in \gstar\times G.            
         \end{equation}
Here we have identified $ T_{(x,\gamma)}(\gstar\times G)$ with $\gstar\times T_{\gamma}G$.
 One could also consider 
 $$X=\sum_i X_{i_1}\odot X_{i_2} \odot \cdots \odot X_{i_n}\in S^n(\galgebra), \quad(\mbox{the $n$-th symmetric product of }\galgebra)  $$
where every $X_{\cdot}$ is a vector in $\galgebra$. This $X$ is also regarded as a function on $\gstar$, and on $\gstar\times G$ as well. Based on \eqref{Eqt:HamiltonianXsingle}, we can inductively show that
\begin{equation}\label{Eqt:HamiltonianXmulti}
\Hamiltonian{X}|_{(x,\gamma)}=
(-\mathrm{ad}_{x \llcorner X}^*(x),\overrightarrow{x\llcorner X}|_{\gamma}).            
         \end{equation}
         Here, we define 
         $$x\llcorner X:=\sum_i \sum_{j=1}^n \langle X_{i_1},x \rangle \langle X_{i_2},x \rangle   \cdots \widehat{X_{i_j}} \cdots \langle X_{i_n}, x \rangle X_{i_j} \in \galgebra. $$

When $X\in S^n(\galgebra)$ is $G$-invariant, we can regard $X$ as a function on $\gstar/G$, and also on $\Bxi$.  However, $X$ being $G$-invariant implies that 
\begin{equation}\label{Eqt:tempXGinvairantzero}
 [X,Y]:=\sum_i \sum_{j=1}^n  X_{i_1}  \odot X_{i_2}    \cdots \widehat{X_{i_j}} \cdots \odot X_{i_n} \odot [X_{i_j},Y]=0,\quad \forall Y \in \galgebra,
\end{equation}
and hence 
$\mathrm{ad}_{x \llcorner X}^*(x)=0$. Therefore, according to \eqref{Eqt:HamiltonianXmulti}, only the second component,  i.e., the $T_xG$-direction, of $\Hamiltonian{X}|_{(x,\gamma)}$  is nontrivial. The flow generated by $\Hamiltonian{X}$ is given by
\begin{equation}\label{Eqt:flowX} \varphi^t_X(x,\gamma)=(x,(\exp(t ~x\llcorner X))\gamma).
\end{equation}
 In this expression, the element $\exp(t ~x\llcorner X)$ belongs to $G_x$ because we can show the following identity:
 \begin{equation}\label{Eqt:temphere}
  \langle Y, \Adjointaction^*_{\exp(t ~x\llcorner X)}(x) \rangle \equiv \langle Y, x \rangle,\quad \forall Y\in \galgebra, t\in \mathbb{R}.
 \end{equation}

  In fact, the left-hand side of Equation \eqref{Eqt:temphere} can be alternatively computed by:
  \begin{eqnarray*}
    \langle  \Adjointaction_{\exp(-t ~x\llcorner X)}(Y),x \rangle &=& \langle  \exp( -t ~\mathrm{ad}_{x\llcorner X} ) Y,x \rangle \\
      &=& \langle Y, x \rangle+\sum_{j=1}^\infty \frac{(-t)^j}{j!} \langle \mathrm{ad}^j_{x\llcorner X}Y, x \rangle.
  \end{eqnarray*}
  It can be inductively shown that
  \begin{equation}\label{Eqt:tempadtoj}
   \langle \mathrm{ad}^j_{x\llcorner X}Y, x \rangle= {[X,\mathrm{ad}^{j-1}_{x\llcorner X}Y ]} (x),\qquad j\geqslant 1.
   \end{equation}
  Here the right-hand side is interpreted as the value of the function $[X,\mathrm{ad}^{j-1}_{x\llcorner X}Y ]\in S^n{(\galgebra)}\subset C^\infty(\gstar)$ on $x\in \gstar$. However, since we have assumed the $G$-invariant property of $X$, namely Equation \eqref{Eqt:tempXGinvairantzero}, the above Equation \eqref{Eqt:tempadtoj} simply yields zero, and \eqref{Eqt:temphere} is justified.
  
  Finally, from Equation \eqref{Eqt:flowX}  and $\exp(t ~x\llcorner X) \in G_x$, it can be easily seen that   $[(x,\gamma)]\in (\Reducedxi)_{(H_1),\beta,(H_2)}$ implies $\varphi^t_X[(x,\gamma)]=[\varphi^t_X(x,\gamma)]\in (\Reducedxi)_{(H_1),\beta,(H_2)}$ as well. Since every function $h \in \Cinf(\Bxi)$ can be (locally) approximated by  polynomial functions on $ \gstar$, the flow $\varphi^t_h$ satisfies the same property.
   
\end{proof}
   
\subsection{The special case of $G=\SU(2)$}\label{Sec:SU(2)} The sCMS system  arising from the special unitary group  \( G = \SU(2) \) and its Lie algebra \( \mathfrak{g} \cong \mathfrak{g}^* = \su(2) \) can be easily derived.  Here we sketch the construction. Suppose that $\Reducedxi$  consists of pairs $[(x,\gamma)]$ where $x$ is diagonalized as $$x=\imagunit \diag(p,-p),\qquad \mbox{ for some }p\geqslant 0.$$
\begin{itemize}
  \item[\textbf{(I)}]       
If $\xi\in \su(2)$ is simply zero, then $x$ and $\gamma$ commute, and hence $\gamma$ can   also be diagonalized, say $\gamma=\diag(e^{\imagunit \alpha},e^{-\imagunit \alpha})$. Therefore, in the first step,
  the reduced space $\Rxizero$ is stratified by its stabilizers as follows:

         $\bullet$ \textit{The maximum stratum}  $(\Rxizero)_{(\Cartan)}$ corresponds to the orbit type $\Cartan$  which is the Cartan subgroup of $\SU(2)$ consisting of diagonal matrices. Indeed, an element $[(x,\gamma)]$ $\in (\Rxizero)_{(\Cartan)}$ can be either of the form $x=\imagunit \diag(p,-p)$ with   $p> 0$, or $\gamma=\diag(e^{\imagunit \alpha},e^{-\imagunit \alpha})$   with $e^{\imagunit \alpha} \neq \pm 1$.
     
     $\bullet$ \textit{The minimum stratum} $(\Rxizero)_{(\SU(2))}$ corresponds to the trivial orbit type $\SU(2)$, and it consists of two points: $[(0,\diag(1,1))]$ and $[(0,\diag(-1,-1))]$.
   
   The space  $\Bxizero$ (consisting of the conjugacy classes $[x]$)  is stratified by two pieces: 
   $$(\Bxizero)^1=\SET{[x= \imagunit \diag(p,-p)]~|~ p> 0},\qquad \mbox{and } (\Bxizero)^0=\SET{0}.$$
  
  The stratified space $\Pxizero=(\Pxizero)^1\sqcup (\Pxizero)^0  $ (consisting of pairs $[(x,-x)]$) can be identified with $\Bxizero$. So, the refined stratification of $\Rxizero$ is as follows. First, $(\Rxizero)_{(\Cartan)}$ is divided into two pieces:
\begin{eqnarray*}  && (\Rxizero)^2_{(\Cartan)  }:=(\Rxizero) _{(\Cartan)}\cap \psi^{-1}(\Pxizero)^1 \\
&=&\SET{[(x=\imagunit \diag(p,-p),\gamma=\diag(e^{\imagunit \alpha},e^{-\imagunit \alpha}))]|p>0}~(\cong \mathbb{R}^{>0}\times \Sone),\\
&& (\Rxizero)^1_{(\Cartan)} := (\Rxizero) _{(\Cartan)}\cap \psi^{-1}(\Pxizero)^0   
\\&=&\SET{[(x=0,\gamma=\diag(e^{\imagunit \alpha},e^{-\imagunit \alpha}))]~|~e^{\imagunit \alpha} \neq \pm 1  }~(\cong (0,\pi)).
\end{eqnarray*}  
Because   $(\Rxizero)_{(\SU(2))}\cap \psi^{-1}(\Pxizero)^0=(\Rxizero)_{(\SU(2))}$, it remains unchanged. Let us denote the only two points in $(\Rxizero)_{(\SU(2))}$ by $(\Rxizero)^0_{(+)}=[(0,\diag(1,1))]$ and $(\Rxizero)^0_{(-)} =[(0,\diag(-1,-1))]$.

The diagram below shows how these refined strata are linked, and the four underlying  sub-superintegrable systems can be read off directly. 

\begin{equation}
	\label{Diagram:SU2arrows00NEW}
	\begin{tikzcd}
		(\Rxizero)^2_{(\Cartan)  }  \arrow[d,dashed] \arrow[rr,"\psi"] && (\Pxizero)^1      \arrow[d,dashed] \arrow[rr ,"\pi(\cong)"]   && (\Bxizero)^1  \arrow[d,dashed]   \\
		(\Rxizero)^1_{(\Cartan)} \arrow[rr,"\psi" ]\arrow[d,dashed] \arrow[dd,bend right=55, dashed] && (\Pxizero)^0    \arrow[rr,"\pi(\cong)" ]  && (\Bxizero)^0  \\ (\Rxizero)^0_{(+)} \arrow[rru,"\psi"] && &&\\
(\Rxizero)^0_{(-)}. \arrow[rruu,"\psi"]&& &&
	\end{tikzcd}
\end{equation}

Elements in the first row of the sequence constitute  a strong classical superintegrable system. The second row starts from an odd-dimensional object. So the whole system can only be regarded as superintegrable in the sense of Definition \ref{Def:stratifiedsuperintNEW}.

  \item[\textbf{(II)}]  If $\xi\in \su(2)$ is not zero, then we  assume that $\xi=\imagunit \diag(a,-a)$ for some $a>0$. Suppose that $(x,\gamma)$  satisfies   
      \begin{equation}\label{Eqt:tempSU2}
         x-\gamma^{-1} x\gamma=z \sim \xi=\imagunit \diag(a,-a),
      \end{equation} where $x$ has eigenvalues $\imagunit p$ and $-\imagunit p$, for some $p>0$ and $z$ has eigenvalues $\imagunit a$ and $-\imagunit a$. According to Theorem \ref{Thm:2X2Horn}, such a pair $(x,\gamma)$ exists if and only if $p\geqslant \frac{a}{2}$. Thus  $\Bxi $ is stratified into the maximum stratum
      $$(\Bxi)^1=\SET{[x=\imagunit \diag(p,-p)]|p>\frac{a}{2}},
      $$
      and the minimum one $$(\Bxi)^0=\SET{[x=\imagunit \diag(\frac{a}{2},-\frac{a}{2})] }.
      $$

      The stratifications of $\Pxi $ and $\Reducedxi $ are accordingly determined as follows:
      \begin{itemize}
        \item[$\bullet$]  The stratum $ (\Pxi)^1$ ($=\pi^{-1}(\Bxi)^1$) consists of    equivalence classes  $[(x,-\gamma^{-1}x\gamma)]$ where $x  =\imagunit \diag(p,-p)$, $ p>\frac{a}{2}$, and 
 $\gamma^{-1}x\gamma$ is of the form
 $$\gamma^{-1}x\gamma=\begin{pmatrix}
		 \imagunit u    & v    \\
		-v 		    &  -\imagunit u	 
	\end{pmatrix}
 $$  
 where $u$, $v$ are two real numbers, and $v>0$. Since  we have the relations $u^2+v^2=p^2$ and $(p-u)^2+v^2=a^2$, this pair $(u,v)$ is indeed uniquely determined by $p$.
 \item[$\bullet$] The stratum $(\Pxi)^0$ ($=\pi^{-1}(\Bxi)^0$) consists of a single point  $[(x_0,-y_0)] $, where $x_0=\imagunit \diag(\frac{a}{2},-\frac{a}{2})$ and $y_0=\imagunit \diag(-\frac{a}{2}, \frac{a}{2})$.
     \item[$\bullet$] The reduced space $\Reducedxi $ is a genuine smooth manifold and it corresponds to the stabilizer $$H=\SET{\diag(1,1),\diag(-1,-1)} .$$
 \item[$\bullet$] As for the refined stratification of $\Reducedxi $, we have the maximum stratum $(\Reducedxi) ^2=\psi^{-1}(\Pxi) ^1$,    which consists of equivalence classes $[(x,\gamma)]$ of  pairs $(x,\gamma)$ satisfying Equation \eqref{Eqt:tempSU2},   where $x=\imagunit \diag(p,-p)$, $p>\frac{a}{2}$, and every entry $\gamma_{ij} $ of $\gamma$ is nonzero.  We also have the minimum stratum $(\Reducedxi) ^0=\psi^{-1}(\Pxi) ^0$, which  consists of a single point  $[(x_0,\gamma_0)]$     where $x_0=\imagunit \diag(\frac{a}{2},-\frac{a}{2})$   and   $$\gamma_0= \begin{pmatrix}
		 0    & 1    \\
		-1 		    &  0	 
	\end{pmatrix}.$$
\item[$\bullet$]
In summary, we obtain   two sub-superintegrable systems as presented in the following diagram.
\begin{equation}
	\label{Diagram:SU2arrows000NEW}
	\begin{tikzcd}
		(\Reducedxi)^2   \arrow[d,dashed] \arrow[rr,"\psi"] && (\Pxi)^1      \arrow[d,dashed] \arrow[rr ,"\pi"]   && (\Bxi)^1  \arrow[d,dashed]   \\
		(\Reducedxi)^0  \arrow[rr,"\psi" ] && (\Pxi)^0    \arrow[rr,"\pi" ]  && (\Bxi)^0  .  
	\end{tikzcd}
\end{equation}

      \end{itemize}
\end{itemize}

	\section{Spin Calogero-Moser-Sutherland  systems   for   $ \SUthree$}\label{Sec:sCMSRPBSU3}

In this section, we analyze   specific sCMS systems arising from the special unitary group  \( G = \SU(3) \) and its Lie algebra \( \mathfrak{g} \cong \mathfrak{g}^* = \su(3) \).

\begin{notation}
	\begin{itemize}
		\item Let $\Cartan$ be the Cartan subgroup within $\SU(3)$. Specifically, $\Cartan$ consists of diagonal matrices $\diag( e^{\imagunit  a},e^{\imagunit  b},e^{ \imagunit  c})$ satisfying $a+b+c=0 \pmod {2\pi}$. 
		\item Denote by $\cartan$ the Cartan subalgebra of $\su(3)$, which comprises   matrices of the form $\imagunit \cdot \diag(a,b,c)$ with $a+b+c=0$. 
		\item We define $\Hreg$ as the set of regular diagonalized matrices in $\SU(3)$, characterized as the open subset of $\Cartan$ containing matrices with distinct diagonal elements.
		\item Similarly, $\hreg$ denotes the set of regular elements within $\cartan$.
	\end{itemize}
\end{notation}

\subsection{Some facts}

\subsubsection{Stabilizers in $\SUthree$}\label{Sec:stabilizerinSUthree}

For a matrix $r\in \SUthree$ or $\su(3)$, we denote  by \( G_r \subset \SUthree \)   the stabilizer of \( r \), which is the group of special unitary matrices that commute with \( r \). 

\begin{itemize}
  \item[(1)]  For any $x\in \Hreg$ (or $\hreg$), we have   $G_x$   $=\Cartan$.
  \item[(2)]  For any $x\in \Cartan$ (or $\cartan$), if  among the three diagonal elements of $x$,  two of them coincide and they are not equal to the third one, then we have
\begin{equation}\label{Eqt:Kgroup}
	G_x\sim	K:=\Big\{
	\begin{pmatrix}
		{r_{11}}     		&   r_{12}				 & 0  \\
		{r_{21}}	    &  {r_{22}}    & 0  \\
		0					 					    & 0 										 & e^{\imagunit  \alpha}
	\end{pmatrix}\in \SU(3)  \Big\}.
\end{equation}

  \item[(3)]  If the three diagonal elements of $x\in \Cartan$ (or $\cartan$) are all equal, then we have  $G_x$    $=\SUthree$.
\end{itemize}
	
\subsubsection{Subgroups of $\Cartan$}
\begin{notation}
  We shall need the following subgroups of $\Cartan$:
 $$
  H(2,1):=\SET{\diag( e^{\imagunit  \alpha},e^{\imagunit  \alpha},e^{-2\imagunit  \alpha})~|~\alpha\in \mathbb{R}}~(\cong \Sone) ,
  $$
  $$
  H(1,2):=\SET{\diag( e^{-2\imagunit  \alpha},e^{\imagunit  \alpha},e^{ \imagunit  \alpha})~|~\alpha\in \mathbb{R}}~(\cong \Sone) ,
  $$
  and
   $$
  H(3):= \SET{\diag( \omega,\omega,\omega )~|~\omega^3=1}~(\cong \Uthree)  .
  $$
  Here $\Uthree$ denotes the set of $s\in \mathbb{C}$ satisfying $s^3=1$.   Note that $H(3)$ is the center of $\SUthree$, and $H(3)\subset H(2,1)\cap H(1,2)$. Note also that $H(2,1)\sim H(1,2)$. 
\end{notation}
The following two propositions can be proved by direct verification.
\begin{proposition}\label{Prop:firstandsecondmethodg0}
   Let $r\in \SUthree$ or $ \su(3)$ be a matrix. 
   \begin{itemize}
     \item[(1)]If $r$ is diagonal, then 
     $$
   \Cartan  \cap G_r = \Cartan.
     $$ 
     \item[(2)]If $r$   takes any of the following forms or their transposes, where $\neq 0$ denotes a nonzero complex number and   $\star$ represents an arbitrary complex number:
     $$
     \begin{pmatrix}
     	\star     		&   \neq 0				 &   0  \\
     	\star	    &  \star    &   0  \\
     	0					 					    &  		0								 & \star
     \end{pmatrix},\mbox{ } \begin{pmatrix}
     	\star     		& 0   				 & \neq 0  \\
     	 0					 					    &  		\star								 & 0\\
     	\star	    &    0    & \star  \\
     \end{pmatrix},\mbox{ } \begin{pmatrix}
     \star     		& 0   				 &  0  \\
     0					 					    &  	\star	& \neq 0								 \\
      0    &    \star    & \star  \\
     \end{pmatrix},
     $$ then 
     $$
     \Cartan  \cap G_r \sim H(2,1).
     $$
     \item[(3)] If $r$ is not included in the previous cases (1) or (2), then we have
     $$
     \Cartan  \cap G_r = H(3).
     $$ 
   \end{itemize}
\end{proposition}
\begin{proposition}\label{Prop:typeII}
    Let $r\in \SU(3)$ be a special unitary   matrix. 
     \begin{itemize}
     \item[(1)]
     If every row of $r$ has only one nonzero entry (which must be a unimodular number $e^{\imagunit  \alpha}\in \Sone$), then
     $$
     \Cartan \cap r^{-1} \Cartan r= \Cartan.
     $$    
     \item[(2)]If there is a unimodular entry $r_{ij}=e^{\imagunit  \alpha}$  (hence the other four positions in the $i$-th row or the $j$-th column are all zero), and the remaining four positions neither in the $i$-th row nor the $j$-th column are all nonzero complex numbers, for example
         $$
         r=\begin{pmatrix}
     	\neq 0     		&   \neq 0				 &   0  \\
     	0	    &  0    &   e^{\imagunit \alpha}  \\
   \neq  	0					 					    &  		\neq 0								 & 0
     \end{pmatrix},
         $$
                  then 
                  $$
     \Cartan \cap r^{-1} \Cartan r\sim H(2,1).
     $$ 
     \item[(3)] If $r$ is not included in the above two cases (1) or (2) (in other words, none of the entries $r_{ij}$ is unimodular), then every row or column of $r$ has at least two nonzero entries, and we have $$
     \Cartan \cap r^{-1} \Cartan r= H(3).
     $$ 
   \end{itemize}
\end{proposition}

\subsubsection{Classification of $\xi$}

 Standard Lie theory implies that any $\xi \in \su(3)$ can be conjugated into the closed subset $S \subset \Cartan$:
   
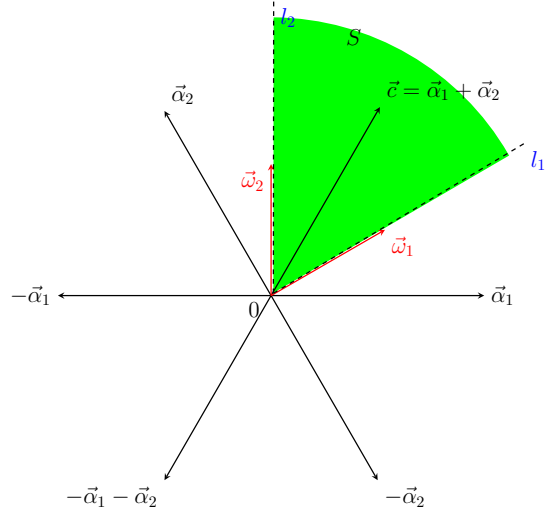
\begin{figure}[htbp]   \centering  
  \scalebox{0.6}{ 
      \begin{tikzpicture}[>=stealth, line width=0.8pt, font=\large]
            \coordinate (O) at (0,0);
       \def\a{60}    
      \def\b{0}     
      \def\c{-60}   
      \def\d{-120}  
      \def\e{180}   
      \def\f{120}   
      \def\lp{30}   
      \def\lm{90}  
      \filldraw[green,line width=2](0.05,0.1)-- (5.2,3.1) arc (30:89.5:6);
      \draw[black] (1.5,5.7) node[right] {$S$};
\draw[black] (-0.1,-0.3) node[left] {0};
      \draw[->] (O) -- ++(\a:4.8) node[above right] {$\vec{c}= \vec{\alpha}_{1}+\vec{\alpha}_{2}$};
      \draw[->] (O) -- ++(\b:4.7) node[right] {$\vec{\alpha}_{1}$};
      \draw[->] (O) -- ++(\c:4.7) node[below right] {$-\vec{\alpha}_{2}$};
      \draw[->] (O) -- ++(\d:4.7) node[below left] {$-\vec{\alpha}_{1}-\vec{\alpha}_{2}$};
      \draw[->] (O) -- ++(\e:4.7) node[left] {$-\vec{\alpha}_{1}$};
      \draw[->] (O) -- ++(\f:4.7) node[above right] {$\vec{\alpha}_{2}$};
      \draw[red,   ->] (O) -- ++(\lp:2.9) node[below right, red] {$\vec{\omega}_1$};
      \draw[red,  ->] (O) -- ++(\lm:2.9) node[below left, red] {$\vec{\omega}_2$};
      \draw[black,dashed] (0.05,0.05)--(89.5:6.5)
      node[below right, blue] {${l_2}$};
      \draw[black,dashed] (0.05,0.05)--(31:6.5)
      node[below right, blue] {${l_1}$};
    \end{tikzpicture}
    }
  \caption{The stratified space $S\cong \su(3)/\SU(3)$\\~} \label{Picture:chamber}~ 
\end{figure}
 
Specifically, in the above figure, the six vectors denoted as $\vec{\alpha}_{1}$, $\vec{\alpha}_{2}$ and their combinations represent the standard roots; for instance, $\vec{\alpha}_{1}=  \diag(1, -1,0)$, $\vec{\alpha}_{2}= \diag(0,1, -1)$, and $$\vec{c}:=\vec{\alpha}_{1}+\vec{\alpha}_{2}=  \diag(1,0,-1) .$$ Denote   by $$\vec{\omega}_1=  \frac{1}{3} \diag(2,-1,-1) ~\mbox{ and }~  \vec{\omega}_2= \frac{1}{3} \diag(1,1,-2)$$ the two weight vectors shown   in the figure.

Let us use the peak point $0$,   the slanted  `wall'  ${l_1}:=\SET{\imagunit d\vec{\omega}_1, ~ d>0}$,   and the vertical `wall' ${l_2}=\SET{\imagunit d\vec{\omega}_2, ~ d>0}$   to enclose the region  $S$ (the green area) as illustrated   in Figure \ref{Picture:chamber}. This $S$ can be identified with the quotient space $\su(3)/\SU(3)$, and it is clearly stratified: 
\begin{equation}\label{Eqt:su3quotientstratified}
\su(3)/\SU(3)\cong S =S_o \sqcup {l_1} \sqcup {l_2} \sqcup \SET{0},
\end{equation}
where $S_o$ denotes the interior,  i.e., the regular part. This stratification   is also determined by the stabilizer relationships. Indeed, since elements in $S_o$ are regular, they are of orbit type $ \Cartan$; elements on the walls ${l_1}\sqcup {l_2} $ have orbit type $K$, while the peak point $0$ is the only point with orbit type $\SU(3)$.

Let $\xi = \imagunit \cdot \diag(a, b, c) \in S$, where $a, b, c \in \mathbb{R}$ and satisfy the constraint $a + b + c = 0$. We call $\xi$   \textbf{generic} if it lies in $S_o$ but not in the $\vec{c}=\vec{\alpha}_1+\vec{\alpha}_2$  direction; for such $\xi$,   the values $a$, $b$, and $c$ are distinct and nonzero. 

We analyze the reduced phase space \( \Reducedxi \). The dimension of a stratum in $\Reducedxi$, specifically $(\Reducedxi)_{(H)}$, is generally an even number that depends on the choice of $\xi$ and the orbit type $H$.  
As will be shown in the following Sections \ref{subSec:xi=0} - \ref{SubSec:xiinalpha13},  we exhibit  four distinct cases: when $\xi$ is \textit{zero} (the peak point), when $\xi$ is \textit{generic}, when $\xi$ lies \textit{on the wall} ${l_1}$ (or ${l_2}$), and when $\xi$ is \textit{in the $\vec{\alpha}_1+\vec{\alpha}_2$ direction}.

\subsubsection{The triple $(x,\gamma,z)$}

For $T^*G\cong \su(3)\times \SU(3)$, 
  	the associated moment map $J_{T^*G}: \su(3)\times \SU(3)\to \su(3)$  given by 
		$J_{T^*G}(x, \gamma)=x- {\gamma^{-1}}x\gamma $ (see Equation \eqref{Eqt:JTstarGgeneral}) is surjective (by Proposition \ref{Prop:SUnmomentmapsurjective}). 

Fixing an element $\xi\in \su(3)$,  our focus is on triples \( (x, \gamma, z) \), where \( x \in \su(3) \), \( \gamma \in \SU(3) \), and \( z \in \mathcal{O}_\xi \subset \su(3) \), which satisfy the equation  
\begin{equation}
	\label{Eqt:SU3eqt}
	x - {\gamma^{-1}}x\gamma = z.
\end{equation}
This equation characterizes the elements of the reduced phase space \( \Reducedxi \) through equivalence classes \( [(x, \gamma)] \). Additionally, the corresponding equivalence classes \( [(x, -{\gamma^{-1}}x\gamma)] \) and \( [x] \) are associated with the spaces \( \Pxi \) and \( \Bxi  \), respectively.

\begin{remark}   
     Consider    the following functions on $\su(3)$:
  $$
  h_1(x)=\frac{1}{2}\mathbf{tr}(x^2),\quad h_2(x)=\frac{1}{3}\mathbf{tr}(x^3),\quad \forall x\in \su(3).
  $$
  Since they are $\SU(3)$-invariant, we treat them as functions on $\Bxi $, and on  $\Reducedxi$ by pulling back via $\pi\circ \psi$. In fact, $h_1$ and $h_2$ together form a coordinate system on $\Bxi $ \cite{KKS,NR2003}.
\end{remark}

We have more facts about the triple $(x,\gamma,z)$ satisfying Equation \eqref{Eqt:SU3eqt}. Let us suppose that the matrix $\gamma$ is diagonalized such that $$\gamma = \diag(\gamma_1, \gamma_2, \gamma_3)$$ with $\gamma_i \in \Sone$, and $\gamma_1 \gamma_2 \gamma_3 = 1$.
\begin{itemize}
  \item The diagonal elements of $z=x - {\gamma^{-1}}x\gamma$ must necessarily vanish, i.e.,  $z$   belongs to the following set
   \begin{equation}
	\label{Notation:Oxi03}
	\Oxi^0:=\SET{ z\in \Oxi, z_{11}=z_{22}= z_{33}=0}.\end{equation}

      Furthermore, for all indices $1 \leqslant i, j \leqslant 3$ with $i \neq j$, the relationship between the off-diagonal elements of $x$ and $z$ is governed by the equation:  
\begin{equation}\label{Eqt:SU3eqtxijzij}
	(1 - \gamma_i^{-1} \gamma_j) x_{ij} = z_{ij}.
\end{equation}
  \item Since \( z \) is conjugate to \( \xi\), say $\xi =\imagunit\cdot \diag(a , b  , c ) $ (where $a$, $b$, $c$ are real numbers satisfying $a+b+c=0$), the off-diagonal elements \( z_{ij} \) (with \( i \neq j \)) must adhere to some specific conditions, derived from the associated characteristic polynomials. Firstly, the sum of products of these elements fulfills the equation $$z_{12}z_{21} + z_{13}z_{31} + z_{23}z_{32} = ab + bc + ca .$$ This implies that the modulus squared of these elements satisfies:
\begin{equation}\label{Eq:ModulusEntries(y)}
	\abs{z_{12}}^2 + \abs{z_{13}}^2 + \abs{z_{23}}^2 = a^2 + ab + b^2.
\end{equation}

Secondly, the product of these elements in a cyclic manner results in $$ z_{12}z_{23}z_{31} + z_{13}z_{32}z_{21} = -\imagunit abc  .$$ Consequently, the imaginary part of the product \( z_{12}z_{23}\overline{z_{13}} \) is expressed as:
\begin{equation}
	\label{Eq:ImaginaryEntries(y)}
	\Im{z_{12}z_{23}\overline{z_{13}}} = \frac{1}{2} abc.
\end{equation}

These conditions reflect the intricate relationship between the elements of \( z \) and the parameters \( a, b, \) and \( c \).
 
\end{itemize}

\subsection{For zero $\xi$}\label{subSec:xi=0}

The associated reduced spaces $\Rxizero $, $\Pxizero $, and $\Bxizero $ are all stratified spaces.  In the subsequent subsections, we will provide a detailed description of the strata and the arrows that connect them (see Diagrams \eqref{Diagram:SU3arrows000NEW}, \eqref{Diagram:SU3arrows000refined}).

First, we assert that the primary stratification of $\Rxizero $ has three distinct pieces: $(\Rxizero)^4_{(\Cartan)}$, $(\Rxizero)^2_{(K)}$, and $(\Rxizero)^0_{(\SU(3))}$.  
In fact, since   $z\in \mathcal{O}_{\xi=0} = \SET{0}$,      we only  need to consider pairs $[(x, \gamma)]\in \Rxizero $ such that $x\gamma = \gamma x$. Consequently, we can assume that \emph{$x$ and $\gamma$ are simultaneously diagonalized, } i.e., $x\in \cartan$ and $\gamma\in \Cartan$. Thus, the orbit type of a stratum in which  $[(x, \gamma)]$ lies, namely the stabilizer $G_{ x,\gamma }=G_x\cap G_{\gamma}$, can be determined by the results in Section \ref{Sec:stabilizerinSUthree}.

 \begin{itemize}
	\item[(1)] The maximum component, denoted as $(\Rxizero)_{(\Cartan)}$, is characterized by the orbit type $ \Cartan$. This stratum contains pairs $[(x,\gamma)]$, where either $x$ or $\gamma$ is regular. It also includes  configurations such as $x = \imagunit v \vec{\omega}_2$ ($\in {l_2}$) with      $v> 0$, and $\gamma = \diag(\gamma_1^{-2},\gamma_1,\gamma_1)$ ($\in H(1,2)\setminus H(3)$) with $\gamma_1\in \Sone$ and $\gamma_1^3\neq 1$ (i.e., $\gamma_1 \neq \gamma_1^{-2}$). A similar arrangement with $x\in {l_1}$ and $\gamma  \in H(2,1)\setminus H(3)$ is included as well.

	\item[(2)] 
	The middle component, denoted as $(\Rxizero)_{(K)}$, exhibits the orbit type $K$ (see Equation \eqref{Eqt:Kgroup}). 
	Within this stratum, pairs $[(x,\gamma)]$ can take the form $x = \imagunit v \vec{\omega}_2$  ($\in {l_2}\sqcup \SET{0}$) with $v\geqslant 0$, paired with $\gamma = \diag(\gamma_1,\gamma_1,\gamma_1^{-2})$ ($\in H(2,1)\setminus H(3)$), under the condition $\gamma_1^3 \neq 1$. Alternatively, it is possible that $x = \imagunit v \vec{\omega}_2 $ $\in {l_2}$  where $v > 0$, and $\gamma = \diag(\gamma_1,\gamma_1,\gamma_1^{-2})$ ($\in H(2,1)$), although $\gamma_1$ may satisfy    $\gamma_1^{3}=1$, among other similar configurations.
	
	\item[(3)] 
	The minimum component, denoted as $(\Rxizero)_{(\SU(3))}$, is characterized by the orbit type $\SUthree$. This stratum (being isomorphic to $\Uthree $) consists solely of three elements, each of the form $(x=0,\gamma )$ where $\gamma\in H(3)$.
	
\end{itemize}
Second, since $\Pxizero $  consists of elements of the form
$[(x,-{\gamma^{-1}}x\gamma )]=[(x,-x)]$ whose   stabilizer reads $G_{x,{\gamma^{-1}}x\gamma }=G_{x,x}=G_x$, and $\Bxizero $ consists of $[x]$, we have
   $$\Pxizero \cong \Bxizero \cong \su(3)/\SU(3)  \cong S, $$ which undergoes the same stratification \eqref{Eqt:su3quotientstratified}. For notational convenience, let us use the following notation  (and similar ones for $\Pxizero $) to denote the strata of $\Bxizero $:

\begin{itemize}
	\item[(1)] The maximum stratum   $(\Bxizero)^2$ ($\cong S_o$)  comprises elements of the form $[x]$, where $x$ is a regular element. This notation should be distinguished from $\Bxizero $, which represents the entire stratified space; in the latter, the superscript $2$ emphasizes that the maximum stratum possesses dimension $2$.

	\item[(2)] The middle stratum, $(\Bxizero)^1  $,   includes elements $[x]$ where $x$ is singular but nonzero. Note that $(\Bxizero)^1  $ has two connected components --- $(\Bxizero)^1_{(1)}  $ ($\cong {l_1}$) and $(\Bxizero)^1_{(2)}  $ ($\cong {l_2}$),   consisting of $[x]$ where $x\in {l_1}$  and   $ {l_2}$, respectively.
	
	\item[(3)] The minimum stratum, $(\Bxizero)^0   $,    consists solely of the point $[0]$.
\end{itemize}

 The maps $\psi: \Rxizero  \to \Pxizero $ and $\pi: \Pxizero  \to \Bxizero $ are decomposed as shown in the following diagram:

\begin{equation}
	\label{Diagram:SU3arrows000NEW}
	\begin{tikzcd}
		(\Rxizero)_{(\Cartan)}\arrow[rrddd] \arrow[d,dashed] \arrow[rrd] \arrow[rrdd] \arrow[rr] && (\Pxizero)^2    \arrow[dd,bend left=55,dashed]\arrow[d,dashed]  \arrow[rr,"\cong" ]   && (\Bxizero)^2  \arrow[d,dashed] \arrow[dd,bend left=55,dashed] \\
		(\Rxizero)_{(K)} \arrow[rrdd] \arrow[dd,dashed]  \arrow[rrd]\arrow[rr ] && (\Pxizero)^1_{(1)}    \arrow[dd,bend left=55,dashed] \arrow[rr,"\cong" ]  && (\Bxizero)^1_{(1)}   \arrow[dd,bend left=55,dashed]\\  ~
		&& (\Pxizero)^1_{(2)}    \arrow[d,dashed] \arrow[rr,"\cong"] && (\Bxizero)^1_{(2)}  \,\arrow[d,dashed] \\
(\Rxizero)_{(\SU(3))}  \arrow[rr]
		&& (\Pxizero)^0    \arrow[rr,"\cong"] && (\Bxizero)^0   \,.
	\end{tikzcd}
\end{equation}


The solid arrows stand for   components of the maps $\psi$ and $\pi$. The dashed arrows represent boundary relations of the relevant strata. 
 It is important to observe that, originating from $(\Rxizero)_{(\Cartan)}$, there are four arrows representing parts of $\psi$. These parts of the maps, however, are not   universally defined across the entire stratum $(\Rxizero)_{(\Cartan)}$. A similar  phenomenon occurs   for $\psi$ starting from $(\Rxizero)_{(K)}$, and is also observed  in subsequent analysis. For our purposes  of finding the  underlying sub-superintegrable systems, we cannot rely on the primary stratification of $\Rxizero$. Instead, we  should use the \emph{refined}   strata of $\Rxizero $, namely the inverse images of relevant strata as we have explained in Appendix \ref{Apd:refinebyinverseimage} and Section    \ref{Sec:refineReducedxi}. Specifically, we find the following subspaces consisting of $[(x,\gamma)]\in  \Rxizero$:
\begin{eqnarray*}
  (\Rxizero)^4_{(\Cartan)} &:=&   (\Rxizero)_{(\Cartan)}\cap \psi^{-1} (\Pxizero)^2   = \SET{[(x,\gamma)]| x\in S_o},\\
  (\Rxizero)^3_{(\Cartan),(1)} &:=&   (\Rxizero)_{(\Cartan)}\cap \psi^{-1} (\Pxizero)^1_{(1)}   \\
  &&= \SET{[(x,\gamma)]| x\in {l_1},\gamma \in \Hreg \cup H(2,1)\setminus H(3) },\\
    (\Rxizero)^3_{(\Cartan),(2)} &:=&   (\Rxizero)_{(\Cartan)}\cap \psi^{-1} (\Pxizero)^1_{(2)}   \\
  &&= \SET{[(x,\gamma)]| x\in {l_2},\gamma \in \Hreg \cup H(1,2) \setminus H(3) },\\
   (\Rxizero)^2_{(\Cartan)} &:=&   (\Rxizero)_{(\Cartan)}\cap \psi^{-1} (\Pxizero)^0  = \SET{[(0,\gamma)]| \gamma\in \Hreg},\\
 (\Rxizero)^2_{(K),(1)} &:=&   (\Rxizero)_{(K)} \cap \psi^{-1} (\Pxizero)^1_{(1)}   = \SET{[(x,\gamma)]|x\in {l_1}, \gamma\in H(1,2)},\\
 (\Rxizero)^2_{(K),(2)} &:=&   (\Rxizero)_{(K)}\cap \psi^{-1} (\Pxizero)^1_{(2)}   = \SET{[(x,\gamma)]|x\in {l_2}, \gamma\in H(2,1)},\\
 (\Rxizero)^1_{(K)} &:=&   (\Rxizero)_{(K)}\cap \psi^{-1} (\Pxizero)^0  = \SET{[(0,\gamma)]| \gamma\in H(2,1)\setminus H(3)}, 
\end{eqnarray*} 
and together with the original minimum stratum $(\Rxizero)_{(\SU(3))}$ ($=(\Rxizero)_{(\SU(3))}\cap \psi^{-1}(\Pxizero)^0  $), $\Rxizero $ is refined into eight strata. Note that among them $(\Rxizero)^3_{(\Cartan),(1)}$,   $(\Rxizero)^3_{(\Cartan),(2)}$,   and $(\Rxizero)^1_{(K)}$ are presymplectic manifolds. As a summary, one finds a total of eight (refined) strata of $\Rxizero $. Their relations with $\Pxizero $ ( $\cong \Bxizero $) and the stabilizers (orbit types) are listed as follows:

       \begin{table}[H]\centering
 \caption{The refined strata of  $\Rxizero $    for  $\xi=0$ } \label{Table:strataxizero}
 \begin{tabularx}{\textwidth}{|>{\raggedright\arraybackslash}X|p{2cm}  |>{\raggedright\arraybackslash}X|p{3cm}|>{\raggedright\arraybackslash}X|>{\raggedright\arraybackslash}X|>{\raggedright\arraybackslash}X|} 
 \hline
 &
 Stratum of    $\Rxizero  = \SET{[(x,\gamma)]}$   &  Orbit type of $G_{x, \gamma}$   & Range of $[x]$ (stratum of $\Bxizero $) & Range of $\gamma$	   &  Stratum of $\Pxizero $ ($\cong \Bxizero $)   & Orbit type of $G_{x,\gamma^{-1}x\gamma}$ ($=G_{x}$)  \\  
  \hline  No.1 ~& 
  $(\Rxizero)^4_{(\Cartan)}$ &	$\Cartan$ & $S_o$ (=$(\Bxizero)^2 $) & $\Cartan$     & $(\Pxizero)^2 $  & $\Cartan$ \\ 
  \hline  No.2 ~& $(\Rxizero)^3_{(\Cartan),(1)}$
   & $\Cartan$ &	$l_1$ (=$(\Bxizero)^1_{(1)}  $) &  $\Hreg$ $\cup$ $H(2,1)\setminus H(3)$    &  $(\Pxizero)^1_{(1)}   $ &  $K$  \\\hline
   No.3 ~& $(\Rxizero)^3_{(\Cartan),(2)}$
   & $\Cartan$ &	$l_2$ (=$ (\Bxizero)^1_{(2)}  $) &  $\Hreg$ $\cup$ $H(1,2)\setminus H(3)$    &  $(\Pxizero)^1_{(2)}   $ &  $K$  \\
   \hline
   No.4 ~& $(\Rxizero)^2_{(\Cartan)}$
   & $\Cartan$ &	$\SET{0}$ (=$ (\Bxizero)^0   $) &  $\Hreg  $    &  $(\Pxizero)^0 $ &  $\SU(3)$ \\
\hline  No.5 ~&
  $(\Rxizero)^2_{(K),(1)}$ & $K$ &	$l_1$ (=$(\Bxizero)^1_{(1)}  $) & $H(1,2)$    &  $(\Pxizero)^1_{(1)}    $ &  $K$  \\
  \hline  No.6 ~&
  $(\Rxizero)^2_{(K),(2)}$ & $K$ &	$l_2$ (=$(\Bxizero)^1_{(2)}  $) & $H(2,1)$    &  $(\Pxizero)^1_{(2)}    $ &  $K$  \\
  \hline  No.7 ~& 
  $(\Rxizero)^1_{(K)}$ & $K$ &	$\SET{0}$ (=$(\Bxizero)^0   $) & $H(2,1)\setminus H(3)$    &  $(\Pxizero)^0  $ &  $\SU(3)$  \\\hline  No.8 ~& 
  $(\Rxizero)_{(\SU(3))}$ & $\SU(3)$ &	$\SET{0}$   (=$(\Bxizero)^0   $) & $H(3) $    &  $(\Pxizero)^0  $ &  $\SU(3)$  \\  
  \hline     
    \end{tabularx}
   \end{table}

The following diagram shows how the relevant strata are linked by $\psi$ and $\pi$:

\begin{equation}
	\label{Diagram:SU3arrows000refined}
	\begin{tikzcd}
		(\Rxizero)^4_{(\Cartan)} \arrow[dd,bend right=55,dashed]\arrow[d,dashed] \arrow[rr] && (\Pxizero)^2    \arrow[dd,bend left=55,dashed]\arrow[d,dashed]  \arrow[rr,"\cong" ]   && (\Bxizero)^2  \arrow[d,dashed] \arrow[dd,bend left=55,dashed] \\
		(\Rxizero)^3_{(\Cartan),(1)}  
 \arrow[ddd,bend right=55,dashed] 
\arrow[dd,bend right=55,dashed]\arrow[rr ] && (\Pxizero)^1_{(1)}    \arrow[dd,bend left=55,dashed] \arrow[rr,"\cong" ]  && (\Bxizero)^1_{(1)}   \arrow[dd,bend left=55,dashed]\\  ~ (\Rxizero)^3_{(\Cartan),(2)} \arrow[ddd,bend right=55,dashed] \arrow[d, dashed]
\arrow[rr]		&& (\Pxizero)^1_{(2)}    \arrow[d,dashed] \arrow[rr,"\cong"] && (\Bxizero)^1_{(2)}  \,\arrow[d,dashed] \\
(\Rxizero)^2_{(\Cartan)} \arrow[ddd,bend right=55,dashed]   \arrow[rr]
		&& (\Pxizero)^0    \arrow[rr,"\cong"] && (\Bxizero)^0   \,\\
(\Rxizero)^2_{(K),(1)} \arrow[rruuu] \arrow[dd,bend right=55,dashed] &&  &&\\
(\Rxizero)^2_{(K),(2)} \arrow[rruuu] \arrow[d,dashed] &&  &&  \\
(\Rxizero)^1_{(K)} \arrow[rruuu] \arrow[d,dashed] &&  &&  \\
(\Rxizero)_{(\SU(3))}\,. \arrow[rruuuu] &&  &&  
	\end{tikzcd}
\end{equation} 
 As illustrated in the diagram, each sequence originating from a   space on the left corresponds to a sub-superintegrable system. The verification of this property involves only routine computations (e.g.,  verifying Poisson bracket relations) and is left to the reader.

\subsection{For generic   $\xi $}\label{subSec:xiabciallnonzero}

Suppose now that $\xi$ is of the generic form $\imagunit\cdot \diag(a  , b  , c  )$, where none of $a$, $b$, and $c$ is zero, and they are \textit{distinct}. Without loss of generality, we assume that $a>b>c$. Since $a+b+c=0$, we have $a>0$ and $c<0$. Through straightforward analysis, it can be determined that $\dim(\Oxi^0) = 4$.

For any solution $(x, \gamma, z)$ to Equation \eqref{Eqt:SU3eqt}, we have the following two observations: \begin{itemize}
	\item[--]   \textbf{Regularity of $\gamma$}:  The matrix $\gamma$ is regular. Specifically, when $\gamma$ is diagonalized as $\gamma = \diag(\gamma_1, \gamma_2, \gamma_3)$, the eigenvalues $\gamma_1$, $\gamma_2$, and $\gamma_3$ are distinct. This regularity is supported by Equation \eqref{Eq:ImaginaryEntries(y)}, which implies that none of $z_{12}$, $z_{23}$, or $z_{13}$ is zero. Consequently, by Equation \eqref{Eqt:SU3eqtxijzij}, $\gamma$ must be regular. It also ensures that all $x_{12}$, $x_{23}$, and $x_{13}$ are nonzero.
	
	\item[--]   \textbf{Regularity of $x$}:  The matrix $x$ is also regular. Clearly, it cannot be zero. If $x$ were diagonalized into a singular matrix, such as $x =  \imagunit\cdot \diag(p , p   , -2p   )$ ($=3\imagunit p \vec{\omega}_2\in l_2$) with $p > 0$,  it could be expressed as $x = \imagunit p (I_3 - 3L)$, where $L = \diag(0, 0, 1)$. This leads to the expression $$z = x - {\gamma^{-1}}x\gamma = 3p\imagunit  (-L + \gamma^{-1} L \gamma).$$ However, both $L$ and $\gamma^{-1} L \gamma$ are rank $1$ matrices, which contradicts the condition that $z \in \Oxi$, a matrix of rank $3$. The case where $x= \imagunit\cdot \diag(2p , -p   , - p   )$ ($=3\imagunit p \vec{\omega}_1\in l_1$) can be similarly seen to be impossible.
	
\end{itemize}

Consequently, we can assert that $\Reducedxi$ is an ordinary manifold. Indeed, we have just explained the regularity of $\gamma$. So when $\gamma\in \Hreg$,     none of the three numbers $x_{12}$, $x_{23}$, and $x_{13}$ is zero.  According to Proposition \ref{Prop:firstandsecondmethodg0}, the stabilizer group at every $(x,\gamma)$ is given by 
$$
G_{x,\gamma} = G_x \cap G_\gamma = G_x \cap \Cartan = H(3)  .
$$   In conclusion, for the primary stratification of $\Reducedxi$, there is only one stratum, $(\Reducedxi)_{(H(3))}$, which is a genuine smooth manifold. The refined stratification of $\Reducedxi$ is demonstrated in the sequel Section \ref{Sec:genericRxirefined}.

\subsubsection{$\Bxi  $      as a stratified space}  \label{Sec:Bxipolyarea}
 
We have shown that every $x$ in the triple $(x, \gamma, z)$ satisfying Equation \eqref{Eqt:SU3eqt} is regular. Hence   the set of $[x]$, namely $  \Bxi $,  must be a region within $S_o$. In fact, we have a precise description of   $\Bxi \subset S_o$.

\begin{lemma}\label{Lemma:genericBxishape} For generic $\xi= \imagunit  \diag(a  , b  , c  )$  as assumed above, we have
\begin{equation}\label{Eqt:Bxi2uvuplusv}
\Bxi =\SET{\imagunit ( u\vec{\omega}_1+v\vec{\omega}_2)~|~ u\geqslant {\abs{b}} ,v\geqslant {\abs{b}} ,u+v\geqslant \max(a,-c)}.
\end{equation}

\end{lemma}
The shape of $\Bxi $ is illustrated by Figure \ref{Picture:genericBxiinchamber}.

 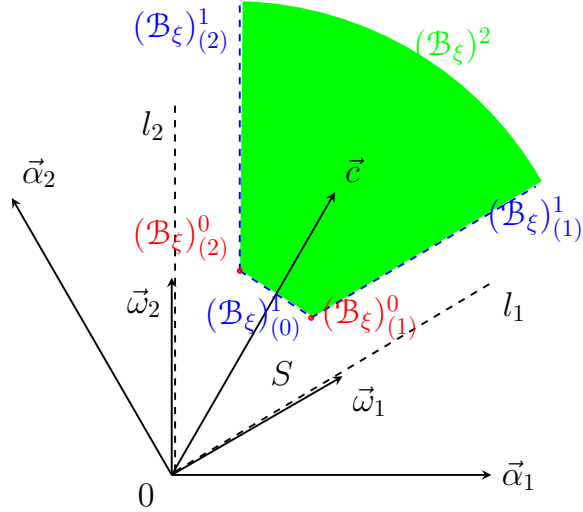
\begin{figure}[htbp]   \centering  
  \scalebox{0.9}{ 
      \begin{tikzpicture}[>=stealth, line width=0.8pt, font=\large]
            \coordinate (O) at (0,0);
       \def\a{60}    
      \def\b{0}     
      \def\c{-60}   
      \def\d{-120}  
      \def\e{180}   
      \def\f{120}   
      \def\lp{30}   
      \def\lm{90}  
    \filldraw[green,line width=2] (2.08,2.37) -- (5.39, 4.33)   arc (30:87.9:5.2) -- (1.04,3.03) -- (2.08,2.37);
         \draw[black] (1.3,1.5) node[right] {$S$};
\draw[black] (3.4,6.3) node[right,green] {$(\Bxi)^2   $};
\draw[black] (-0.1,-0.3) node[left] {0};
      \draw[->] (O) -- ++(\a:4.8) node[above right] {$ \vec{c}$};
      \draw[->] (O) -- ++(\b:4.7) node[right] {$\vec{\alpha}_{1}$};
       \draw[->] (O) -- ++(\f:4.7) node[above right] {$\vec{\alpha}_{2}$};
      \draw[black,   ->] (O) -- ++(\lp:2.9) node[below right, black] {$\vec{\omega}_1$};
      \draw[black,  ->] (O) -- ++(\lm:2.9) node[below left, black] {$\vec{\omega}_2$};
      \draw[black,dashed] (0.05,0.05)--(89.5:5.5)
      node[below left, black] {${l_2}$};
      \draw[black,dashed] (0.05,0.05)--(31:5.5)
      node[below right, black] {${l_1}$}; 
      \draw[blue,dashed] (1,3)--(1,7)
      node[below left, blue] {$(\Bxi)^1_{(2)}  $};
      \draw[blue,dashed] (1,3)--(2.06,2.32);
      \draw[blue] (1.2,2.3) node[blue] {$(\Bxi)^1_{(0)}  $};
            \draw[blue,dashed] (2.05,2.31)--(5.35, 4.24)
      node[below, blue] {$(\Bxi)^1_{(1)}  $};
      \draw[red] (1,3) circle (0.03) node[above left,red] {$(\Bxi)^0_{(2)}$};
\draw[red] (2.05,2.31) circle (0.03) node[right,red] {$(\Bxi)^0_{(1)}$};
        \end{tikzpicture}
        }
  \caption{The stratified space $\Bxi \subset S$ for generic $\xi$; \\ the corner point $(\Bxi)^0_{(1)}  $ is $\max(a,-c) \vec{\omega}_1+ \abs{b} \vec{\omega}_2$, and $(\Bxi)^0_{(2)}  $ is $\abs{b} \vec{\omega}_1+  \max(a,-c) \vec{\omega}_2$.\\ ~} \label{Picture:genericBxiinchamber}~ 
\end{figure}

\begin{proof}[Proof of Lemma \ref{Lemma:genericBxishape}] Write $x=\imagunit \underline{x}$ and $z=\imagunit \underline{z}$. Hence $\underline{x}$ and $\underline{z}$ are traceless Hermitian matrices. Equation \eqref{Eqt:SU3eqt} becomes
  \begin{equation}
	\label{Eqt:Hermitian3eqt}
	\underline{x} - {\gamma^{-1}}\underline{x}\gamma = \underline{z}.
\end{equation}

Suppose that $\underline{x}$ has eigenvalues $p$, $q$, and $r$, ordered by $p> q> r$ and satisfying  $p+q+r=0$. Then we simply assume that $x=\imagunit \underline{x}$ is of the form: \begin{equation}\label{Eqt:xinpqr}\imagunit\cdot \diag(p,q,r)=\imagunit\cdot (u\vec{\omega}_1+ v\vec{\omega}_2)\, \end{equation}
where
\begin{equation}\label{Eqt:uvpqr} u= {p-q} ,\quad v= {q-r} .\quad \mbox{ (Hence  } u+v=p-r.)
\end{equation}

Now $(-\gamma^{-1}\underline{x}\gamma)$ has eigenvalues ordered by  $-r> -q> -p$, and $\underline{z}$ has eigenvalues   ordered by $a> b> c$. According to Equation \eqref{Eqt:Hermitian3eqt} and   Weyl's inequalities (\eqref{Eqt:Weyl1} -- \eqref{Eqt:Weyl6} and \eqref{Eqt:Weyl1moremore} -- \eqref{Eqt:Weyl6moremore}) recalled by Theorem \ref{Thm:3X3Weyl},   we have
the following relations:
\begin{eqnarray}\label{tempineq1}
  a &\leqslant &  p-r, \\\label{tempineq2}
  b &\leqslant & \min(p-q,q-r) ,\\\label{tempineq3}
  \max(q-p,r-q) &\leqslant & b,\\\label{tempineq4}
  r-p &\leqslant& c.
\end{eqnarray}
The inequalities \eqref{tempineq1} together with \eqref{tempineq4} are combined to yield 
\begin{equation}\label{Eqt:p-r}
   p-r\geqslant \max(a,-c),
\end{equation}
and \eqref{tempineq2} together with \eqref{tempineq3} are merged into\begin{equation}\label{Eqt:p-qq-r} \min(p-q,q-r)\geqslant \abs{b}.\end{equation}
This proves that $\Bxi $ sits within the set described on the right-hand side of Equation \eqref{Eqt:Bxi2uvuplusv}.  

To see the converse inclusion, let $\underline{x}':=u\vec{\omega}_1+ v\vec{\omega}_2$ satisfy  $u\geqslant {\abs{b}}$, $v\geqslant {\abs{b}}$, and $u+v\geqslant \max(a,-c)$. We wish to show that $[\imagunit  \underline{x}']\in \Bxi $. Let us rewrite $$\underline{x}' =u\vec{\omega}_1+ v\vec{\omega}_2=\diag(p,q,r),$$
where $$p=\frac{2u+v}{3},\quad q=\frac{-u+v}{3},\quad r=\frac{-u-2v}{3}.$$ 
The inequalities \eqref{Eqt:p-r} and \eqref{Eqt:p-qq-r} hold true.

We resort to Horn's Theorem \ref{Thm:3X3Horn}, which is recalled in Appendix \ref{appendixSec:Horn}. In fact, if we set $(\mu_1,\mu_2,\mu_3)=(p,q,r)$, $(\nu_1,\nu_2,\nu_3)=(-r,-q,-p)$, and $(\lambda_1,\lambda_2,\lambda_3)=(a,b,c) $, then all conditions of Theorem \ref{Thm:3X3Horn} are satisfied because \eqref{Eqt:p-r} and \eqref{Eqt:p-qq-r} are true. So there exist Hermitian matrices $\underline{x}$ with eigenvalues $p> q> r $, $\underline{y}$ with eigenvalues $ -r>  -q>  -p $, and $\underline{z}$ with eigenvalues $ a> b> c $, such that $\underline{z}=\underline{x}+\underline{y}$. We have $\imagunit\underline{z}\in \Oxi$ and hence $[\imagunit \underline{x}']=[\imagunit \underline{x}] $ belongs to $\Bxi $. 
\end{proof}

Following the description of $\Bxi $ in \eqref{Eqt:Bxi2uvuplusv}, for generic $\xi$, we decompose  $\Bxi $ into six pieces to become a stratified space:
$$
\Bxi =(\Bxi)^2    \sqcup (\Bxi)^1_{(0)}  \sqcup (\Bxi)^1_{(1)}  \sqcup (\Bxi)^1_{(2)}  \sqcup
(\Bxi)^0_{(1)}  \sqcup(\Bxi)^0_{(2)}  ,
$$
as illustrated in Figure \ref{Picture:genericBxiinchamber}. In particular, if $a>b>0$, $c=-a-b$, then the two corner points are $(\Bxi)^0_{(1)}  =a\vec{\omega}_1+b \vec{\omega}_2$ and $(\Bxi)^0_{(2)}  =b\vec{\omega}_1+a \vec{\omega}_2$.

 \subsubsection{Classification of elements in $\Reducedxi   $ and $\Pxi $}
\label{Sec:genericRxirefined}

In the following analysis, let us assume that $a>b>0>c$, and hence $\max(a,-c)=-c$. The case where $a>0>b>c$ can be analogously analyzed and the results are similar.

Our notation for a pair $(x,\gamma)$  always refers to a solution to Equation \eqref{Eqt:SU3eqt} for some $z\in \Oxi$, where  $x\in\su(3)$ must be regular, as established in Section \ref{subSec:xiabciallnonzero}. So we will always suppose that $x$ is of the form $$x =\imagunit \cdot \diag(p, q, r)=\imagunit\cdot (u\vec{\omega}_1+ v\vec{\omega}_2) ,$$ where $p>q>r$, $p + q + r = 0$,  $u= {p-q}$, and $v= {q-r}$, according to \eqref{Eqt:uvpqr}.

We claim that   $\gamma$ \emph{cannot have any unimodular diagonal element}, i.e., $\gamma_{ii}\neq e^{\imagunit \alpha}$.  Otherwise, one of $a$, $b$, $c$   must be   zero. One can also see this fact by Proposition \ref{Prop:firstandsecondmethodg0} and by   $
G_{x,\gamma} = G_x \cap G_\gamma =\Cartan \cap  G_\gamma = H(3)$.  

Now we determine all possible forms of    $\gamma$ and the corresponding $x$.
 
\begin{flalign}~  \mbox{\it The case where   three  entries of  $\gamma$ are unimodular.} &&\label{CASE-A} 
    \tag{A}\end{flalign}
\begin{itemize}
  \item [] 
\begin{flalign}~\mbox{\it We suppose that  $\gamma$ is of    the following form:}&&\label{CASE-A-1} 
    \tag{A-1}\end{flalign} 
   \begin{equation*}
    \gamma=
  \begin{pmatrix}
		0 & e^{\imagunit  \alpha_1} & 0 \\
		0 & 0 & e^{\imagunit  \alpha_2} \\
		e^{\imagunit  \alpha_3} & 0 & 0
	\end{pmatrix}.\end{equation*} By proper conjugation to $\gamma$ induced by some  element in $\Cartan$, which preserves $x$, it becomes
  $$\gamma=\begin{pmatrix}
		0 & 1 & 0 \\
		0 & 0 & 1 \\
		1 & 0 & 0
	\end{pmatrix}.$$
Therefore, it suffices to  examine $(x,\gamma)$ where $\gamma$ is as above. We then have
$$
x-\gamma^{-1}x\gamma=z=\imagunit \cdot \diag(p-r ,  q-p, r-q).
$$
Since $q-p$ and $r-q$ are both negative, the above $z$ does not exist because we have assumed $a>b>0>c$. So we have ruled out the possible   existence of this kind of $\gamma$.
\item[] 
\begin{flalign}~\mbox{\it We suppose the other possibility of $\gamma$: }&&\label{CASE-A-2} 
    \tag{A-2}\end{flalign} 
   \begin{equation}
   \label{Eqt:modulusonegamma}
  \begin{pmatrix}
		0 & 0 & e^{\imagunit  \alpha_1} \\
		e^{\imagunit  \alpha_2} & 0 & 0 \\
		0 & e^{\imagunit  \alpha_3} & 0
	\end{pmatrix}, 
	\end{equation} which, after $\Cartan$-conjugation, can be simplified to \begin{equation}\label{Eqt:connergamma}
\gamma=\begin{pmatrix}
		0 & 0 & 1 \\
		1 & 0 & 0 \\
		0 & 1 & 0
	\end{pmatrix}.\end{equation}
    We then have
$$
x-\gamma^{-1}x\gamma=z=\imagunit \cdot \diag(p-q ,  q-r, r-p).
$$
    As $a>b>0>c$, we must have $r-p=c$,
     $p-q=a$ and $q-r=b$, or
     $p-q=b$ and $q-r=a$.
    The corresponding two solutions are
$$
p_1=\frac{a-c}{3},q_1=\frac{b-a}{3},r_1=\frac{c-b}{3},
\mbox{ and } 
p_2 =\frac{b-c}{3},q_2=\frac{a-b}{3},r_2=\frac{c-a}{3}. 
$$ 
We then obtain    two solutions of $x$:
\begin{equation}\label{Eqt:endpointsx1x2}
   x_1=\imagunit (a\vec{\omega}_1+ b\vec{\omega}_2),\mbox{ and } x_2=\imagunit (b\vec{\omega}_1+a\vec{\omega}_2).
\end{equation}
These two points correspond to   $(\Bxi)^0_{(1)}  $ and $(\Bxi)^0_{(2)}  $, respectively, marked   in Figure \ref{Picture:genericBxiinchamber}.
So we have found the only two points $[(x_1,\gamma)]$ and $[(x_2,\gamma)]$ in $\Reducedxi  $ with $\gamma$   given by \eqref{Eqt:connergamma}.

One can also easily verify the converse fact --- \emph{if $[(x_1,\gamma')]$ (or $[(x_2,\gamma')]$)is a point in $\Reducedxi  $, then $\gamma'$ must be of the form   \eqref{Eqt:modulusonegamma}, which is $\Cartan$-conjugate to $\gamma$ in \eqref{Eqt:connergamma}.
}

For the associated point $\psi([(x_1,\gamma)])=[(x_1,\gamma^{-1}x_1\gamma)]\in \Pxi $, the stabilizer reads
$$G_{x_1, {\gamma^{-1}}x_1\gamma } = G_{x_1} \cap \gamma^{-1}G_{x_1}\gamma=\Cartan \cap \gamma^{-1}\Cartan\gamma=\Cartan,
$$
according to Proposition \ref{Prop:typeII}-(1). The stabilizer at $\psi([(x_1,\gamma)])$ is also $\Cartan$.
\end{itemize}

\begin{flalign}~\mbox{\it There is only one unimodular entry of  $\gamma$.}&&\label{CASE-B} 
    \tag{B}\end{flalign}

  In this case, $\gamma$ takes any of the following forms or their transposes:
	$$
	\begin{pmatrix}
		0     		&     0				 &   e^{\imagunit  \alpha} \\
		   \neq 0    & \neq 0 & 0  \\
							 	\neq 0		    &  		\neq 0								 & 0
	\end{pmatrix},\mbox{ } \begin{pmatrix}
		0     		&   \neq 0				 & \neq 0  \\
		e^{\imagunit  \alpha}					 					    &  		0								 & 0\\
		0	    &  \neq 0    & \neq 0  \\
	\end{pmatrix}, \mbox{ } 
	\begin{pmatrix}
		\neq 0 &	0     		 			 & \neq 0  \\
		\neq 0 &	0     		 			 & \neq 0  \\
		0&	e^{\imagunit  \alpha}					 					    &  		0								  
	\end{pmatrix}.$$
For all solutions $x$ and $\gamma$ discussed here, the corresponding points $[(x,-\gamma^{-1}x\gamma)]\in \Pxi $ have the same stabilizer group:
$$G_{x, {\gamma^{-1}}x \gamma } = G_{x } \cap \gamma^{-1}G_{x }\gamma=\Cartan \cap \gamma^{-1}\Cartan\gamma\sim H(2,1),
$$
according to Proposition \ref{Prop:typeII}-(2).

By $\Cartan$-conjugation, we can also assume that the $e^{\imagunit  \alpha}$ entry in $\gamma$ is exactly $1$. Next,   let us discuss  all possible types of $\gamma$ one by one.  

\begin{itemize}
\item[] 
\begin{flalign}~\mbox{\it Suppose that}&&\label{CASE-B-1} 
    \tag{B-1}\end{flalign}   
\begin{equation}\label{Eqt:gammaB1}\gamma=\begin{pmatrix}
		0     		&     0				 &   1 \\
		   \gamma_{21}    & \gamma_{22} & 0  \\
		\gamma_{31} 		    &  \gamma_{32}	 								 & 0
	\end{pmatrix}
\end{equation} where every $\gamma_{ij}$ is \emph{nonzero} and 
$$\gamma_0=\begin{pmatrix}
		 \gamma_{21}    & \gamma_{22}    \\
		\gamma_{31} 		    &  \gamma_{32}	 
	\end{pmatrix}\in \SUtwo.
$$ 
A direct  computation shows that
$$
 x-\gamma^{-1}x\gamma=z=\begin{pmatrix}
		\bullet     		&    \bullet				 &   0 \\
		   \bullet    & \bullet & 0  \\
		0 		    &  0	 								 & \imagunit(r-p)
	\end{pmatrix}.
$$
Since $r-p$ is negative, we must have $r-p=c$. Also, the   upper left corner is indeed the   $2\times 2$-matrix (up to a multiplication of $\imagunit$):
$$
  \diag(p,q)-  \gamma_0^{-1}  \diag(q,r)\gamma_0\,
$$ 
which must have eigenvalues $a>b$ (both being positive). The existence of such solutions $p,q,r$, and $\gamma_0$ is guaranteed by Theorem \ref{Thm:2X2Horn}. 

Now, we have $p-r=-c$. By Equations \eqref{Eqt:xinpqr} and \eqref{Eqt:uvpqr}, any solution $x=\imagunit\cdot \diag(p,q,r)=\imagunit\cdot (u\vec{\omega}_1+ v\vec{\omega}_2)$ satisfies $u+v=-c$. In other words, this kind of $x$ lives within the   line between $(\Bxi)^0_{(1)}  $ and $(\Bxi)^0_{(2)}  $. This line segment is    marked by $(\Bxi)^1_{ (0)}$  in Figure \ref{Picture:genericBxiinchamber}. According to our discussions in the earlier Case \eqref{CASE-A-2},    the endpoints $x_1$ and $x_2$ can only be solved from $\gamma$ of the form \eqref{Eqt:modulusonegamma}. So we exclude these two endpoints when we require $\gamma$ to be of the form \eqref{Eqt:gammaB1}.

     A converse fact can be verified as well ---  \emph{if $[(x,\gamma)]$ is a point in $\Reducedxi  $ where $x$ belongs to the boundary $(\Bxi)^1_{(0)} $, then $\gamma$ must be $\Cartan$-conjugate to the form \eqref{Eqt:gammaB1}. }

  \item[] 
\begin{flalign}~\mbox{\it Suppose that}&&\label{CASE-B-2} 
    \tag{B-2}\end{flalign} 
 \begin{equation}\label{Eqt:gammaB2}\gamma=\begin{pmatrix}
		0     		&     \gamma_{12}				 &   \gamma_{13} \\
		     0  & \gamma_{22} & \gamma_{23}  \\
	1 		    &  0 	 								 & 0
	\end{pmatrix}
\end{equation} where every $\gamma_{ij}$ is \textit{nonzero} and 
$$\gamma_0=\begin{pmatrix}
		 \gamma_{12}    & \gamma_{13}    \\
		\gamma_{22} 		    &  \gamma_{23}	 
	\end{pmatrix}\in \SUtwo.
$$ 

We can directly compute that
$$
 x-\gamma^{-1}x\gamma=z=\begin{pmatrix}
 \imagunit(p-r) & 0 & 0\\0 &
		\bullet     		&    \bullet\\				 0 &
		   \bullet    & \bullet   
	\end{pmatrix}.
$$
Since $p-r$ is positive, we must have $p-r=a$ or $b$. However, from Equation \eqref{Eqt:uvpqr} we have $u+v=p-r$ and from Equation \eqref{Eqt:Bxi2uvuplusv} we have $u+v\geqslant -c$, which is strictly greater than $a$ and $b$. This contradiction implies that $\gamma$ of the form \eqref{Eqt:gammaB2} does not exist.  

  \item[] 
\begin{flalign}~\mbox{\it Suppose that}&&\label{CASE-B-3} 
    \tag{B-3}\end{flalign}  
\begin{equation}\label{Eqt:gammaB3}\gamma=\begin{pmatrix}
		0     		&     \gamma_{12}				 &   \gamma_{13} \\-1 		    &  0 	 								 & 0\\
		     0  & \gamma_{32} & \gamma_{33}	
	\end{pmatrix}
\end{equation} where every $\gamma_{ij}$ is \textit{nonzero} and 
$$\gamma_0= \begin{pmatrix}
		 \gamma_{12}    & \gamma_{13}    \\
		\gamma_{32} 		    &  \gamma_{33}	 
	\end{pmatrix}\in \SUtwo.
$$ 

Then we can directly compute that
\begin{equation}\label{Eqt:temptemp1}
 x-\gamma^{-1}x\gamma=z=\begin{pmatrix}
 \imagunit(p-q) & 0 & 0\\0 &
		\bullet     		&    \bullet\\				 0 &
		   \bullet    & \bullet   
	\end{pmatrix}.
\end{equation}
Since $p-q$ is positive, we either have $u=p-q=a$,  or $u=b$. Equation \eqref{Eqt:temptemp1} also implies the following relation:
$$
\diag(q,r)-\gamma_0^{-1}\diag(p,r)\gamma_0  \sim \diag(b,c) \mbox{ (or } \diag(a,c)~). 
$$
In fact, the existence of such $\gamma_0$ is verified by Theorem \ref{Thm:2X2Horn} and the assumption $p-q=u=a$ (or $b$) and $p-r=u+v> -c$.

In summary, we conclude that there are two types of solutions $x$ to $x-\gamma^{-1}x\gamma\in \Oxi$    (where $\gamma$ is of the form \eqref{Eqt:gammaB3}) ---  \begin{equation}\label{Eqt:xinC2}x=\imagunit\cdot (a\vec{\omega}_1+ v\vec{\omega}_2), \quad v> b, \end{equation}
and 
\begin{equation}\label{Eqt:xB1mid2}x=\imagunit\cdot (b\vec{\omega}_1+ v\vec{\omega}_2), \quad v> a. \end{equation}
The collection of such $[x]$ is illustrated by the lines $\mathcal{C}_{(2)}$ and $(\Bxi)^1_{(2)} $, respectively, in the following Figure \ref{Picture:genericBxiinchamberwith2morelines}.

A converse fact can be verified ---  \emph{if $[(x,\gamma)]$ is a point in $\Reducedxi  $ where $x$ belongs to the boundary $(\Bxi)^1_{(2)} $ (i.e., $x$ is of the form \eqref{Eqt:xB1mid2}), then $\gamma$ must be $\Cartan$-conjugate to the form \eqref{Eqt:gammaB3}. }

 \begin{figure}[htbp]   \centering  
  \scalebox{0.9}{ 
      \begin{tikzpicture}[>=stealth, line width=0.8pt, font=\large]
            \coordinate (O) at (0,0);
       \def\a{60.3}    
      \def\b{0}     
      \def\c{-60}   
      \def\d{-120}  
      \def\e{180}   
      \def\f{120}   
      \def\lp{30}   
      \def\lm{90}  
    \filldraw[green,line width=2] (2.08,2.37) -- (5.39, 4.33)   arc (30:87.9:5.2) -- (1.04,3.03) -- (2.08,2.37);
         \draw[black] (1.3,1.5) node[right] {$S$};
\draw[black] (3.4,6.3) node[right,green] {$(\Bxi)^2   $};
\draw[black] (-0.1,-0.3) node[left] {0};
      \draw[->] (O) -- ++(\a:4.8) node[above right] {$ \vec{c}$};
      \draw[->] (O) -- ++(\b:4.7) node[right] {$\vec{\alpha}_{1}$};
       \draw[->] (O) -- ++(\f:4.7) node[above right] {$\vec{\alpha}_{2}$};
      \draw[black,   ->] (O) -- ++(\lp:2.9) node[below right, black] {$\vec{\omega}_1$};
      \draw[black,  ->] (O) -- ++(\lm:2.9) node[below left, black] {$\vec{\omega}_2$};
      \draw[black,dashed] (0.05,0.05)--(89.5:5.5)
      node[below left, black] {${l_2}$};
      \draw[black,dashed] (0.05,0.05)--(31:5.5)
      node[below right, black] {${l_1}$}; 
      \draw[blue,dashed] (1,3)--(1,7)
      node[below left, blue] {$(\Bxi)^1_{(2)}  $};
      \draw[blue,dashed]  (2.06,2.32)--(2.06,7) node[blue,right]{$\mathcal{C}_{(2)}$};
      \draw[blue,dashed] (1,3)--(2.06,2.32);
      \draw[blue] (1.2,2.3) node[blue] {$(\Bxi)^1_{(0)}  $};
            \draw[blue,dashed] (2.05,2.31)--(5.35, 4.24)
      node[below, blue] {$(\Bxi)^1_{(1)}  $};
\draw[blue,dashed] (1,2.98)--(5.24, 5.50)
      node[below, blue] {$\mathcal{C}_{(1)}$};

      \draw[red] (1,3) circle (0.03) node[above left,red] {$(\Bxi)^0_{(2)}$};
\draw[red] (2.05,2.31) circle (0.03) node[right,red] {$(\Bxi)^0_{(1)}$};
        \end{tikzpicture}
        }
  \caption{The two lines $\mathcal{C}_1$ and $\mathcal{C}_2$  in $\Bxi \subset S$ for generic $\xi$\\~} \label{Picture:genericBxiinchamberwith2morelines}~ 
\end{figure}
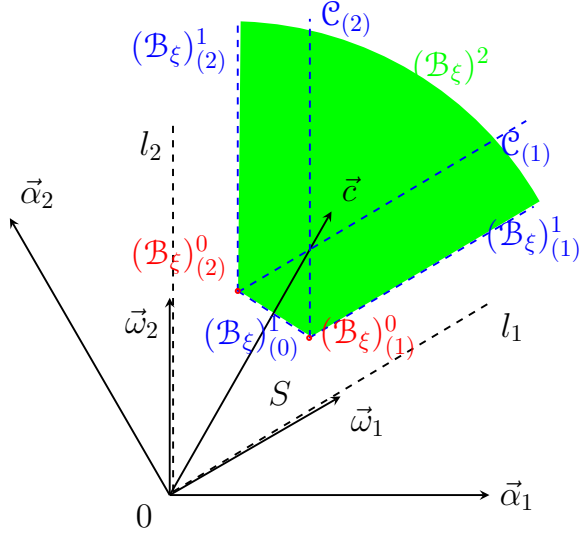

  \item[] 
\begin{flalign}~\mbox{\it Suppose that}&&\label{CASE-B-4} 
    \tag{B-4}\end{flalign}   
\begin{equation}\label{Eqt:B4}\gamma=\begin{pmatrix}
		0     		&     -1				 &   0 \\\gamma_{21} 		    &  0 	 								 & \gamma_{23}\\
		       \gamma_{31} & 0& \gamma_{33}	
	\end{pmatrix}
\end{equation} where every $\gamma_{ij}$ is \textit{nonzero} and 
$$\gamma_0=\begin{pmatrix}
		 \gamma_{21}    & \gamma_{23}    \\
		\gamma_{31} 		    &  \gamma_{33}	 
	\end{pmatrix}\in \SUtwo.
$$ 

Then we can directly compute that
\begin{equation}\label{Eqt:temptemp2}
 x-\gamma^{-1}x\gamma=z=\begin{pmatrix}
\bullet  & 0 & \bullet\\0 &
		     		   \imagunit(q-p)& 0\\				 \bullet &
		   0    & \bullet   
	\end{pmatrix}.
\end{equation}
Since $q-p$ is negative, we must have $q-p=c$. Meanwhile, Equation \eqref{Eqt:temptemp2} implies that
$$
\diag(p,r)-\gamma_0^{-1}\diag(q,r)\gamma_0  \sim \diag(a,b)  . 
$$

According to Theorem \ref{Thm:2X2Horn}, this relation requires 
$p-q=-c\leqslant a$, which is impossible.  So we exclude $\gamma$ of the form \eqref{Eqt:B4}.

  \item[] 
\begin{flalign}~\mbox{\it Suppose that}&&\label{CASE-B-5} 
    \tag{B-5}\end{flalign}   
\begin{equation}\label{Eqt:gammaB5}\gamma=\begin{pmatrix}
		\gamma_{11}     		&     0				 &   \gamma_{13} \\\gamma_{21} 		    &  0 	 								 & \gamma_{23}\\
		       0 & -1& 0	
	\end{pmatrix}
\end{equation} where every $\gamma_{ij}$ is \textit{nonzero} and 
$$\gamma_0=\begin{pmatrix}
		 \gamma_{11}    & \gamma_{13}    \\
		\gamma_{21} 		    &  \gamma_{23}	 
	\end{pmatrix}\in \SUtwo.
$$ 
In fact, this situation can be analyzed analogously to the situation of  Case \eqref{CASE-B-3}, and the conclusion is as follows:   all solutions $x$ to $x-\gamma^{-1}x\gamma\in \Oxi$    (where $\gamma$ is of the form \eqref{Eqt:gammaB5}) are of the form \begin{equation}\label{Eqt:xB1mid1} x=\imagunit\cdot (u\vec{\omega}_1+ b\vec{\omega}_2), \quad u> a, 
\end{equation}
or 
\begin{equation}\label{Eqt:xinC1} x=\imagunit\cdot (u\vec{\omega}_1+ a\vec{\omega}_2), \quad u> b. \end{equation}
The set of  such $[x]$  consists of  two lines   $(\Bxi)^1_{(1)} $ and $\mathcal{C}_{(1)}$, as shown by   Figure \ref{Picture:genericBxiinchamberwith2morelines}.

Moreover, a converse fact can be verified ---  \emph{if $[(x,\gamma)]$ is a point in $\Reducedxi  $ where $x$ belongs to the boundary $(\Bxi)^1_{(1)} $ (i.e., $x$ is of the form \eqref{Eqt:xB1mid1}), then $\gamma$ must be $\Cartan$-conjugate to the form \eqref{Eqt:gammaB5}. }

\item[] 
\begin{flalign}~\mbox{\it Suppose that}&&\label{CASE-B-6} 
    \tag{B-6}\end{flalign} 
\begin{equation}\label{Eqt:gammaB6}
\gamma=\begin{pmatrix}
		 \gamma_{11}     		&        \gamma_{12}				 &   0 \\
		      0  & 0 & -1  \\
	 \gamma_{31} 		    &   \gamma_{32} 	 								 & 0
	\end{pmatrix}
\end{equation} 
where every $\gamma_{ij}$ is \textit{nonzero} and 
$$\gamma_0=\begin{pmatrix}
		 \gamma_{11}    & \gamma_{12}    \\
		\gamma_{31} 		    &  \gamma_{32}	 
	\end{pmatrix}\in \SUtwo.
$$ 
In fact, this situation can be analyzed analogously to the situation of Case \eqref{CASE-B-4}, and the conclusion is that $\gamma$ of the form \eqref{Eqt:gammaB6} does not exist.
\end{itemize} 
 
\begin{flalign}~\mbox{\it The case where none of the entries of  $\gamma$ is unimodular.}&&\label{CASE-C} 
    \tag{C}\end{flalign}

Suppose that $\gamma\in \SUthree$ is not in  Cases \eqref{CASE-A} and \eqref{CASE-B}. In other words, none of the entries $\gamma_{ij}$ is unimodular. Then any diagonal solution $x$ satisfying $x-\gamma^{-1}x\gamma\in \Oxi$ must be of the form
\begin{equation}\label{Eqt:xinBxiinnerpart}
x=\imagunit\cdot (u\vec{\omega}_1+ v\vec{\omega}_2), \quad \mbox{ where }~u>b, v> b, u+v>-c. \end{equation}
In other words, this kind of  $x$ belongs to the inner part $(\Bxi)^2   $ of $\Bxi $.
  
  The corresponding points   $[(x, -{\gamma^{-1}}x\gamma)] \in \Pxi $ are grouped into $(\Pxi)^4 $.

\subsubsection{Stratification of  $\Reducedxi   $ and $\Pxi $} 
  We apply the method of refined stratification in Sections \ref{subsubSec:RPBxi} and \ref{Sec:refineReducedxi}.  For all   $[(x,\gamma)]\in \Reducedxi  $, 
  we already have $G_{x,  \gamma } = H(3)$. To find $G_{x, {\gamma^{-1}}x\gamma}$, we can use    Proposition \ref{Prop:typeII}-(3). Therefore,  based on the stratification of $\Bxi $ in Section \ref{Sec:Bxipolyarea} and the classification cases of $\gamma$ in   Section  \ref{Sec:genericRxirefined},      
     what we obtain is a  decomposition of $\Reducedxi  $    into eight pieces as demonstrated in the following table,  and similarly for $\Pxi $.

       \begin{table}[H]\centering
 \caption{The refined strata of  $\Reducedxi  $   for generic $\xi$ } \label{Table:stratagenericxi} 
       \begin{supertabular}{|p{0.9cm}|p{2.3cm}|p{2.3cm}|p{2.1cm}|  p{2.9cm} | p{2.4cm}   |  }
 \hline 
& Stratum of   $\Reducedxi  $ $=$ $\SET{[(x,\gamma)]}$   & Range of $[x]$ & Range of $\gamma$	   &  Stratum of   $\Pxi $ $=$ $\SET{[(x,-\gamma^{-1}x\gamma)]}$  & Orbit type of $G_{x,\gamma^{-1}x\gamma}$ \\  
  \hline No.1~&
  $(\Reducedxi)^6  $ 	& $ (\Bxi)^2   $ & Case \eqref{CASE-C}     & $(\Pxi)^4  $  & $H(3)$ \\ 
  \hline No.2~&
  $(\Reducedxi)^2_{(0)}$ &	$ (\Bxi)^1_{(0)}  $ & Case \eqref{CASE-B-1}     &  $(\Pxi)^1_{(0)}$ &  $H(2,1)$  \\
\hline No.3~&
  $(\Reducedxi)^2_{(1)}$ &	$ (\Bxi)^1_{(1)}  $ & Case \eqref{CASE-B-5}     &  $(\Pxi)^1_{(1)}$ &  $H(2,1)$  \\
  \hline No.4~&
  $(\Reducedxi)^2_{(2)}$ &	$ (\Bxi)^1_{(2)}  $ & Case \eqref{CASE-B-3}     &  $(\Pxi)^1_{(2)}$ &  $H(2,1)$  \\  
  \hline No.5~&
  $(\Reducedxi)^2_{(1')}$ &	$\mathcal{C}_{(1)}\subset (\Bxi)^2   $ & Case \eqref{CASE-B-5}     &  $(\Pxi)^1_{(1')}$ & $H(2,1)$   \\
   \hline No.6~&
  $(\Reducedxi)^2_{(2')}$ &	$\mathcal{C}_{(2)}\subset (\Bxi)^2   $ & Case \eqref{CASE-B-3}     &  $(\Pxi)^1_{(2')}$ & $H(2,1)$   \\
   \hline No.7~&
  $(\Reducedxi)^0_{(1)}$ &	$(\Bxi)^0_{(1)}  $ & Case \eqref{CASE-A-2}   &  $(\Pxi)^0_{(1)}$ &  $\Cartan$   \\
  \hline No.8~&
  $(\Reducedxi)^0_{(2)}$ &	$(\Bxi)^0_{(2)}  $ & Case \eqref{CASE-A-2}   &  $(\Pxi)^0_{(2)}$ &  $\Cartan$   \\
   \hline
   \end{supertabular}

   \end{table}  
Note that the strata $(\Reducedxi)^2_{(1')}$ and $(\Reducedxi)^2_{(2')}$ (Nos.5 and 6) are the two connected components of the stratum which is denoted by $(\Reducedxi)_{(H(3)),(\Bxi)^2   ,(H(2,1))}$
 (in Section \ref{Sec:refineReducedxi}).

 Finally, we   draw a diagram to show the roles of these strata and the decomposition of maps $\psi$ and $\pi$.

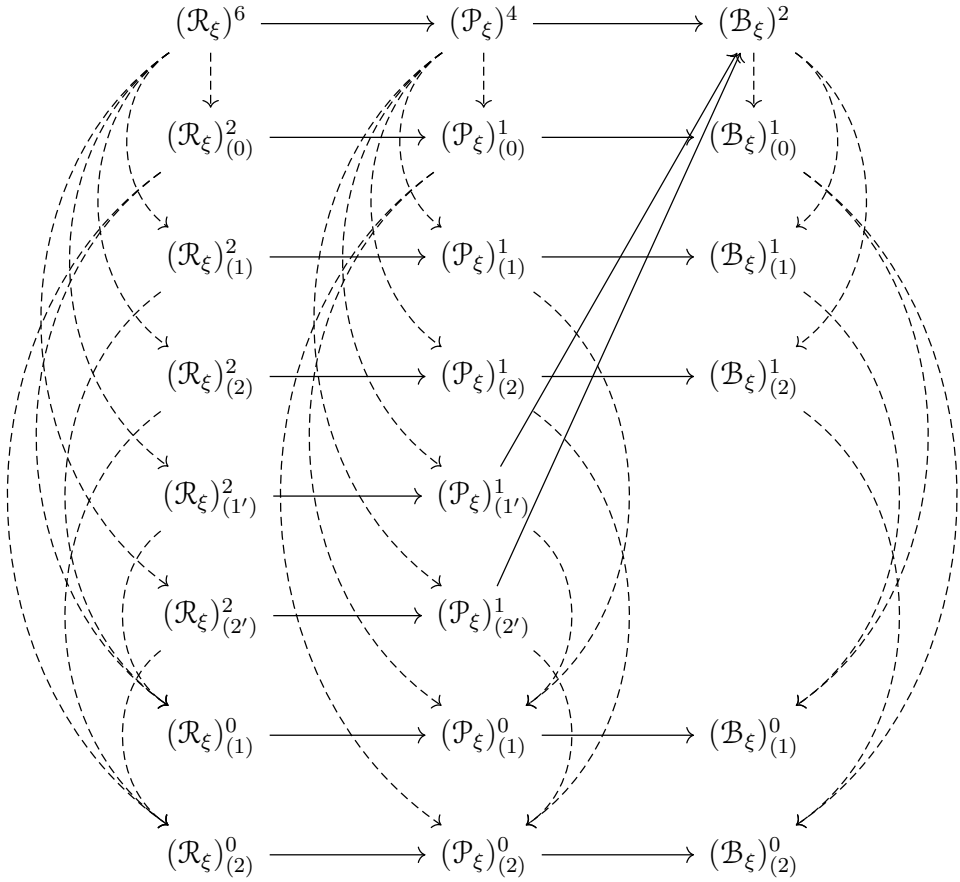
\begin{equation}
	\label{Diagram:SU3arrowsabcrefined}
	\begin{tikzcd}
		(\Reducedxi)^6  \arrow[ddddd,bend right=55,dashed] \arrow[dddd,bend right=55,dashed]\arrow[ddd,bend right=55,dashed]\arrow[dd,bend right=55,dashed]\arrow[d,dashed] \arrow[rr] && (\Pxi)^4   \arrow[ddddd,bend right=55,dashed] \arrow[dddd,bend right=55,dashed]\arrow[ddd,bend right=55,dashed]\arrow[dd,bend right=55,dashed]\arrow[d,dashed] \arrow[rr]   && (\Bxi)^2    \arrow[d,dashed] \arrow[dd,bend left=55,dashed]\arrow[ddd,bend left=55,dashed] \\
		(\Reducedxi)^2_{(0)}  
 \arrow[ddddd,bend right=55,dashed] 
\arrow[dddddd,bend right=55,dashed]\arrow[rr ] && (\Pxi)^1_{(0)} \arrow[ddddd,bend right=55,dashed] 
\arrow[dddddd,bend right=55,dashed] \arrow[rr ]&& (\Bxi)^1_{(0)}   \arrow[ddddd,bend left=55,dashed]\arrow[dddddd,bend left=55,dashed]\\  
		(\Reducedxi)^2_{(1)}  
 \arrow[dddd,bend right=55,dashed] \arrow[rr ] && (\Pxi)^1_{(1)} \arrow[dddd,bend left=55,dashed] \arrow[rr ]  && (\Bxi)^1_{(1)}   \arrow[dddd,bend left=55,dashed]\\  
		(\Reducedxi)^2_{(2)}  
 \arrow[dddd,bend right=55,dashed] 
 \arrow[rr ] && (\Pxi)^1_{(2)} \arrow[dddd,bend left=55,dashed] \arrow[rr ]  && (\Bxi)^1_{(2)}   \arrow[dddd,bend left=55,dashed]\\
       (\Reducedxi)^2_{(1')}  
 \arrow[dd,bend right=55,dashed] \arrow[rr ] && (\Pxi)^1_{(1')} \arrow[dd,bend left=55,dashed] \arrow[rruuuu ]  &&  ~\\  
       (\Reducedxi)^2_{(2')}  
 \arrow[dd,bend right=55,dashed]  \arrow[rr ] && (\Pxi)^1_{(2')} \arrow[dd,bend left=55,dashed] \arrow[rruuuuu ]  && ~\\  
		(\Reducedxi)^0_{(1)}   \arrow[rr ] && (\Pxi)^0_{(1)}  \arrow[rr ]  && (\Bxi)^0_{(1)}    \\
		(\Reducedxi)^0_{(2)}  \arrow[rr ] && (\Pxi)^0_{(2)}  \arrow[rr ]  && (\Bxi)^0_{(2)}     
	\end{tikzcd}
\end{equation}
 In this diagram, there are 8 sub-superintegrable systems, each   sequence starting from a left-hand stratum of $\Reducedxi  $ and ending in a right-hand stratum of $\Bxi $.

\subsection{For   $\xi  $   on the wall}\label{subSec:rankone} 
When $\xi$ lies on either the ${l_1}$ or ${l_2}$ wall, the analysis proceeds analogously to that presented in the preceding section. Consequently, the regularity properties of $\gamma$ and $x$ remain valid as previously established. We have $\dim(\Oxi^0) = 2$. The reduced space $\Reducedxi $ is a $4$-dimensional smooth manifold (which will be refined to find sub-superintegrable systems).  The orbit type of $ \Reducedxi$ is again $H(3)$.

The space $\Bxi $ continues to be a   $2$-dimensional stratified space   described as follows. 

\begin{lemma}\label{Lemma:wallBxishape} Given   $\xi$ on the wall  $l_1$ or $l_2$, i.e., of the form $\imagunit d \vec{\omega}_1=\imagunit  \diag(2d  , -d  , -d  )$ or $ \imagunit d \vec{\omega}_2=\imagunit  \diag( d  ,  d  , -2d  )$ (where $d>0$), we have
\begin{equation}\label{Eqt:wallBxi2uvuplusv}
\Bxi =\SET{\imagunit ( u\vec{\omega}_1+v\vec{\omega}_2)~|~ u\geqslant d ,v\geqslant d ,u+v\geqslant 2d}.
\end{equation}
As a stratified space, we have 
$$
\Bxi =(\Bxi)^2   \sqcup (\Bxi)^1_{(1)}  \sqcup (\Bxi)^1_{(2)}  \sqcup (\Bxi)^0 .
$$
\end{lemma}
The shape of $\Bxi $ is illustrated by Figure \ref{Picture:wallBxiinchamber}. The vertex point $(\Bxi)^0 $ is $d(\vec{\omega}_1+\vec{\omega}_2)$.

 \begin{figure}[htbp]   \centering  
  \scalebox{0.9}{ 
      \begin{tikzpicture}[>=stealth, line width=0.8pt, font=\large]
            \coordinate (O) at (0,0);
       \def\a{60}    
      \def\b{0}     
      \def\c{-60}   
      \def\d{-120}  
      \def\e{180}   
      \def\f{120}   
      \def\lp{30}   
      \def\lm{90}  
    \filldraw[green,line width=2] (1.03,1.81) -- (5.39, 4.33)   arc (30:88.4:5.2) ;
         \draw[black] (1.9,1.8) node[right] {$S$};
\draw[black] (3.4,6.3) node[right,green] {$(\Bxi)^2   $};
\draw[black] (-0.1,-0.3) node[left] {0};
      \draw[->] (O) -- ++(\a:4.8) node[above right] {$ \vec{c}$};
      \draw[->] (O) -- ++(\b:4.7) node[right] {$\vec{\alpha}_{1}$};
       \draw[->] (O) -- ++(\f:4.7) node[above right] {$\vec{\alpha}_{2}$};
      \draw[black,   ->] (O) -- ++(\lp:2.9) node[below right, black] {$\vec{\omega}_1$};
      \draw[black,  ->] (O) -- ++(\lm:2.9) node[below left, black] {$\vec{\omega}_2$};
      \draw[black,dashed] (0.05,0.05)--(89.5:5.5)
      node[below left, black] {${l_2}$};
      \draw[black,dashed] (0.05,0.05)--(31:5.5)
      node[below right, black] {${l_1}$}; 
      \draw[blue,dashed] (1,1.777)--(1,7)
      node[below left, blue] {$(\Bxi)^1_{(2)}  $};
             \draw[blue,dashed] (1,1.745)--(5.35, 4.24)
      node[below, blue] {$(\Bxi)^1_{(1)}  $};
      \draw[red] (1,1.763) circle (0.03) node[below ,red] {$(\Bxi)^0 $};
         \end{tikzpicture}
        }
  \caption{The stratified space $\Bxi \subset S$ for  $\xi$ on the wall\\~} \label{Picture:wallBxiinchamber}~ 
\end{figure}
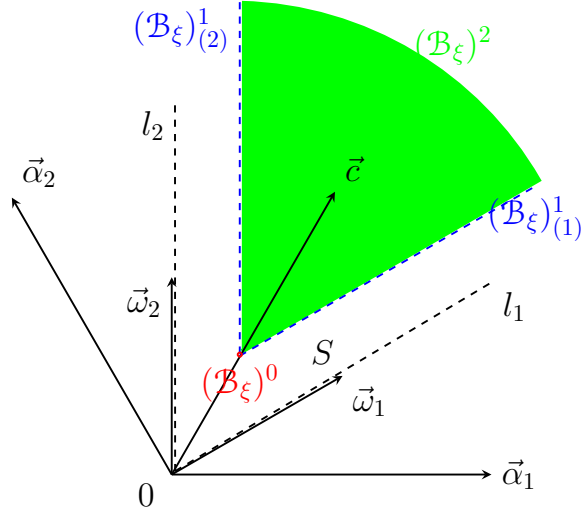

{The stratification of  $\Reducedxi  $ and $\Pxi $ can be considered as a limiting case of the previous section.  The eight pieces of $\Reducedxi  $ for generic $\xi$ in Table \ref{Table:stratagenericxi} now merge into four pieces and we obtain four   sub-superintegrable systems (as demonstrated in the following table and diagram).
   
      \begin{table}[H]
 \caption{The refined strata of  $\Reducedxi $   for   $\xi$ on the wall } \label{Table:stratawallxi} 
  \begin{supertabular}{|p{0.9cm}|p{2.4cm}|p{2.4cm}|p{2.1cm}|  p{3cm} | p{2.5cm}   |  }
 \hline
 &
 Stratum of   $\Reducedxi $ $=$ $\SET{[(x,\gamma)]}$   & Range of $[x]$ & Range of $\gamma$	   &  Stratum of   $\Pxi $ $=$ $\SET{[(x,-\gamma^{-1}x\gamma)]}$  &   Orbit~ type~of $G_{x,\gamma^{-1}x\gamma}$ \\  
 
 \hline No.1 ~& 
  $(\Reducedxi)^4$ 	& $ (\Bxi)^2   $ & Case \eqref{CASE-C}     & $(\Pxi)^2$  & $H(3)$ \\ 
\hline  No.2 ~& 
  $(\Reducedxi)^2_{(1)}$ &	$ (\Bxi)^1_{(1)}  $ & Case \eqref{CASE-B-5}     &  $(\Pxi)^1_{(1)}$ &  $H(2,1)$  \\
  \hline  No.3 ~& 
  $(\Reducedxi)^2_{(2)}$ &	$ (\Bxi)^1_{(2)}  $ & Case \eqref{CASE-B-3}     &  $(\Pxi)^1_{(2)}$ &  $H(2,1)$  \\  
  \hline  No.4 ~& 
  $(\Reducedxi)^0$   &	$(\Bxi)^0$ & Cases \eqref{CASE-B-1}   or   \eqref{CASE-A-2}   &  $(\Pxi)^0$ &  $\Cartan$   \\
  \hline
   \end{supertabular}
\end{table}

\begin{equation}
	\label{Diagram:SU3arrowsonwallrefined}
	\begin{tikzcd}
		(\Reducedxi)^4\arrow[dd,bend right=55,dashed] \arrow[d,dashed] \arrow[rr] && (\Pxi)^2  \arrow[dd,bend right=55,dashed]\arrow[d,dashed] \arrow[rr]   && (\Bxi)^2    \arrow[d,dashed] \arrow[dd,bend left=55,dashed]  \\
		(\Reducedxi)^2_{(1)}  
\arrow[dd,bend right=55,dashed]\arrow[rr ] && (\Pxi)^1_{(1)} \arrow[dd,bend right=55,dashed]  \arrow[rr ]&& (\Bxi)^1_{(1)}   \arrow[dd,bend left=55,dashed] \\  
		(\Reducedxi)^2_{(2)}  
 \arrow[d, dashed] \arrow[rr ] && (\Pxi)^1_{(2)} \arrow[d, dashed] \arrow[rr ]  && (\Bxi)^1_{(2)}   \arrow[d, dashed]\\  
		(\Reducedxi)^0   
 \arrow[rr ] && (\Pxi)^0  \arrow[rr ]  && (\Bxi)^0      
	\end{tikzcd}
\end{equation}

We also refer to \cite{MGRarxiv2506.16610v2} for  more studies of this particular stratified space.

\subsection{For $\xi $   in  the   $\vec{c}$ direction}\label{SubSec:xiinalpha13} 
This is indeed \emph{the most complicated case}. Let us fix $\xi =\imagunit a\vec{c}=\imagunit\cdot \diag(a $, $0$, $-a  )$, where $a$ is a positive constant. A straightforward analysis reveals that $\dim(\Oxi^0) = 4$.  
The associated reduced spaces $\Reducedxi  $, $\Pxi $, and $\Bxi $ are all stratified spaces.  
In the subsequent subsections, we will provide a detailed description of these strata and the arrows that connect them (see Diagram \ref{Diagram:SU3arrowscdirectionrefined}).

\begin{lemma}\label{Lemma:veccBxishape} For   $\xi =\imagunit a\vec{c}=\imagunit\cdot \diag(a $, $0$, $-a  )$ with $a>0$, the space $\Bxi$ is given by
\begin{eqnarray}\label{Eqt:veccBxi2uvuplusv}
\Bxi  &=& \SET{\imagunit ( u\vec{\omega}_1+v\vec{\omega}_2)~|~ u\geqslant 0 ,v\geqslant 0 ,u+v\geqslant a}
\\\nonumber 
&=&  (\Bxi)^2    \sqcup (\Bxi)^1_{(0)}  \sqcup (\Bxi)^1_{(1)}  \sqcup (\Bxi)^1_{(2)}  \sqcup
(\Bxi)^0_{(1)}  \sqcup(\Bxi)^0_{(2)}  .
\end{eqnarray}

\end{lemma}
The shape of $\Bxi $ is illustrated by Figure \ref{Picture:veccBxiinchamber}. The two corners are $(\Bxi)^0_{(1)}  =a\vec{\omega}_1$, and $(\Bxi)^0_{(2)}  =a\vec{\omega}_2$, respectively.

 \begin{figure}[htbp]   \centering  
  \scalebox{0.9}{ 
      \begin{tikzpicture}[>=stealth, line width=0.8pt, font=\large]
            \coordinate (O) at (0,0);
       \def\a{60}    
      \def\b{0}     
      \def\c{-60}   
      \def\d{-120}  
      \def\e{180}   
      \def\f{120}   
      \def\lp{30}   
      \def\lm{90}  
    \filldraw[green,line width=2] (3.28,1.98) -- (5.95, 3.33)   arc (30:89.2:6.9) -- (0.05,4.0) ;
         \draw[black] (1.1,1.5) node[right] {$S$};
\draw[black] (4.8,4.9) node[right,green] {$(\Bxi)^2   $};
\draw[black] (-0.1,-0.3) node[left] {0};
      \draw[->] (O) -- ++(\a:4.8) node[above right] {$ \vec{c}$};
      \draw[->] (O) -- ++(\b:4.7) node[right] {$\vec{\alpha}_{1}$};
       \draw[->] (O) -- ++(\f:4.7) node[above right] {$\vec{\alpha}_{2}$};
      \draw[black,   ->] (O) -- ++(\lp:2.9) node[below right, black] {$\vec{\omega}_1$};
      \draw[black,  ->] (O) -- ++(\lm:2.9) node[below left, black] {$\vec{\omega}_2$};
       \draw[blue,dashed] (0,4.1)--(0,7)
      node[below left, blue] {$(\Bxi)^1_{(2)}  $};

\draw[blue,dashed] (0,4.0)--(3.1,5.9)
      node[above, blue] {$\mathcal{C}_{(2)}$};
      
      \draw[blue,dashed] (0,4.0)--(3.30,1.93);
      \draw[blue] (1.4,2.5) node[blue] {$(\Bxi)^1_{(0)}  $};
            \draw[blue,dashed] (3.30,1.93)--(5.95, 3.28)
      node[below, blue] {$(\Bxi)^1_{(1)}  $};
      \draw[red] (0,4.0) circle (0.03) node[above left,red] {$(\Bxi)^0_{(2)}$};
\draw[red] (3.30,1.93) circle (0.03) node[right,red] {$(\Bxi)^0_{(1)}$};
\draw[blue,dashed] (3.30,1.93)--(3.3,5.3)
      node[below right, blue] {$\mathcal{C}_{(1)}$};
        \end{tikzpicture}
        }
  \caption{The stratified space $\Bxi \subset S$ for   $\xi$ in the $\vec{c} $ direction\\~} \label{Picture:veccBxiinchamber}~ 
\end{figure}
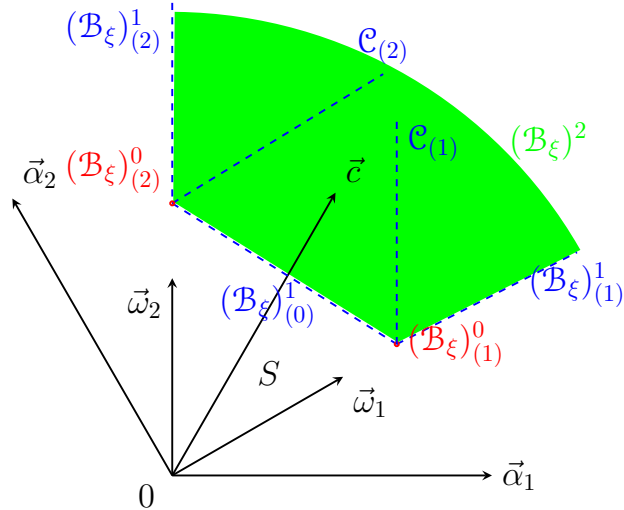

\subsubsection{Stratification of $\Reducedxi  $}

Our notation for a pair $(x,\gamma)$  always refers to a solution to $x-\gamma^{-1}x\gamma=z$ for some $z\in \Oxi$, where  $x $ is diagonalized as usual $x=\imagunit (p,q,r)$ with $p\geqslant q\geqslant r$ and $p+q+r=0$.

If $\xi$ is generic, we have claimed that none of the diagonal elements of $\gamma$ can be unimodular. But when $\xi=\imagunit\cdot \diag(a  ,  0 , -a  )$ is in the $\vec{c}$ direction, $\gamma$ can take this form. Specifically, it can be any of the following (D) and (E) cases up to $\Cartan$-conjugacy:
\begin{itemize}
  \item[] 
\begin{flalign}~\mbox{\it $\gamma$ is of the form}&&\label{CASE-D-1} 
    \tag{D-1} \end{flalign} $$ \begin{pmatrix}
	e^{\imagunit \theta_1}	& 0									  &0 \\0&0&e^{\imagunit \theta_2}
	 	  \\
	0 	  	&   			    -1 & 0
\end{pmatrix},\quad (\theta_1+\theta_2\equiv 2\pi).$$
  \item[] 
\begin{flalign}~\mbox{\it $\gamma$ is of the form}&&\label{CASE-D-2} 
    \tag{D-2}\end{flalign}  $$ \begin{pmatrix}
	0	& 0&			e^{\imagunit \theta_1}						  \\0&e^{\imagunit \theta_2}
	 		& 0									  \\
	  			    1 &0&0
\end{pmatrix},\quad (\theta_1+\theta_2\equiv \pi).$$
  \item[] 
\begin{flalign}~\mbox{\it $\gamma$ is of the form}&&\label{CASE-D-3} 
    \tag{D-3}\end{flalign}  
   $$ \begin{pmatrix}
	0	& 			-1						  &0 \\e^{\imagunit \theta_1}
	 		& 0									  &   0\\
	0 	  	& 0 			  & e^{\imagunit \theta_2}
\end{pmatrix},\quad (\theta_1+\theta_2\equiv 2\pi). $$
  \item[] 
\begin{flalign}~\mbox{\it $\gamma$ is of the form}&&\label{CASE-E-1} 
    \tag{E-1}\end{flalign}   $$ \begin{pmatrix}
	e^{\imagunit \theta_1}	& 0									  &0 \\0&\gamma_{22}&\gamma_{23}
	 	  \\
	0 	  	&   			   \gamma_{32} & \gamma_{33}
\end{pmatrix},\quad \mbox{ where every }\gamma_{ij}\neq 0.$$
  \item[] 
\begin{flalign}~\mbox{\it $\gamma$ is of the form}&&\label{CASE-E-2} 
    \tag{E-2}\end{flalign}  
  $$ \begin{pmatrix}
	\gamma_{11}	& 0&			\gamma_{13}						  \\0&e^{\imagunit \theta_2}
	 		& 0									  \\
	  			    \gamma_{31} &0&\gamma_{33}
\end{pmatrix},\quad \mbox{ where every }\gamma_{ij}\neq 0.$$
    
  \item[] 
\begin{flalign}~\mbox{\it $\gamma$ is of the form}&&\label{CASE-E-3} 
    \tag{E-3}\end{flalign}
   $$ \begin{pmatrix}
	\gamma_{11}	& 		\gamma_{12}						  &0 \\\gamma_{21}
	 		& \gamma_{22}									  &   0\\
	0 	  	& 0 			  & e^{\imagunit \theta_3}
\end{pmatrix},\quad \mbox{ where every }\gamma_{ij}\neq 0. $$
\end{itemize}

By carefully analyzing Cases \eqref{CASE-D-1}–\eqref{CASE-E-3}, together with Cases \eqref{CASE-A}, \eqref{CASE-B}, and \eqref{CASE-C} introduced in Section \ref{Sec:genericRxirefined}, we refine $\Pxi$ into 9 pieces and $\Reducedxi$ into 16 pieces. For brevity, we omit the details of this refinement. The properties of the resulting 16 strata of $\Reducedxi$ are summarized in the following three tables.
 
       \begin{table}[H]\centering
 \caption{The refined strata of  $\Reducedxi  $  for  $\xi$ in the $\vec{c}$ direction (Part I)} \label{Table:strataveccxi1} 
       \begin{supertabular}{|p{0.9cm}|p{2.2cm}| p{1.2cm}|p{2.3cm}|p{2.4cm}|  p{2.6cm} | p{1.5cm}      |  }
 \hline
 &
 Stratum of   $\Reducedxi  $  $=$ $\SET{[(x,\gamma)]}$   &  Orbit type of $G_{x, \gamma}$   & Range of $[x]$ & Range of $\gamma$	   &  Stratum of   $\Pxi $ $=$ $\SET{[(x,-\gamma^{-1}x\gamma)]}$  & Orbit type of $G_{x,\gamma^{-1}x\gamma}$  \\  
  \hline  No.1 ~& 
  $(\Reducedxi)^6  $ &	$H(3)$ & $ (\Bxi)^2   $ & Case \eqref{CASE-C}     & $(\Pxi)^4  $  & $H(3)$ \\ 
  \hline  No.2 ~&
  $(\Reducedxi)^2_{(0)}$ & $H(3)$ &	$ (\Bxi)^1_{(0)}  $ & Case \eqref{CASE-B-1} or \eqref{CASE-B-2}    &  $(\Pxi)^1_{(0)}$ &  $H(2,1)$  \\
\hline  No.3 ~&
  $(\Reducedxi)^2_{(1)}$ & $H(3)$ &	$ (\Bxi)^1_{(1)}  \subset l_1$ & Case \eqref{CASE-B-5} or \eqref{CASE-B-6}    &  $(\Pxi)^1_{(1)}$ &  $H(2,1)$  \\
  \hline  No.4 ~& 
  $(\Reducedxi)^2_{(2)}$ & $H(3)$ &	$ (\Bxi)^1_{(2)}  \subset l_2$ & Case \eqref{CASE-B-3} or \eqref{CASE-B-4}    &  $(\Pxi)^1_{(2)}$ &  $H(2,1)$  \\  
  \hline  No.5 ~&
  $(\Reducedxi)^2_{(1')}$ & $H(3)$ &	$\mathcal{C}_{(1)}\subset (\Bxi)^2   $ & Case \eqref{CASE-B-5} or \eqref{CASE-B-6}    &  $(\Pxi)^1_{(1')}$ & $H(2,1)$   \\
   \hline  No.6 ~&
  $(\Reducedxi)^2_{(2')}$ & $H(3)$ &	$\mathcal{C}_{(2)}\subset (\Bxi)^2   $ & Case \eqref{CASE-B-3}  or \eqref{CASE-B-4}   &  $(\Pxi)^1_{(2')}$ & $H(2,1)$   \\
   \hline  
   \end{supertabular}
   \end{table}
   The strata numbered 1 $\sim$ 6 in the above table arise for the same reasons as in the case of generic \(\xi\). However, the corresponding ranges of \(\gamma\) differ from those reported in Table \ref{Table:stratagenericxi}.

The following table presents four special strata, denoted by Nos.7, 8, \(7^{+}\), and \(8^{+}\). In fact, the two rows  Nos.7 and 7$^+$ together replace No.7   in Table \ref{Table:stratagenericxi}, in the generic-\(\xi\) case. The roles of Nos.8 and 8$^+$ are similar.
    

   \begin{table}[H]\centering
 \caption{The refined strata of  $\Reducedxi  $  for  $\xi$ in the $\vec{c}$ direction (Part II)} \label{Table:strataveccxi2} 
       \begin{supertabular}{|p{0.9cm}|p{2.2cm}| p{1.2cm}|p{2.3cm}|p{2.4cm}|  p{2.6cm} | p{1.5cm}      |  }\hline
         No.7 ~&
   $(\Reducedxi)^0_{(1)}$ & $H(3)$	& $(\Bxi)^0_{(1)}  $ & Case \eqref{CASE-A-1} or \eqref{CASE-A-2}   &  $(\Pxi)^0_{(1)}$ &  $\Cartan$   \\
   \hline  No.8 ~&
   $(\Reducedxi)^0_{(2)}$ & $H(3)$	& $(\Bxi)^0_{(2)}  $ & Case \eqref{CASE-A-1} or \eqref{CASE-A-2}   &  $(\Pxi)^0_{(2)}$ &  $\Cartan$   \\ 
\hline  
       No.7$^+$ ~&
  $(\Reducedxi)^2_{(1+)}$ & $H(3)$	& $(\Bxi)^0_{(1)}  $ & Case   \eqref{CASE-B-1}, \eqref{CASE-B-2},    \eqref{CASE-B-3}, or \eqref{CASE-B-6}   &  $(\Pxi)^0_{(1)}$ &  $\Cartan$   \\
 \hline   No.8$^+$ ~&
  $(\Reducedxi)^2_{(2+)}$ & $H(3)$	& $(\Bxi)^0_{(2)}  $ & Case   \eqref{CASE-B-1}, \eqref{CASE-B-2}, \eqref{CASE-B-4} or \eqref{CASE-B-5}    &  $(\Pxi)^0_{(2)}$ &  $\Cartan$   \\ 
   \hline  
   \end{supertabular}
   \end{table}
   The  strata $(\Reducedxi)^0_{(1)}$ and $(\Reducedxi)^0_{(2)}$ are, respectively, the limit points of $(\Reducedxi)^2_{(1+)}$ and $(\Reducedxi)^2_{(2+)}$, and  many other strata among Nos.1 $\sim$ 6 (see Diagram \ref{Diagram:SU3arrowsabcrefined}). 
   
   Notably, both corresponding strata $(\Reducedxi)^2_{(1+)}$ and $(\Reducedxi)^2_{(2+)}$ have dimension \(2\)  and \emph{trivial} presymplectic structure.

   The remaining  strata are described as in the following table.
   \begin{table}[H]\centering
 \caption{The refined strata of  $\Reducedxi  $  for  $\xi$ in the $\vec{c}$ direction (Part III) } \label{Table:strataveccxi3} 
       \begin{supertabular}{|p{0.9cm}|p{2.2cm}| p{1.2cm}|p{2.3cm}|p{2.4cm}|  p{2.6cm} | p{1.5cm}      |  }\hline
   No.9 ~& $(\Reducedxi)^4 $
& $H(2,1)$ & $(\Bxi)^2   $ & Case  \eqref{CASE-E-1},\eqref{CASE-E-2} or \eqref{CASE-E-3} & $(\Pxi)^2 $ & $H(2,1)$\\
   \hline  No.10 ~& $(\Reducedxi)^3_{(0)}$
 &$H(2,1)$& $ (\Bxi)^1_{(0)}  $ & Case \eqref{CASE-E-1},\eqref{CASE-E-2} or \eqref{CASE-E-3} & $(\Pxi)^1_{(0)}$ & $H(2,1)$\\
   \hline  No.11 ~& $(\Reducedxi)^3_{(1)}$
 &$H(2,1)$&$ (\Bxi)^1_{(1)}  \subset l_1$& Case \eqref{CASE-E-2} or \eqref{CASE-E-3} & $(\Pxi)^1_{(1)}$ & $H(2,1)$\\
   \hline   
   No.12 ~& $(\Reducedxi)^3_{(2)}$
 &$H(2,1)$& $ (\Bxi)^1_{(2)}  \subset l_2$ & Case \eqref{CASE-E-1} or \eqref{CASE-E-2}& $(\Pxi)^1_{(2)}$ & $H(2,1)$ \\
   \hline  
   No.13 ~& $(\Reducedxi)^1_{(1)}$
 &$H(2,1)$&$(\Bxi)^0_{(1)}  $& Case \eqref{CASE-D-2} or \eqref{CASE-D-3} & $(\Pxi)^0_{(1)}$ & $\Cartan$ \\
   \hline  No.14 ~& $(\Reducedxi)^1_{(2)}$
 &$H(2,1)$& $(\Bxi)^0_{(2)}  $ & Case \eqref{CASE-D-1} or \eqref{CASE-D-2} & $(\Pxi)^0_{(2)}$ & $\Cartan$\\
   \hline   
    \end{supertabular}
   \end{table}

Compared with Table \ref{Table:stratagenericxi} for generic $\xi$,  more strata of $\Reducedxi  $ appear for $\xi$ in the $\vec{c}$ direction (Nos.7$^+$, 8$^+$, and 9 $\sim$ 14). Note also that  $\Pxi $ has a stratum $(\Pxi)^2 $ (appeared in row No.9) which does not show up for generic $\xi$.    Moreover, five strata of $\Reducedxi  $ (Nos.10 $\sim$ 14) are   odd-dimensional, and they are  indeed presymplectic.
As for the boundary relations, we can draw the following diagram connecting  strata Nos.9 $\sim$ 14:

\begin{equation}
	\label{Diagram:SU3arrowscdirectionrefined}
	\begin{tikzcd}
		(\Reducedxi)^4  \arrow[ddd,bend right=55,dashed]\arrow[dd,bend right=55,dashed]\arrow[d,dashed] \arrow[rr] && (\Pxi)^2   \arrow[ddd,bend right=55,dashed]\arrow[dd,bend right=55,dashed]\arrow[d,dashed] \arrow[rr]   && (\Bxi)^2    \arrow[d,dashed] \arrow[dd,bend left=55,dashed]\arrow[ddd,bend left=55,dashed] \\
		(\Reducedxi)^3_{(0)}  
 \arrow[ddd,bend right=55,dashed] 
\arrow[dddd,bend right=55,dashed]\arrow[rr ] && (\Pxi)^1_{(0)} \arrow[ddd,bend right=55,dashed] 
\arrow[dddd,bend right=55,dashed] \arrow[rr ]&& (\Bxi)^1_{(0)}   \arrow[ddd,bend left=55,dashed]\arrow[dddd,bend left=55,dashed]\\  
		(\Reducedxi)^3_{(1)}  
 \arrow[dd,bend right=55,dashed] \arrow[rr ] && (\Pxi)^1_{(1)} \arrow[dd,bend right=55,dashed] \arrow[rr ]  && (\Bxi)^1_{(1)}   \arrow[dd,bend left=55,dashed]\\  
		(\Reducedxi)^3_{(2)}  
 \arrow[dd,bend right=55,dashed] 
 \arrow[rr ] && (\Pxi)^1_{(2)} \arrow[dd,bend right=55,dashed] \arrow[rr ]  && (\Bxi)^1_{(2)}   \arrow[dd,bend left=55,dashed]\\    
		(\Reducedxi)^1_{(1)}   \arrow[rr ] && (\Pxi)^0_{(1)}  \arrow[rr ]  && (\Bxi)^0_{(1)}    \\
		(\Reducedxi)^1_{(2)}  \arrow[rr ] && (\Pxi)^0_{(2)}  \arrow[rr ]  && (\Bxi)^0_{(2)}     
	\end{tikzcd}
\end{equation}

Of course there are other boundary relations between the strata Nos.1 $\sim$ 8, 7$^+$, 8$^+$, and   9 $\sim$ 14. To see the whole picture of relations between $\Reducedxi  $, $\Pxi $, and $\Bxi $, one needs to  modify Diagram \ref{Diagram:SU3arrowsabcrefined}, and combine it with Diagram \ref{Diagram:SU3arrowscdirectionrefined}.

\appendix

\section{Review of stratified spaces}\label{Apd:stratifiedspace}

To establish precise terminology, we adopt the definition of a stratified space as presented in \cite{Pflaumbook, SLannals,PreSLstratifiedspace1,PreSLstratifiedspace2}:

\begin{definition}
	A \textit{stratified space} $\mathcal{M}=(\mathcal{M},{A})$ is a Hausdorff and paracompact topological space $\mathcal{M}$ equipped with a decomposition, which is a collection of disjoint, locally finite, locally closed, and connected smooth manifolds $\mathcal{M}_\alpha\subset \mathcal{M}$, termed strata, indexed by a finite, partially ordered set ${A}$ (with the partial order denoted by $\partialless $). The following conditions must hold:
	\begin{itemize}
		\item The space $\mathcal{M}$ is the disjoint union of its strata: $\mathcal{M}=\bigsqcup_{\alpha\in{A}} \mathcal{M}_{\alpha}$.
		\item (The Frontier Condition) For all  strata $\mathcal{M}_{\alpha}$ and $\mathcal{M}_{\beta}$, the condition $\mathcal{M}_{\alpha}\cap \overline{\mathcal{M}}_{\beta}\neq \emptyset$ is equivalent to   $\mathcal{M}_{\alpha}\subseteq \overline{\mathcal{M}}_{\beta}$, and to  $\alpha\partialless \beta$.
	\end{itemize}
	Furthermore, the decomposition is required to satisfy Whitney's  (B)-regularity condition, which ensures that $\mathcal{M}$ possesses the property of conical local triviality (see \cite{Pflaumpaper,Pflaumbook}).
\end{definition}

 The stratum of highest dimension is called the \textit{maximum stratum}. If this stratum is unique, it is denoted by $\mathcal{M}_{\maxpart}$ or $(\mathcal{M})^m$, where $m = \dim \mathcal{M}_{\maxpart}$, and the dimension of the stratified space $\mathcal{M}$ is defined as $m$. Analogously, a unique stratum of dimension $k$ within $\mathcal{M} $ is denoted by $(\mathcal{M})^k$. In cases where the dimension $k$ stratum is not unique, it is instead denoted by $(\mathcal{M})^k_i$, where the index $i$ is employed to distinguish between multiple strata of the same dimension.

\begin{example}\label{Example:demodashlines} The half-disc $\mathcal{M}  \subset \mathbb{R}^2$, defined by the set of points $(x, y)$ satisfying $x^2 + y^2 \leqslant 1$ and $y \geqslant 0$, can be stratified into several parts. This stratification includes the interior, denoted as $(\mathcal{M})^2$; the upper semi-circular boundary, denoted as $(\mathcal{M})^1_{\mathrm{upp}}$; the lower straight edge, denoted as $(\mathcal{M})^1_{\lowpart}$; and finally, the two endpoints $(\mathcal{M})^0_{-} = \{(-1, 0)\}$ and $(\mathcal{M}^{0 })_{ +} = \{(1, 0)\}$. To visually represent such a stratified space, we adopt a diagrammatic style where a dashed arrow  connecting an upper stratum to a lower stratum indicates that the latter constitutes a boundary component of the former:
	
	
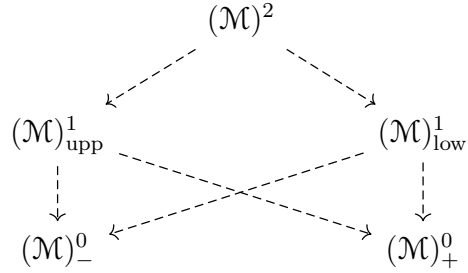
\begin{equation*}
		\begin{tikzcd}
			&(\mathcal{M})^2 \arrow[ld,dashed]\arrow[rd,dashed]&   \\
			(\mathcal{M})^1_{\mathrm{upp}} \arrow[d,dashed]\arrow[rrd,dashed]&&  (\mathcal{M})^1_{\lowpart} \arrow[d,dashed]\arrow[lld,dashed]       \\ (\mathcal{M})^0_{-} && (\mathcal{M})^0_{+}   
		\end{tikzcd}
	\end{equation*}

In fact, there should be two dashed arrows starting from $(\mathcal{M})^{2} $ and pointing to $(\mathcal{M})^{0}_{-}$ and $(\mathcal{M})^{0 }_{ +}$, respectively, because $(\mathcal{M})^{0}_{ -}$ and $(\mathcal{M})^{0 }_{ +}$ are also boundaries of $(\mathcal{M})^{2} $. However, since these boundary relations are implied by compositions of other existing arrows via established paths within the diagram,  {the explicit drawings of the arrows  from $(\mathcal{M})^{2} $   to $(\mathcal{M})^{0}_{ -}$ and $(\mathcal{M})^{0 }_{ +}$ are  \emph{omitted}}.

\end{example}

 Pflaum also introduced   smooth structures on stratified spaces using maximal atlases \cite{Pflaumpaper,Pflaumbook}.   Note that Sjamaar and Lerman also developed a   theory of smooth structures on   stratified spaces \cite{SLannals}, which has been further elaborated upon by numerous researchers \cite{BatesLerman1997,Huebschmann2004}. In the current paper, we adhere to Pflaum's treatment of smooth structures.

\begin{definition}
	A singular chart on a stratified space $\mathcal{M}$ is defined as a continuous map $x: U\to \mathbb{R}^n$, where $U$ is an open subset of $\mathcal{M}$ and $n$ is a positive integer. Note that $n$ is not intrinsically determined by $\mathcal{M}$. This map $x$ must satisfy the following conditions: first, its image $x(U)$ is locally closed in $\mathbb{R}^n$, and $x$ is a homeomorphism onto $x(U)$. Second, for each stratum $\mathcal{M}_{\alpha}$ of $\mathcal{M}$, the image $x(\mathcal{M}_{\alpha}\cap U)$ is a smooth submanifold of $\mathbb{R}^n$, and the restriction $x|_{\mathcal{M}_{\alpha}\cap U}$ is a smooth diffeomorphism onto $x(\mathcal{M}_{\alpha}\cap U)$.
	
	Two singular charts $x: U\to \mathbb{R}^n$ and $y: V\to \mathbb{R}^{m}$ on $\mathcal{M}$ are called compatible if, for every point $p$ in the intersection $U\cap V$, there exists an open neighborhood $W$ of $p$ contained in $U\cap V$, an integer $k$ such that $k\geqslant \max(n,m)$, and a transition map $H: O\to O'$ between open sets of $\mathbb{R}^{k}$ such that $y|_W = H\circ x|_W$.
	
\end{definition}

Our definition of a stratified space corresponds to what is known as a Whitney stratified space as described in \cite{Pflaumpaper, Pflaumbook}. The definition ensures the existence of a smooth singular atlas on any stratified space $\mathcal{M}$, which is defined as a family of singular charts $x_j:U_j\to \mathbb{R}^{n_j}$, where $j\in J$. These charts satisfy the condition that their domains cover the entire stratified space, i.e., $\cup_{j\in J}U_j=\mathcal{M}$, and that every pair of charts is compatible.
\begin{definition}
	Let $U\subset \mathcal{M}$ be an open set. We define the space of smooth functions on $U$, denoted by $\Cinf_\mathcal{M}(U)$, as the commutative algebra of all continuous functions $f: U\to \mathbb{R}$ satisfying the following condition: for every singular chart $x: V\to \mathbb{R}^n$, there exists an open subset $W\subset V\cap U $ and a smooth function $g: \mathbb{R}^n\to \mathbb{R}$ such that   $g\circ x|_W=f|_W$. The space of smooth functions on the entire stratified space $\mathcal{M}$ is denoted by $\Cinf{(\mathcal{M})}$.

\end{definition}

\begin{definition}\label{Def:smoothmapstratified}
A     map $\phi: (\mathcal{M}, A) \to (\mathcal{M}', A')$ between stratified spaces is defined as \textbf{smooth} if it is  continuous and its smoothness is verifiable through local charts of $\mathcal{M}$ and $\mathcal{M}'$. In other words, for any open set $U \subset \mathcal{M}'$ and any smooth function $f \in \Cinf_{\mathcal{M}'}(U)$, the pullback function $\phi^*f$ is also smooth, i.e., $\phi^*f \in \Cinf_{\mathcal{M}}(\phi^{-1}(U))$.
	
 A   smooth map   $\phi:  (\mathcal{M},{A})\to (\mathcal{M}',A')$  is called   a \textbf{morphism} of stratified spaces if, for every stratum $\mathcal{M}_{\alpha}$ ($\alpha \in A$), its image under $\phi$ is entirely contained within a stratum $\mathcal{M}'_{\alpha'}$ ($\alpha' \in A'$).
\end{definition}

\section{Refined stratification by taking inverse images} \label{Apd:refinebyinverseimage} 
\begin{definition}
  Let $\mathcal{M}=(\mathcal{M},{A})$ be a stratified space. A \emph{refinement} of $(\mathcal{M},{A})$ is another stratification of the same topological space, given by
  $$\mathcal{M}=\sqcup_{\alpha'\in A'} \mathcal{M}_{\alpha'},$$
  such that the identity map on $\mathcal{M}$ is a morphism of stratified spaces from $(\mathcal{M},A')$ to $(\mathcal{M},A)$.
\end{definition}
In practice, refinements of $(\mathcal{M},{A})$ are often realized by further stratifying each stratum $\mathcal{M}_{\alpha}$ ($\alpha\in A$) as $\mathcal{M}_{\alpha}=\sqcup_{\beta\in B_\alpha} \mathcal{M}_{\alpha,\beta}$. In this setting, the refined decomposition of $\mathcal{M}$ is given by $\sqcup_{(\alpha,\beta)\in A' } \mathcal{M}_{\alpha,\beta}$, where the new index set $A'$ is defined as $A' =\{(\alpha,\beta):~\alpha\in A,\beta\in B_{\alpha}\}$. The partial order on $A'$ is then defined by
     $$(\alpha,\beta)\partialless (\alpha',\beta') \Leftrightarrow \mathcal{M}_{\alpha,\beta}\cap \overline{\mathcal{M}}_{\alpha',\beta'}\neq \emptyset \quad (\Leftrightarrow \mathcal{M}_{\alpha,\beta}\subseteq \overline{\mathcal{M}}_{\alpha',\beta'}).
     $$

Let \(\mathcal{M}=(\mathcal{M},A)\) and \(\mathcal{N}=(\mathcal{N},B)\) be stratified spaces, and let
\[
f:(\mathcal{M},A)\to(\mathcal{N},B)
\]
be a smooth map. In general, \(f\) need not map each stratum \(\mathcal{M}_{\alpha}\) into a single stratum \(\mathcal{N}_{\beta}\). To ensure that \(f\) has this property, and hence is a morphism of stratified spaces, we refine the stratification of \(\mathcal{M}\) by taking inverse images under \(f\). More precisely, for each \(\alpha\in A\), decompose \(\mathcal{M}_{\alpha}\) as
\[
\mathcal{M}_{\alpha}
=\sqcup_{\beta\in B}\mathcal{M}_{\alpha,\beta},
\qquad
\mathcal{M}_{\alpha,\beta}
:=\mathcal{M}_{\alpha}\cap f^{-1}(\mathcal{N}_{\beta}),
\]
so that
\begin{equation}\label{Eqt:MdecomposedtoAtimesB}
   \mathcal{M}=\sqcup_{(\alpha,\beta)\in A\times B}\mathcal{M}_{\alpha,\beta}.
\end{equation}
If \(\mathcal{M}_{\alpha,\beta}\) is empty, it is omitted, and the corresponding index \((\alpha,\beta)\) is removed from \(A\times B\). \emph{This convention is used throughout the paper.}

Although one might expect \eqref{Eqt:MdecomposedtoAtimesB} to define a new stratification of $\mathcal{M}$, this is not true in general. The following example illustrates the issue.

\begin{example}
   \label{Example:Mdecomposesticky}
   Let
   \[
   \mathcal{M}=\SET{(x,y)\in \mathbb{R}^2\mid x\geqslant 0}
   =(\mathcal{M})^1\sqcup (\mathcal{M})^2,
   \]
   where $(\mathcal{M})^1=\SET{x=0}$ is the boundary line and $(\mathcal{M})^2=\SET{  x>0}$ is the interior. Let $\mathcal{N}=\mathbb{R}^{\geqslant 0}$ be stratified into its endpoint $(\mathcal{N})^0=\{0\}$ and its interior $(\mathcal{N})^1=\mathbb{R}^{>0}$. Consider the smooth map
   \[
   f:\mathcal{M}\to\mathcal{N},\qquad (x,y)\mapsto x(x-1)^2y^2.
   \]
   Taking the inverse images of the strata of $\mathcal{N}$ decomposes $\mathcal{M}$ into the following three subsets:
   \[
   (\mathcal{M})^1_{\mathrm{v}}
   =(\mathcal{M})^1\cap f^{-1}((\mathcal{N})^0)
   =(\mathcal{M})^1=\SET{x=0},
   \]
   \[
   (\mathcal{M})^1_{\mathrm{x}}
   =(\mathcal{M})^2\cap f^{-1}((\mathcal{N})^0)
   =\SET{  x=1 \mbox{ or } y=0},
   \]
   and
   \[
   (\mathcal{M})^2_{\mathrm{o}}
   =(\mathcal{M})^2\cap f^{-1}((\mathcal{N})^1)
   =\SET{x\neq 0,x\neq 1,y\neq 0}.
   \]
   These three subsets do \emph{not} form a stratification of $\mathcal{M}$ for three reasons. First, $(\mathcal{M})^2_{\mathrm{o}}$ is disconnected, having four connected components. Second, $(\mathcal{M})^1_{\mathrm{x}}$ is not a smooth manifold. Third, although the closure $\overline{(\mathcal{M})^1_{\mathrm{x}}}$ meets $(\mathcal{M})^1_{\mathrm{v}}$, the latter is not contained in $\overline{(\mathcal{M})^1_{\mathrm{x}}}$ (violating the Frontier Condition). Hence, a further refinement is required.

 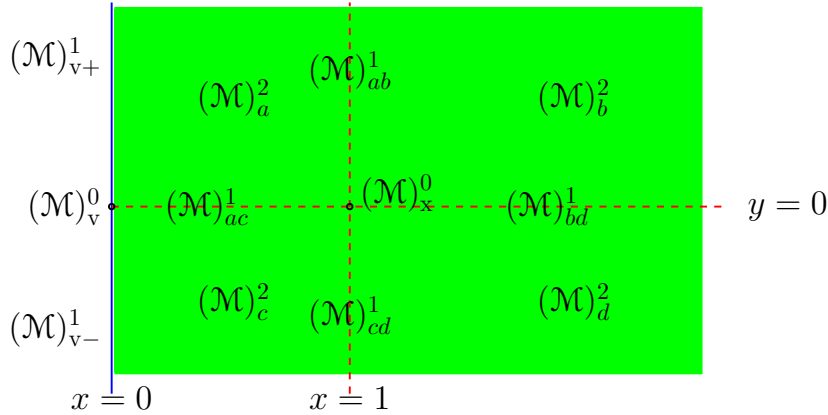
\begin{figure}[htbp]   \centering  
  \scalebox{0.9}{ 
      \begin{tikzpicture}[>=stealth, line width=0.8pt, font=\large]
            \coordinate (O) at (0,0);
\filldraw[green,line width=2] (0.04,0.07) -- (8.65, 0.07)     -- (8.65,5.4) -- (0.04,5.4) ;
         \draw[blue] (0,-0.25)--(0,5.5);
       \draw[red,dashed] (0,2.5)--(9,2.5);
       \draw[red,dashed] (3.5,-0.25)--(3.5, 5.5);
\draw[black] (3.5,2.5) circle (0.04) ; 
\draw[black] (0,2.5) circle (0.04) ;
\draw[black] (1.1,1.1) node[right] {$(\mathcal{M})^2_c$};
\draw[black] (1.1,4.1) node[right] {$(\mathcal{M})^2_a$};
\draw[black] (6.1,1.1) node[right] {$(\mathcal{M})^2_d$};
\draw[black] (6.1,4.1) node[right] {$(\mathcal{M})^2_b$};

\draw[black] (1.4,2.5) node {$(\mathcal{M})^1_{ac}$};

\draw[black] (6.4,2.5) node {$(\mathcal{M})^1_{bd}$};
\draw[black] (3.5,0.9) node {$(\mathcal{M})^1_{cd}$};
\draw[black] (3.5,4.5) node {$(\mathcal{M})^1_{ab}$};

\draw[black] (0,2.5) node[left] {$(\mathcal{M})^0_{\mathrm{v}} $};
\draw[black] (3.5,2.7) node[right] {$(\mathcal{M})^0_{\mathrm{x}} $};
\draw[black] (0,4.7) node[left] {$(\mathcal{M})^1_{\mathrm{v}+} $};
\draw[black] (0,0.7) node[left] {$(\mathcal{M})^1_{\mathrm{v}-} $};

\draw[black] (0,0) node[below] {$x=0$}; 
\draw[black] (3.5,0) node[below] {$x=1$};
\draw[black] (9.2,2.5) node[right] {$y=0$};
        \end{tikzpicture}
        }
  \caption{Refined stratification of the half plane $\mathcal{M}=\SET{x\geqslant 0}$} \label{Fig:refinedhalfplane}~ 
\end{figure}
 We separate the four connected components of $(\mathcal{M})^2_{\mathrm{o}}$ and denote them by $(\mathcal{M})^2_{a}$, $(\mathcal{M})^2_{b}$, $(\mathcal{M})^2_{c}$, and $(\mathcal{M})^2_{d}$, as shown in Figure \ref{Fig:refinedhalfplane}. Next, we decompose $(\mathcal{M})^1_{\mathrm{v}}$ into three pieces: $(\mathcal{M})^1_{\mathrm{v+}}$ (the upper line), $(\mathcal{M})^1_{\mathrm{v-}}$ (the lower line), and $(\mathcal{M})^0_{\mathrm{v}}$ (the origin). Similarly, $(\mathcal{M})^1_{\mathrm{x}}$ is decomposed into four one-dimensional pieces, $(\mathcal{M})^1_{ab}$, $(\mathcal{M})^1_{bd}$, $(\mathcal{M})^1_{cd}$, and $(\mathcal{M})^1_{ac}$, together with the zero-dimensional stratum $(\mathcal{M})^0_{\mathrm{x}}$, as illustrated in Figure \ref{Fig:refinedhalfplane}. These twelve pieces collectively constitute a genuine stratification of $\mathcal{M}$. With respect to this refined stratification, the map $f:\mathcal{M}\to\mathcal{N}$ is a morphism of stratified spaces.
 
  \end{example}

This example indicates that the decomposition \eqref{Eqt:MdecomposedtoAtimesB} generally requires three additional refinement procedures to define a stratification of $\mathcal{M}$ appropriately: separating the connected components of each $\mathcal{M}_{\alpha,\beta}$, decomposing each connected component into stratified subsets, and incorporating the limit points of each candidate stratum. The principle is to modify \eqref{Eqt:MdecomposedtoAtimesB}   minimally by introducing suitable subsets. Throughout this paper, a \textbf{refined stratification} of $\mathcal{M}$ via the inverse images of $f$ refers to a stratification obtained after these procedures have been performed, such that $f$ becomes a morphism of stratified spaces. For notational simplicity, however, we continue to refer to \eqref{Eqt:MdecomposedtoAtimesB} as the refined stratification and retain $A\times B$ as its index set. Accordingly, a refined stratum denoted by $\mathcal{M}_{\alpha,\beta}$ may itself represent a subset that requires further decomposition.
 
\section{Surjectivity of the moment   map}\label{Sec:JTstarGsurjectiveSUn}

Let us consider the setting where an integer \( n \geqslant 2 \) is fixed, along with the Lie group \( G = \SU(n) \), and the Lie algebra \( \galgebra = \su(n) \).  
We need the following fact,  well known   among experts. 
For the sake of completeness, we provide a proof.  \begin{proposition}\label{Prop:anyxitoO0} Any nonzero matrix in $\su(n)$ can be conjugated via some element in $\SU(n)$ to a matrix within $\su(n)$ such that all of its diagonal elements are zero and all elements on its   superdiagonal are equal to a single nonzero complex number. \end{proposition}
\begin{proof}Suppose that $\xi\in \su(n)$ is given and $\xi\neq 0$. It can be diagonalized by conjugation. So we simply consider  \(\xi = \imagunit \cdot \diag(\xi_1, \ldots, \xi_n)\), where \(\xi_i\) are real numbers satisfying \(\xi_1 + \cdots + \xi_n = 0\).   	
	Without loss of generality, we assume the ordering \(\xi_1 \leqslant \xi_2 \leqslant \cdots \leqslant \xi_n\). Let \( F \) denote the Fourier matrix of size \( n \), defined by the elements \( F_{ij} = \zeta^{(i-1)(j-1)} \), where \(\zeta = \rme^{\imagunit \theta} = \cos \theta + \imagunit \sin \theta\) and \(\theta = \frac{2\pi}{n}\), representing the primitive \( n \)-th root of unity. It is a well-established fact that \( g = \frac{1}{\sqrt{n}} F \) is a special unitary matrix. We proceed to show that the matrix \[ z := g^{-1} \xi g = \frac{1}{n} \overline{F}^T \xi F \in \Oxi \] possesses the required properties: specifically, \( z_{ii} = 0 \) and \( z_{i(i+1)} \neq 0 \).

	To compute the desired results, we first check the expression: $$ (\overline{F}^T \xi F)_{ii} = \imagunit \sum_{k=1}^n \xi_k = 0. $$ Next, we evaluate: \[ (\overline{F}^T \xi F)_{i(i+1)} = \imagunit \sum_{k=1}^n \xi_k \zeta^{k-1} = \imagunit (\xi_1 + \xi_2 \zeta + \cdots + \xi_n \zeta^{n-1}). \] This expression can be further expanded as follows: \[ \begin{aligned} \imagunit &\left(\xi_1 + (\xi_2 + \xi_n)\cos \theta + (\xi_3 + \xi_{n-1})\cos 2\theta + \cdots + (\xi_{m+1} + \xi_{m+2})\cos m\theta\right) \\ &- \left((\xi_2 - \xi_n)\sin \theta + (\xi_3 - \xi_{n-1})\sin 2\theta + \cdots + (\xi_{m+1} - \xi_{m+2})\sin m\theta\right), \end{aligned} \] for $n = 1 + 2m$, and \[ \begin{aligned} \imagunit &\left(\xi_1 + \xi_{m+1} + (\xi_2 + \xi_n)\cos \theta + (\xi_3 + \xi_{n-1})\cos 2\theta + \cdots + (\xi_{m+1} + \xi_{m+3})\cos m\theta\right) \\ &- \left((\xi_2 - \xi_n)\sin \theta + (\xi_3 - \xi_{n-1})\sin 2\theta + \cdots + (\xi_{m+1} - \xi_{m+3})\sin m\theta\right), \end{aligned} \] for $n = 2 + 2m$. Given the assumption $\xi_2 \leqslant \xi_3 \leqslant \cdots \leqslant \xi_n$, and since all terms $\sin \theta, \ldots, \sin m\theta$ are positive, the real part of $(\overline{F}^T \xi F)_{i(i+1)}$ is non-negative. In the special case where $\xi_2 = \xi_3 = \cdots = \xi_n > 0$, it follows that $\xi_1 = -(n-1)\xi_2$. Consequently, $(\overline{F}^T \xi F)_{i(i+1)}$ simplifies to a purely imaginary number: $$ \imagunit \cdot \xi_2 \cdot (-(n-1) + \zeta + \zeta^2 + \cdots + \zeta^{n-1}) = -n\xi_2 \imagunit. $$ In all scenarios, it is evident that $(\overline{F}^T \xi F)_{i(i+1)}$ is nonzero.

\end{proof}

We can now prove the surjectivity of the moment map.

\begin{proposition}
	\label{Prop:SUnmomentmapsurjective}
	For the group $G=\SU(n)$ and the Lie algebra $\galgebra\cong \gstar=\su(n)$, the associated moment map $J_{T^*G}: \su(n)\times \SU(n)\to \su(n)$  given by 
	$J_{T^*G}(x, \gamma)=x- \Adjointaction_{\gamma^{-1}}(x) $ (see Equation \eqref{Eqt:JTstarGgeneral}) is surjective.
	
\end{proposition}

\begin{proof}  Given any $z \in \su(n)$ with zero diagonal entries, and any regular diagonal matrix $\gamma = \diag(\gamma_1, \dots, \gamma_n) \in \SU(n)$, we can find an $x \in \su(n)$ such that $J_{T^*G}(x, \gamma) = z$. Furthermore, according to Proposition \ref{Prop:anyxitoO0}, any nonzero $\xi \in \su(n)$ is conjugate to some such $z$. Therefore, a preimage of $\xi$ with respect to the moment map $J_{T^*G}$ can be found.
\end{proof}

\section{Weyl's inequalities and Horn's theorem}\label{appendixSec:Horn}

We hereby recall Weyl's inequalities and Horn’s theorem \cite{Horn}, which address the relationship between the eigenvalues of two Hermitian matrices    and those of their sum.  For the purposes of this study, we restrict our attention to the  $3$-dimensional case, formulated as follows:

\begin{theorem}
  \label{Thm:3X3Weyl}[Weyl's inequalities] \cite{Weylinequality1,Weylinequality2} Let 
  $\underline{x}$, $\underline{y}$ be  $3\times 3$ Hermitian matrices, and  let $\underline{z}=\underline{x}+\underline{y}$. Then   the 
   eigenvalues    $\mu_1\geqslant \mu_2\geqslant \mu_3$ of $\underline{x}$,   $\nu_1\geqslant \nu_2\geqslant \nu_3$ of $\underline{y}$, and $\lambda_1\geqslant \lambda_2\geqslant \lambda_3$ of $\underline{z}$ satisfy the following inequalities:
   \begin{align}\label{Eqt:Weyl1}
\mu_1 + \nu_1 &\geqslant\lambda_1, \\\label{Eqt:Weyl2}
\min(\mu_2 + \nu_1,\mu_1 + \nu_2) &\geqslant\lambda_2, \\\label{Eqt:Weyl6} 
\min(\mu_3 + \nu_1,\mu_2 + \nu_2,\mu_1 + \nu_3) &\geqslant\lambda_3.
\end{align}
\end{theorem} 
By considering the  matrices $-\underline{x}, -\underline{y},$ and $-\underline{z}$, we have another family of Weyl's inequalities:
 \begin{align}\label{Eqt:Weyl1moremore}
\mu_3 + \nu_3 &\leqslant\lambda_3, \\\label{Eqt:Weyl2moremore}
\max(\mu_2 + \nu_3,\mu_3 + \nu_2) &\leqslant\lambda_2, \\\label{Eqt:Weyl6moremore} 
\max(\mu_1 + \nu_3,\mu_2 + \nu_2,\mu_3 + \nu_1) &\leqslant\lambda_1.
\end{align}

 In 1962, A. Horn conjectured that the eigenspectrum of the sum of two Hermitian matrices could be fully characterized by a set of recursive inequalities derived from the eigenspectra of the summands. This conjecture was definitively established following the proof of the saturation conjecture by A. Knutson and T. Tao \cite{Tao}. Again,   we only need to recall the $3$-dimensional case:
 
\begin{theorem}\label{Thm:3X3Horn}[Horn] 
There exist $3\times 3$ Hermitian matrices 
  $\underline{x}$, $\underline{y}$, and $\underline{z}$ satisfying $\underline{z}=\underline{x}+\underline{y}$  
  if and only if the 
   eigenvalues    $\mu_1\geqslant \mu_2\geqslant \mu_3$ of $\underline{x}$,   $\nu_1\geqslant \nu_2\geqslant \nu_3$ of $\underline{y}$, and $\lambda_1\geqslant \lambda_2\geqslant \lambda_3$ of $\underline{z}$  satisfy    Weyl's inequalities \eqref{Eqt:Weyl1} -- \eqref{Eqt:Weyl6} and the  following relations:
  \begin{align*} 
\mu_1 + \mu_2 + \nu_1 + \nu_2 &\geqslant\lambda_1 + \lambda_2, \\
\min(\mu_1 + \mu_3 + \nu_1 + \nu_2,\mu_1 + \mu_2 + \nu_1 + \nu_3 ) &\geqslant\lambda_1 + \lambda_3, \\ 
\min(\mu_1 + \mu_3 + \nu_1 + \nu_3,\mu_2 + \mu_3 + \nu_1 + \nu_2,\mu_1 + \mu_2 + \nu_2 + \nu_3) &\geqslant\lambda_2 + \lambda_3,  
\\
\mu_1 + \mu_2 + \mu_3 + \nu_1 + \nu_2 + \nu_3 &= \lambda_1 + \lambda_2 + \lambda_3\,.
\end{align*}
     
\end{theorem}

Let us also   recall the $2$-dimensional case of Horn's theorem, which is relatively simple.
 
\begin{theorem}\label{Thm:2X2Horn}[Horn] 
There exist $2\times 2$ Hermitian matrices 
  $\underline{x}$, $\underline{y}$, and $\underline{z}$ satisfying $\underline{z}=\underline{x}+\underline{y}$  
  if and only if the 
   eigenvalues    $\mu_1\geqslant \mu_2 $ of $\underline{x}$,   $\nu_1\geqslant \nu_2 $ of $\underline{y}$, and $\lambda_1\geqslant \lambda_2 $ of $\underline{z}$  satisfy  the following Weyl inequalities   
   \begin{align*} 
\max(\mu_1 + \nu_2,     \mu_2 + \nu_1)   &\leqslant \lambda_1 \leqslant  \mu_1 + \nu_1  \,,\\
\mu_2+\nu_2 &\leqslant \lambda_2 \leqslant \min(\mu_1 + \nu_2,     \mu_2 + \nu_1),
\end{align*}
   
   and the linearity relation:
  \begin{align*} 
\mu_1 + \mu_2   + \nu_1 + \nu_2   &= \lambda_1 + \lambda_2  \,.
\end{align*}
     
\end{theorem}

\end{document}